\documentclass[a4paper,UKenglish,cleveref, autoref, thm-restate]{lipics-v2021}

\usepackage{commands-tam}
\usepackage{wrapfig}
\usepackage{float}

\usepackage[textsize=scriptsize]{todonotes}

\title{Universal Shape Replication Via Self-Assembly With Signal-Passing Tiles} 

\author{Andrew Alseth}{University of Arkansas, USA}{awalseth@uark.edu}{https://orcid.org/0000-0002-0055-0788}{This author's work was supported in part by NSF grant CAREER-1553166}

\author{Daniel Hader}{University of Arkansas, USA}{dhader@uark.edu}{}{This author's work was supported in part by NSF grant CAREER-1553166}

\author{Matthew J. Patitz}{University of Arkansas, USA}{patitz@uark.edu}{https://orcid.org/0000-0001-9287-4028}{This author's work was supported in part by NSF grant CAREER-1553166}

\authorrunning{A. Alseth, D. Hader, and M.\,J. Patitz}

\Copyright{Andrew Alseth, Daniel Hader, and Matthew J. Patitz}



\category{} 

\relatedversion{} 

\EventEditors{John Q. Open and Joan R. Access}
\EventNoEds{28}
\EventLongTitle{28th International Conference on DNA Computing and Molecular Programming}
\EventShortTitle{DNA28}
\EventAcronym{DNA}
\EventYear{2022}
\EventDate{August 8--12, 2022}
\EventLocation{Albuquerque, New Mexico, USA}
\EventLogo{}
\SeriesVolume{}
\ArticleNo{}

\usepackage[utf8]{inputenc}
\usepackage{subcaption}

\begin{document}

\maketitle

\begin{abstract}
In this paper, we investigate shape-assembling power of a tile-based model of self-assembly called the Signal-Passing Tile Assembly Model (STAM). In this model, the glues that bind tiles together can be turned on and off by the binding actions of other glues via ``signals''. Specifically, the problem we investigate is ``shape replication'' wherein, given a set of input assemblies of arbitrary shape, a system must construct an arbitrary number of assemblies with the same shapes and, with the exception of size-bounded junk assemblies that result from the process, no others. We provide the first fully universal shape replication result, namely a single tile set capable of performing shape replication on arbitrary sets of any 3-dimensional shapes without requiring any scaling or pre-encoded information in the input assemblies. Our result requires the input assemblies to be composed of signal-passing tiles whose glues can be deactivated to allow deconstruction of those assemblies, which we also prove is necessary by showing that there are shapes whose geometry cannot be replicated without deconstruction. Additionally, we modularize our construction to create systems capable of creating binary encodings of arbitrary shapes, and building arbitrary shapes from their encodings. Because the STAM is capable of universal computation, this then allows for arbitrary programs to be run within an STAM system, using the shape encodings as input, so that any computable transformation can be performed on the shapes.
\end{abstract}

\clearpage
\pagenumbering{arabic}

\section{Introduction}\label{sec:intro}

Artificial self-assembling systems are most often designed with the goal of building structures ``from scratch''. That is, they are designed so that they will start from a disorganized set of relatively simple components (often abstractly called \emph{tiles}) that autonomously combine to form more complex target structures. This process often begins from collections of only unbound, singleton tiles, or sometimes also includes so-called \emph{seed assemblies} which may be small (in relation to the target structure) ``pre-built'' assemblies that encode some information which \emph{seeds} the growth of larger assemblies. This growth occurs as additional tiles bind to those seed assemblies according to the rules of the system, allowing them to eventually grow into the desired structures. Examples have been shown in both experimental settings (e.g. \cite{evans2014crystals,drmaurdsa,ke2012three}), as well as in the mathematical domains of abstract models (e.g. \cite{SolWin07,RotWin00,AGKS05g,IUSA,2HAMIU}). However, in the subdomain of algorithmic self-assembly, in which systems are designed so that the tile additions implicitly follow the steps of pre-designed algorithms, other goals have also been pursued. These have included, for instance, performing computations (e.g. \cite{jCCSA,jSADS}), identifying input assemblies that match target shapes \cite{ShapeIdentAlgo}, replicating patterns on input assemblies \cite{SignalsReplication,SchulYurWinfEvolution}, and replicating  (the shapes of) input assemblies \cite{chalkUniversalShapeReplicators2017,jNegativeGluesShapes,RNaseSODA2010,SelfReplicationDNA,STAMshapes}. In this paper, we explore the latter, particularly the theoretical limits of systems within a mathematical model of self-assembling tiles to replicate shapes.

We use the term \emph{shape replication} to refer to the goal of designing self-assembling systems that take as input seed assemblies and which produce new assemblies that have the same shapes as those seed assemblies \cite{RNaseSODA2010}. In order for tile-based self-assembling systems to perform shape replication, dynamics beyond those of the original abstract Tile Assembly Model (aTAM), introduced by Winfree \cite{Winf98} and widely studied (e.g. \cite{SolWin07,RotWin00,IUSA,jCCSA,Versus,Temp1PathsSTOC,DirectedNotIU,jSSADST}), are required. In the aTAM, tiles attach to the seed assembly and the assemblies which grow from it, one tile at a tile, and tile attachments are irreversible. A generalization of the aTAM, the hierarchical assembly model known as the 2-Handed Assembly Model \cite{Versus,AGKS05g}, allows for the combination of pairs of arbitrarily large assemblies, but it too only allows irreversible attachments. However, for shape replication, it is fundamentally important that at least some tiles are able to bind to the input assemblies to gather information about their shapes which is then used to direct the formation of the output assemblies, since binding to an assembly is the only mechanism for interacting with it. These output assemblies eventually must not be connected to the input assemblies if they are to have the same shapes as the original input assemblies. This requires that at some point tile bindings can be broken. A number of theoretical models have been proposed with mechanisms for breaking tiles apart, for example: glues with repulsive forces \cite{SingleNegative,NegativeGluesShapes}, subsets of tiles which can be dissolved at given stages of assembly \cite{RNaseSODA2010,RNAPods}, tiles which can turn glues on and off \cite{jSignals,JonoskaSignals1} (a.k.a. \emph{signal-passing tiles}), and systems where the temperature can be increased to cause bonds to break \cite{AGKS05g,SummersTemp}. Within these models, previous results have shown the power of algorithmic self-assembling systems to perform shape replication. In \cite{chalkUniversalShapeReplicators2017}, they used glues with repulsive forces, and in \cite{RNaseSODA2010} they used the ability to dissolve away certain types of tiles at given stages during the self-assembly process, and each showed how to replicate a large class on two-dimensional shapes. In \cite{STAMshapes}, signal-passing tiles were shown to be capable of replicating arbitrary hole-free two-dimensional shapes if they are scaled up by a factor of 2.  The results of \cite{SelfReplicationDNA} deal with the replication of three-dimensional shapes, and will be further discussed below. 

The results of this paper are the first which provide for shape replication of all 3-dimensional shapes with no requirement for scaling those shapes. Additionally, although in \cite{SelfReplicationDNA} all three-dimensional shapes can be replicated at the small scale factor of 2, there it is necessary for the input assemblies to have relatively complex information embedded within them (in the form of Hamiltonian paths through all of their points being encoded by their glues). In our results, the input assemblies require no such embedded information. Furthermore, the model used in \cite{SelfReplicationDNA} is more complex, allowing not only for hierarchical assembly and signal-passing tiles, but also for tiles of differing shapes, and glue bindings that are flexible and thus allow for assemblies to reconfigure by folding. For the results of this paper, we have not only limited the dynamics to those of the Signal-Passing Tile Assembly Model (STAM), but have even placed an additional restriction on the model. Rather than assigning fixed orientations to tiles, in the model we use and call the STAM$^R$ (i.e. the ``STAM with rotation'') tiles and assemblies are allowed to rotate. This allows us to consider an even more general, and difficult, version of the shape replication problem. Namely, the input assemblies in our constructions have glues of a single generic type covering their entire exteriors, and there is no distinction between a north-facing glue and an east-facing glue, for instance, as there is in the standard STAM. This makes several aspects of working with such generic input assemblies more difficult, but it is notable that our constructions need only trivial, simplifying modifications to work in the standard STAM and that our positive results thus also hold for the STAM. We show that there is a ``universal shape replicator'' which is a tileset in the STAM$^R$ that can be used in conjunction with any set of generic input assemblies and will cause assemblies of every shape in the input set to be simultaneously produced in parallel. This is the first truly universal shape replicator for two or three dimensional shapes\footnote{Note that while replicating two-dimensional shapes, which consist of points in a single plane, our construction will utilize three dimensions.}. Furthermore, we break our construction into two major components, a ``universal encoder'' and a ``universal decoder'' (see Figure \ref{fig:intro-example} for a depiction). The universal encoder is capable of taking generic input assemblies and creating assemblies that expose binary sequences that encode those shapes, and the universal decoder is capable of taking assemblies exposing those encodings and creating assemblies of the encoded shapes. Due to the Turing universality of this model, this also allows for the full range of all possible computational transformations to occur between the encoding and decoding, and thus enables the generation of any transformations of the shapes of the input assemblies, such as creating scaled versions or complementary shapes.

\begin{figure}
    \centering
    \includegraphics[width=4.5in]{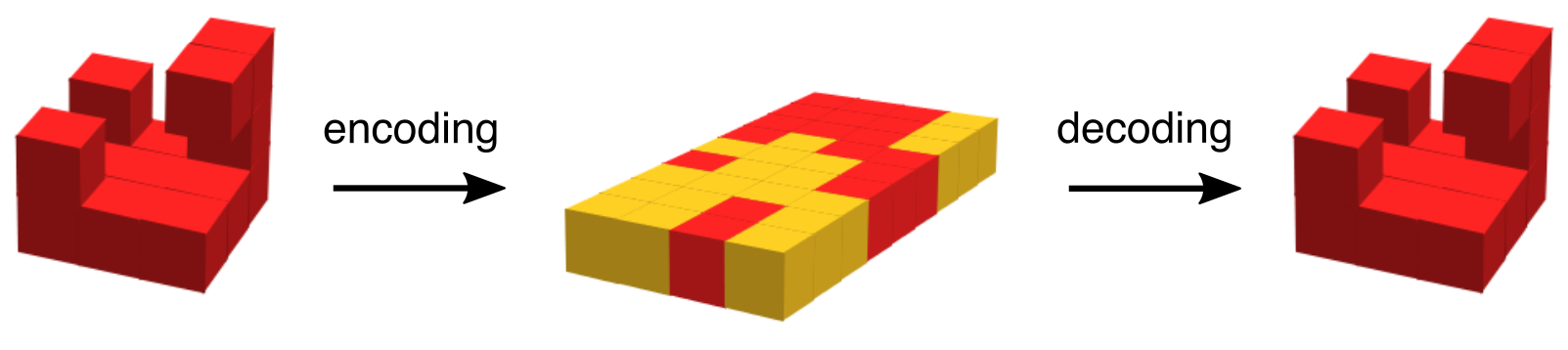}
    \caption{Schematic depiction of shape replication: (Left) An input assembly, (Middle) The assembly resulting from the encoding process which deconstructs the input assembly and encodes its shape, (Right) The assembly created by the decoding process, which uses the encoding as its input.}
    \label{fig:intro-example}
    \vspace{-10pt}
\end{figure}

In order for our universal shape replication construction to operate, the input assemblies must be created from signal-passing tiles which are capable of turning off their glues and dissociating from the assemblies. This allows for the assemblies to be ``deconstructed'', and we prove that this is necessary in order to replicate arbitrary shapes, specifically those which have enclosed or narrow, curved cavities, and this is intuitively clear since otherwise there would be no way to determine which locations in the interior of an input shape are included in the shape, and which are part of an enclosed void. Our proof that it is also impossible to replicate shapes with curved, but not enclosed, cavities further exhibits the additional difficulty of working within the STAM$^R$ model which allows tile rotations.

While our universal shape encoder, decoder, and replicator achieve the full goal of the line of research into shape replication, and also provide the ability to augment shape-building with arbitrary computational transformations, we note that the results are highly theoretical and serve more generally as an exploration of the theoretical limits of self-assembling systems. The tilesets are relatively large and require tiles with large numbers of signals, and although the input assemblies are not required to have complex information embedded within them, a trade-off that occurs compared with the results of \cite{SelfReplicationDNA} is that our constructions make use of a large amount of ``fuel''. That is, a large number of tiles are used during various phases but they are only temporary and aren't contained within the target assemblies and thus are ``consumed'' by the construction process. Despite the complexity of these theoretical constructions, we think that several modules and techniques developed may be of future use within other constructions (e.g. our ``leader election'' procedure which is guaranteed to uniquely select a single corner of an input assembly's bounding prism, to serve as a staring location for our encoding procedure within a constant number of assembly steps despite the lack of directional information provided by such an assembly), and also that these results may lead the way to similarly powerful but less complex constructions that may eventually achieve a level of being physically plausible to construct.

This paper is organized as follows. In Section \ref{sec:defns} we provide definitions of the STAM$^R$ and other terminology used throughout the paper, plus a series of subconstructions that appear throughout the main constructions. In Section \ref{sec:3d-rep} we state our main theorem and supporting lemmas, and present the constructions that prove them. In Section \ref{sec:STAM} we show that the constructions can be easily adapted to also work in the standard STAM. In Section \ref{sec:extensions} we briefly describe some of the computational transformations that could be used to augment our constructions, and in Section \ref{sec:imposs} we prove deconstruction is necessary for shape replication of certain classes of shapes.

\section{Definitions}\label{sec:defns}

In this section we provide definitions of the model used, and also for several of the terms and subconstructions used throughout the paper.

\subsection{Definition of the STAM$^R$ model}

    
    
    
    

Here we provide a definition of the model used in this paper, called the STAM$^R$ (i.e. the ``STAM with rotation''), which is based upon the 3D Signal-passing Tile Assembly Model (STAM) \cite{jSignals3D}. The STAM is itself based upon the 2-Handed Assembly Model (2HAM) \cite{AGKS05g,DDFIRSS07}, also referred to as the ``Hierarchical Assembly Model'', which is a mathematical model of tile-based self-assembling systems in which arbitrarily large pairs of assemblies can combine to form new assemblies.

A \emph{glue} is an ordered pair $(l,s)$, where $l \in \Sigma^+ \cup \{s^* : s\in\Sigma^+\}$ is a non-empty string, called the \emph{label}, over some alphabet $\Sigma$, possibly concatenated with the symbol `$^*$', and $s \in \mathbb{Z}^+$ is a positive integer, called the \emph{strength}. A glue label $l$ is said to be \emph{complementary} to the glue label $l^*$.

A \emph{tile type} is a mapping of zero or more glues, along with \emph{glue states} and possibly \emph{signals}, which will be defined shortly, to the $6$ faces of a unit cube. A \emph{tile} is an instance of a tile type, and is the base component of the STAM$^R$. Each tile type is defined in a canonical orientation, but tiles can be in that orientation or any rotation which is orthogonal to it (i.e. they are embedded in $\mathbb{Z}^3$).

Every glue can be in one of three \emph{glue states}: $\{\texttt{on}, \texttt{latent}, \texttt{off}\}$. If two tiles are placed next to each other, and their adjacent faces have glues $g_1 = (l,s)$ and $g_2 = (l^*,s)$, then those glues can form a \emph{bond} whose \emph{strength} is $s$. We require any copies of glues with the label $l$, or its complement $l^*$, in any given system have the same strength (e.g. it is not allowed to have one glue labeled $l$ with strength $1$ and another labeled $l$ or $l^*$ with strength $2$).

A \emph{signal} is a mapping from a glue $g_s$ (the \emph{source glue}) to an ordered pair, $(g_t, s)$, where $g_t$ (the \emph{target glue}) is a glue on the same tile as $g_s$ (possibly $g_s$ itself) and $s \in \{\texttt{on}, \texttt{off}\}$. If and when $g_s$ forms a bond with its complementary glue on an adjacent tile, the signal is \emph{fired} to change the state of $g_t$ to state $s$. Each glue of a tile type can be defined to have zero or more signals assigned to it. Each signal on a tile can fire at most a single time. When a glue is fired, the state of the target glue is not immediately changed, but the pair $(g_t,s)$ is added to a queue of \emph{pending signals} for the tile containing its glues. When a pending glue is selected for completion (in a process described below), then the state of $g_t$ is changed to $s$ if and only if its current state is $s_0$ and $(s_0,s) \in \{(\texttt{on}, \texttt{off}), (\texttt{latent},\texttt{on}), (\texttt{latent},\texttt{off})\}$. That is, the only valid glue state transitions are $\texttt{on}$ to $\texttt{off}$, or $\texttt{latent}$ to $\texttt{on}$ or $\texttt{off}$.

A \emph{supertile} is (the set of all translations and rotations of) a positioning of one or more connected tiles on the integer lattice $\Z^3$.  Two adjacent tiles in a supertile can form a bond if the glues on their abutting sides are complementary and both are in the $\texttt{on}$ state. Each supertile induces a \emph{binding graph}, a grid graph whose vertices are tiles, with an edge between every pair of bound tiles whose weight is the strength of the bound glues. A supertile is \emph{$\tau$-stable} if every cut of its binding graph cuts edges whose weights sum to at least $\tau$. That is, the supertile is $\tau$-stable if at least energy $\tau$ is required to separate the supertile into two parts. \emph{Assembly} is another term for a supertile, and we use the terms interchangeably, to mean the same thing.

Each tile has a \emph{tile state} that contains the current state of every glue as well as a (possibly empty) set of pending signals and a (possibly empty) set of completed signals. Every supertile consists of not only its set of constituent tiles, but also their tile states, and a set bonds that have formed between pairs of glues on adjacent tiles.

A system in the STAM$^R$ is an ordered triple $(T,S,\tau)$ where $T$ is a finite set of tiles called the \emph{tileset}, $S$ is a \emph{system state} which consists of a multiset of supertiles that each have a count (possibly infinite), and $\tau \in \mathbb{Z}^+$ is the \emph{binding threshold} (a.k.a. \emph{temperature}) parameter of the system which specifies the minimum strength of bonds needs to hold a supertile together. In the initial state of a system, no tiles have pending signals, all pairs of adjacent glues which are both complementary and in the $\texttt{on}$ state in all supertiles have formed bonds and any signals which would have been fired by those bonds are completed, and all distinct supertiles are assumed to start arbitrarily far from each other (i.e. none is enclosed within another). By default (and unless otherwise specified), the initial state contains an infinite count of all singleton tiles in $T$.

A system evolves as a (possibly infinite) series of discrete steps, called an \emph{assembly sequence}, beginning from its initial state. Each step occurs by the random selection and execution of one of the following actions:
\begin{enumerate}
    \item Two supertiles currently in the system, $\alpha$ and $\beta$, are translated and/or rotated without ever overlapping so that they can form bonds whose strengths sum to at least $\tau$. The count of the newly formed supertile is increased by $1$ in the system state and the counts of each of $\alpha$ and $\beta$ are decreased by $1$ (if they aren't $\infty$). In the newly created supertile, from the entire set of pairs of glues which can form bonds, a random subset whose strengths sum to $\ge \tau$ is selected and bonds formed by those glues are added to the set of bonds that have formed for that supertile. Additionally, for each glue which forms a bond, all signals for which it is a source glue, but which aren't already pending or completed, are added to the set of pending signals for its tile.\label{item:combine}
    
    \item For any supertile currently in the system, from the set of pairs of glues which can form bonds but haven't, a glue pair is selected and a bond formed by those glues is added to the set of bonds that have formed for that supertile. Additionally, for each glue which forms that bond, all signals for which it is a source glue, but which aren't already pending or completed, are added to the set of pending signals for its tile.
    
    \item For any supertile currently in the system, a pending signal is selected from the set of pending signals of one of its tiles. If the action specified by that signal is valid, the state of the target glue is changed to the state specified by the signal. The signal is removed from the set of pending signals and added to the set of completed signals. If the action is not valid (i.e. the pair specifying the current state of the target glue and the desired end state is not in $\{(\texttt{on}, \texttt{off}), (\texttt{latent},\texttt{on}), (\texttt{latent},\texttt{off})\}$), then the signal is just removed from the pending set and added to the completed set, and there is no change to the target glue.
    
    \item For a supertile $\gamma$ currently in the system for which there exists one or more cuts of $< \tau$ (which could be the case due to one or more glues changing to the $\texttt{off}$ state), one of those cuts is randomly selected and $\gamma$ is split into two supertiles, $\alpha$ and $\beta$, along that cut. The count of $\gamma$ in the system state is decreased by one (if it isn't $\infty$) and the counts of $\alpha$ and $\beta$ are increased by one (if they aren't $\infty$).
    
\end{enumerate}



Given a system $\calT=(T,S,\tau)$, a supertile is \emph{producible}, written as $\alpha \in \prodasm{T}$, if it either is contained in the initial state $S$ or it can be formed, starting from $S$, by any series of the above steps. A supertile is \emph{terminal}, written as $\alpha \in \termasm{T}$, if it is producible and none of the above actions are possible to perform with it (and any other producible assembly, for list item \ref{item:combine}).


Note that tiles are not allowed to diffuse through each other, and therefore a pair of combining supertiles must be able to translate and/or rotate without ever overlapping into positions for binding. It is allowed, though, for two supertiles, $\alpha$ and $\beta$, to translate and/or rotate into locations which are partially enclosed by another supertile $\gamma$ before binding, potentially creating a new supertile, $\delta$, which would not have been able to translate and/or rotate into that location inside $\gamma$, without overlapping $\gamma$, after forming. However, although the model allows for supertiles to assemble ``inside'' of others, in order to strengthen our results we do not utilize it for the constructions of our positive results, but its possibility does not impact our negative result.

    


\begin{definition}\label{def:finitely-completes}
Given an STAM$^R$ system $\mathcal{T} = (T,S,\tau)$, we say that it \emph{finitely completes} with respect to a set of terminal assemblies $\hat{\alpha}$ if and only if there exists some constant $c \in \mathbb{N}$ such that, if in the initial configuration $S$, each element of $S$ was assigned count $c$, in every possible valid assembly sequence of $\mathcal{T}$, every element of $\hat{\alpha}$ is produced.
\end{definition}

A system which finitely completes with respect to assemblies $\hat{\alpha}$ is guaranteed to always produce those assemblies as long as it begins with enough copies of the (super)tiles in its initial configuration, i.e. it cannot follow any assembly sequence which would consume one or more (super)tiles needed to form those assemblies before making them.

\begin{definition}\label{def:shape}
    A \emph{shape} is a non-empty connected subset of $\mathbb{Z}^3$, i.e. a connected set of unit cubes each of which is centered at a coordinate $\vec{v} \in \mathbb{Z}^3$. A \emph{finite shape} is a finite connected subset of $\mathbb{Z}^3$.
\end{definition}

In this paper, we consider shapes to be equivalent up to rotation and translation and unless stated otherwise explicitly, we will use the word \emph{shape} to refer only to \emph{finite shapes}.

\begin{definition}\label{def:bounding-box}
    Given a shape $s$, a \emph{bounding box} is a rectangular prism in $\mathbb{Z}^3$ which completely contains $s$. The \emph{minimum bounding box} is the smallest such rectangular prism.
\end{definition}

\begin{definition}\label{def:enclosed-cavity}
    Given a shape $s$, we use the term \emph{enclosed cavity} in $s$ to refer to a set of connected points in $\mathbb{Z}^3$ that are not contained in $s$ and for which no path in $\mathbb{Z}^3$ exists that does not intersect at least one point in $s$ and gets infinitely far from all points in $s$.
\end{definition}
    
\begin{definition}\label{def:bent-cavity}
    Given a shape $s$, we use the term \emph{bent cavity} in $s$ to refer to a set of connected points in $\mathbb{Z}^3$ contained inside of the minimum bounding box of $s$, $b_s$, but not contained within $s$ itself, such that it includes some points which can be reached by straight lines in $\mathbb{Z}^3$ beginning from points in $b_s$, and some points which cannot be reached by straight lines in $\mathbb{Z}^3$ beginning from points in $b_s$.
\end{definition}

See Figure \ref{fig:bent-cavity} for an example of a bent cavity.
\begin{figure}
    \centering
    \includegraphics[width=1.0in]{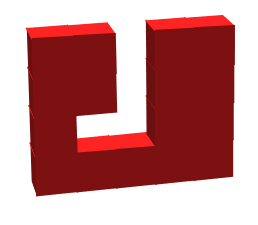}
    \caption{Example of a bent cavity, assuming that the planes on the sides into and out of the page were also filled in, leaving a single-cube-wide path into the interior of the shape.}
    \label{fig:bent-cavity}
\end{figure}

\begin{definition}\label{def:shape-encoding}
We define a \emph{shape encoding function} $f_e$ as a function which, given as input an arbitrary shape $s$, returns a unique finite set $E$ of binary strings, each unique for the shape $s$, such that there exists a \emph{shape decoding function}, $f_d$ and $f_d(e) = s$ for all $e\in E$.
\end{definition}


The shape encoding function we will define by construction in the proof of Lemma \ref{lem:encoder} will generate a set of binary strings for each input shape $s$ such that each string encodes the points of the shape starting from a different reference corner and rotation of a bounding box. That can lead to up to 24 unique binary strings (for 3 rotations of each of 8 corners) for most shapes, but less for those with symmetry.




\begin{definition}\label{def:shape-neighborhood}
Given a shape $S$ and a point $p=(x, y, z)\in S$, we define the \emph{neighborhood} of $p$ in $S$ to be the set $S\cap\{(x+1,y,z),(x-1, y,z),(x,y+1,z),(x,y-1,z),(x,y,z+1),(x,y,z-1)\}$. We also say that neighborhoods are equivalent up to rotation, so there is 1 neighborhood containing 1 point, 2 with 2 points, 2 with 3 points, 2 with 4 points, 1 with 5 points, and 1 with 6 points.
\end{definition}

\begin{definition}\label{def:uniformly-covered-assembly}
We define a \emph{uniformly covered assembly} as an assembly $\alpha$ where every exposed side of every tile has the same strength 1 glue which is $\texttt{on}$. Additionally, if $s$ is the shape of $\alpha$, we require that for every 2 points $p, q\in s$ with the same neighborhood, a tile of the same type is located in both locations $p$ and $q$ in $\alpha$.
\end{definition}

A uniformly covered assembly has the same glue all over its surface, with no glues marking special or unique locations, and has the same tile type in each location with the same neighborhood, so such an assembly can convey no information specific to particular locations, orientation, etc. 

\begin{definition}\label{def:deconstructable-assembly}
We define a \emph{deconstructable assembly} as an assembly where (1) all neighboring tiles are bound to each other by one or more glues whose strengths sum to $\ge \tau$, and (2) each tile contains the glue(s) and signal(s) necessary to allow for all glues binding it to its neighbors to be turned $\texttt{off}$.
\end{definition}

In the following definitions, we will use the term \emph{junk assembly} to refer to an assembly that is not a ``desired product'' of a system, but which is a small assembly composed of tiles which were used to facilitate the construction but are now terminal and cannot interact any further.

\begin{definition}[Universal shape encoder]\label{def:binary-encoding}
Let $\mathcal{S}$ be the set of all finite shapes, let $f_e$ be a shape encoding function, let $c \in \mathbb{N}$ be a constant, and let $E$ be a tileset in the STAM$^R$. If, for every finite subset of shapes $S' \subset S$, there exists an STAM$^R$ system $\mathcal{E}_{S'} = (E,\sigma_{S'},\tau)$, where $\sigma_{S'}$ consists of infinite copies of assemblies of each shape $s \in S'$ and also infinite copies of the singleton tiles from $E$, such that (1) for every shape $s \in S'$ there exists at least one binary string $b_s \in f_e(s)$ and there exist infinite terminal assemblies of $\mathcal{E}_{S'}$ that contain glues in the $\texttt{on}$ state on the exterior surfaces of those assemblies that encode $b_s$ 
(which we refer to as an \emph{assembly encoding $s$}), (2) every terminal assembly is either an assembly encoding some $s \in S'$ or a ``junk assembly'' whose size is bounded by $c$, and (3) no non-terminal assembly grows without bound, then we say that $E$ is a \emph{universal shape encoder} with respect to $f_e$.

\end{definition}

\begin{definition}[Universal shape decoder]
Let $\mathcal{S}$ be the set of all finite shapes, let $f_e$ be a shape encoding function, let $c \in \mathbb{N}$ be a constant, and let $D$ be a tileset in the STAM$^R$. If, for every finite subset of shapes $S' \subset S$, there exists an STAM$^R$ system $\mathcal{D}_{S'} = (D,\sigma_{S'},\tau)$, where $\sigma_{S'}$ consists of infinite copies of assemblies each of which encode a shape $s \in S'$ with respect to $f_e$, and also infinite copies of the singleton tiles from $D$, such that (1) for every shape $s \in S'$ there exist infinite terminal assemblies of shape $s$, (2) every terminal assembly is either an assembly of the shape of some $s \in S'$ or a ``junk assembly'' whose size is bounded by $c$, and (3) no non-terminal assembly grows without bound, then we say that $D$ is a \emph{universal shape decoder} with respect to $f_e$.
\end{definition}

\begin{definition}[Universal shape replicator]
Let $\mathcal{S}$ be the set of all finite shapes and let $R$ be a tileset in the STAM$^R$, and let $c \in \mathbb{N}$ be a constant. If, for every finite subset of shapes $S' \subset S$, there exists an STAM$^R$ system $\mathcal{R}_{S'} = (R,\sigma_{S'},\tau)$, where $\sigma_{S'}$ consists of infinite copies of assemblies of each shape $s \in S'$ and also infinite copies of the singleton tiles from $R$, such that (1) for every shape $s \in S'$ there exist infinite terminal assemblies of shape $s$, (2) every terminal assembly is either an assembly of the shape of some $s \in S'$ or a ``junk assembly'' whose size is bounded by $c$, (3) the number of assemblies of each shape $s \in S'$ grows infinitely, and (4) no non-terminal assembly grows without bound, then we say that $R$ is a \emph{universal shape replicator}.
\end{definition}

\subsection{STAM$^R$ Gadgets and Tools}\label{sec:gadgets}

Throughout our results we repeatedly make use of several small assemblies of tiles, referred to as \emph{gadgets}, and patterns of signal activations to accomplish tasks such as keeping track of state, removing specific tiles, and passing information across an assembly. In this section we describe several of these gadgets and signal patterns so that they can later be referenced during our construction. We intend that this section also serve as a basic introduction by example to the dynamics of signal tile assembly.

\paragraph*{Detector Gadgets}
Detector gadgets are used to detect when a specific set of tiles exist in a particular configuration relative to one another in an assembly. For a detector gadget to work, the tiles to be detected need to each be presenting a glue unique to the configuration to be detected. The strength of these glues should add to at least the binding threshold $\tau$, but the total strength of any proper subset of the glues should not. If two or more tiles then exist in the configuration expected by the detector gadget, the gadget can cooperatively bind with the relevant glues. Upon binding, any signals with the newly bonded glues as a source will fire. These signals can be in the ``detected tiles'' or in the detector itself and can be used to initiate some other process based on the condition that the tiles exist in the specified configuration. More often than not, it's also desirable for signals within the detector gadget to deactivate its own glues so that it does not remain attached to the assembly after the detection has occurred.

\begin{wrapfigure}{r}{0.5\textwidth}
    \centering
    \includegraphics[width=0.48\textwidth]{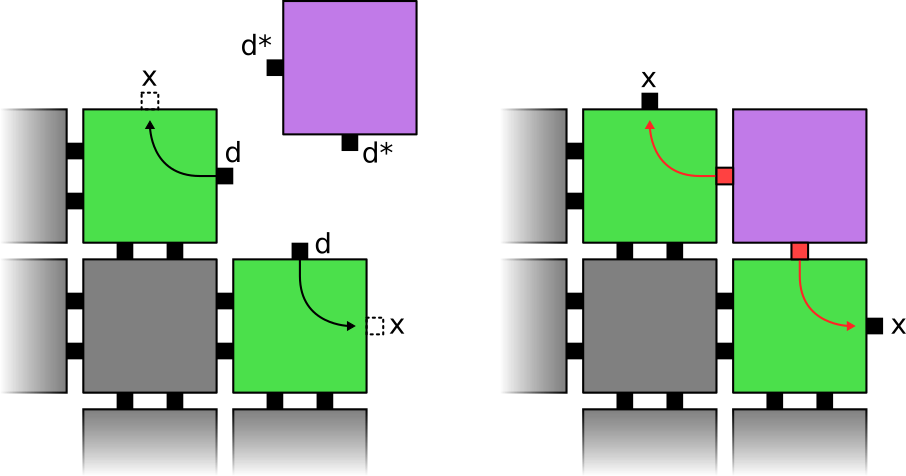}
    \caption{A simple detector gadget example.}
    \label{fig:detector-single-tile}
\end{wrapfigure}

Detector gadgets can exist in many forms depending on the configuration to detect, but the most simple is a single tile. Illustrated in Figure \ref{fig:detector-single-tile} is a simple detector gadget designed to detect 2 diagonally adjacent tiles, each presenting a strength-1 glue of type $d$ towards a shared adjacent empty tile location. In this case, $\tau=2$ and the detected tiles are designed to activate their $x$ glues upon a successful detection. In general, detector gadgets can be made up of more than 1 tile. Duples of tiles can be used for instance to detect immediately adjacent tiles each presenting some specific glue on the same side. For detector gadgets consisting of more than 1 tile, the component tiles must be designed to have unique $\tau$-strength glues between them so that the components can bind together piece-wise to form the whole gadget. Because all of the glues presented for the detection are needed to reach a cumulative strength of $\tau$, only after it is fully formed will it be able to detect tiles and thus partially assembled detector gadgets will not erroneously perform partial detections. It is assumed in our results that signals within a detector gadget itself will cause the gadget to dissolve after a detection.

\paragraph*{Corner Gadgets}
\begin{figure}[H]
    \centering
    \includegraphics[width=0.9\textwidth]{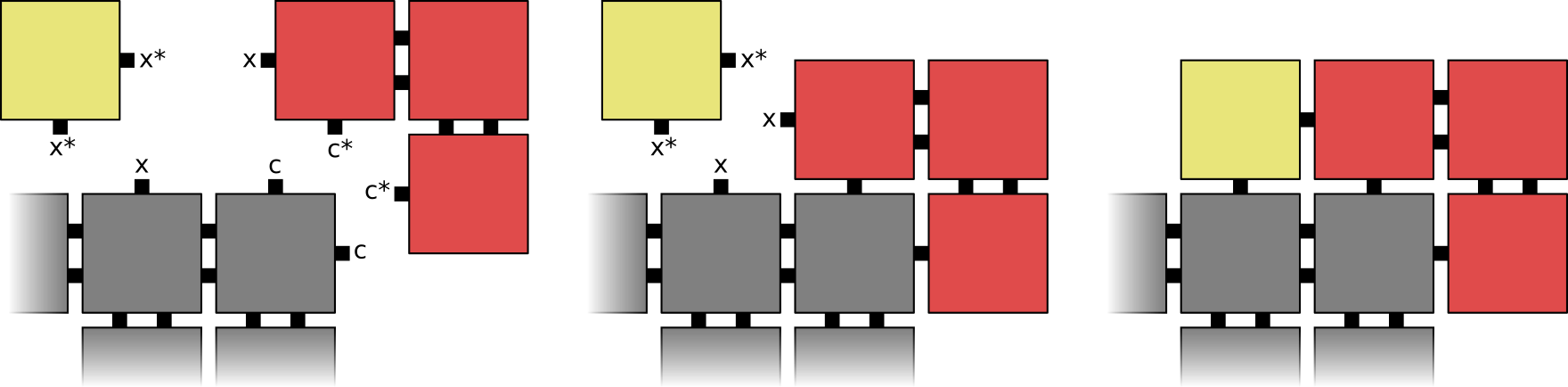}
    \caption{A corner gadget example.}
    \label{fig:corner-gadget-example}
\end{figure}
Corner gadgets are a specific type of detector gadget which are used primarily for facilitating the attachment of other tiles on the surface of some assembly. Corner gadgets can either be 2D, consisting of 3 tiles arranged in a $2\times 2$ square with one corner missing, or 3D, consisting of 7 tiles arranged in a $2\times 2\times 2$ cube with one of the corners missing. Because of this shape, a corner gadget is able to cooperatively bind to any single tile of an assembly with 2 accessible, adjacent faces. These faces must be presenting specified glues whose cumulative strength is at least $\tau$, but neither individually is. Illustrated in Figure \ref{fig:corner-gadget-example} is the side view of a 2D corner gadget attaching to an assembly. After the attachment, it is then possible for additional tiles to cooperatively bind along the surface of the assembly. This behavior is useful for initiating the growth of shells of tiles around an assembly as will be seen in our later construction.

Like with detector gadgets, signals fired from the binding of a corner gadget can also be used to initiate other tasks, though special care needs to be taken for 3D corner gadgets when $\tau=2$. Because a 3D corner gadget has 3 interior faces which can have glues to bind with a tile on the corner of an assembly, it is often desirable to fire signals from all 3 of these glues; however, because only 2 glues are necessary to meet the binding threshold when $\tau=2$, the third may not form a bond immediately. If it is planned for the corner gadget to eventually detach, then it is crucial that any signals causing the corner gadget to detach cannot fire until all 3 of the interior glues have first bound. This can often be accomplished using \emph{sequential signaling} as described below.

\begin{figure}[H]
    \centering
    \includegraphics[width=0.9\textwidth]{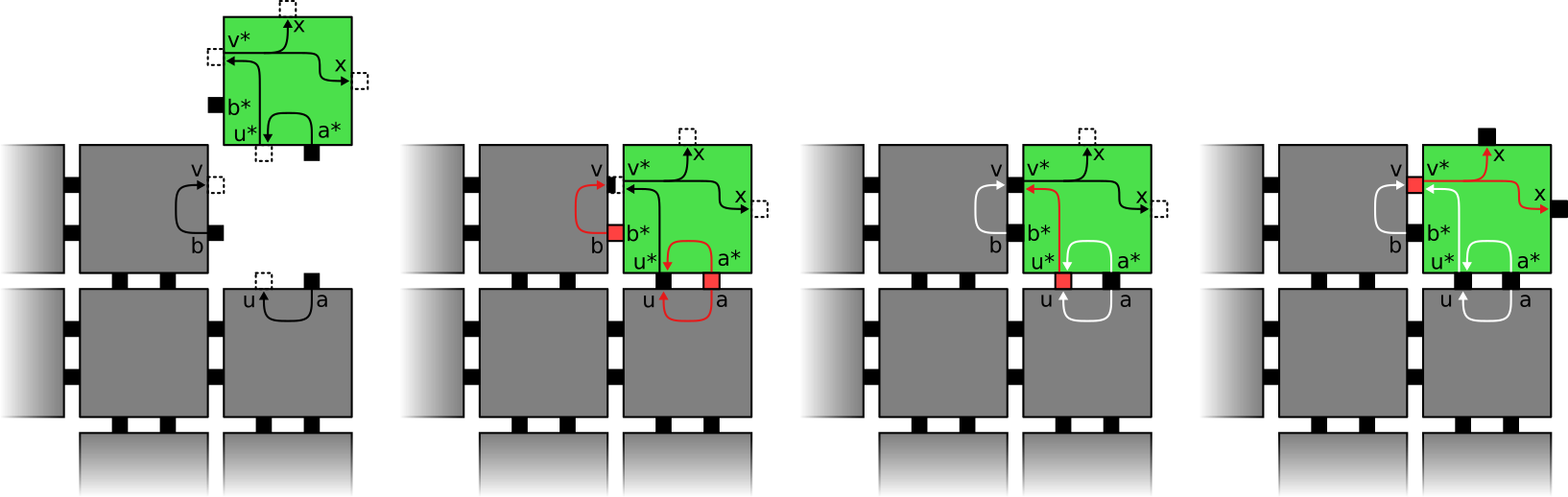}
    \caption{Sequential signaling example.}
    \label{fig:sequential-signaling}
\end{figure}
\paragraph*{Sequential Signaling}
By carefully adding additional helper glues and signals to a tile or tiles, we can ensure that signals in our tiles are fired in a specific order or ensure that a certain set of glues has successfully bound before certain signals are fired. The way in which this is done depends on the exact situation, but as an example consider the situation illustrated in Figure \ref{fig:sequential-signaling}. In this situation we want the green tile to cooperatively bind to the assembly via glues of type $a$ and $b$. Once this happens, we want to first activate additional glues of type $u$ and $v$ between the green tile and assembly so that each side of the green tile is attached to the assembly with strength 2, then we want glues of type $x$ on the other sides of the green tile to activate. The arrangement of signals illustrated in Figure \ref{fig:sequential-signaling} guarantees that the $x$ glues cannot activate before both the $u$ and $v$ glues do, since the signals which activate the $x$ glues are dependent on the glues $u$ and $v$. A similar arrangement of signals and glues is used to implement gadgets called \emph{filler tiles} in our construction.

\paragraph*{Tile Conversion}
It is often useful for tiles to change behavior after receiving a specific signal. This can be done by having signals activate a new set of glues on the tile and deactivate old ones. This can be thought of as converting the tile into a different type of tile, but it's important to note that this process cannot happen indefinitely nor arbitrarily. Every tile conversion has to be prepared in the signals and latent glues of the tile and once those signals fire, they cannot fire again. It is possible for a tile to convert to another several times, but such a tile must have the necessary glues and signals for each conversion separately. It is also often possible achieve this behavior by detachment of one tile and attachment of another in the same location, though special care needs to be taken so that no other tiles can attach in the location during the conversion.

\paragraph*{Tile Dissolving}
For any arbitrary set of glues on a tile, we use the term \emph{dissolving} to refer to the process of initiating signals which turn all possible glues to the \texttt{off} state (Figure~\ref{fig:dissolve}).
We note that due to the asynchronous nature of the model that no guarantee can be made with regards to the order of the processing of the signals.
Tiles break apart from their supertile once a strength $\tau$ bond no longer exists between itself and its neighbor tiles.
However other glues may be active when the tile does so, leading to the possibility of undesired binding due to exposed glues which are in the \texttt{on} state with a pending \texttt{off} signal.

\begin{figure}[htp]
    \centering
    \includegraphics[width=3.5cm]{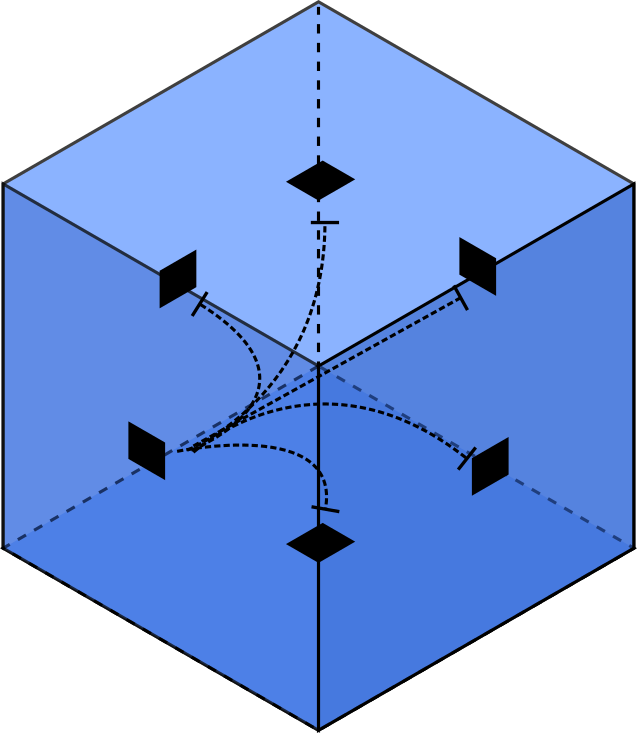}
    \caption{For some glue which initiates the dissolution of a tile, when bound to its complement it can send messages to all glues on all faces to turn to the \texttt{off} state. We use the flat head to indicate that the glue adjacent to the flat head is sent an \texttt{off} signal by the binding of the glue at the opposite end of the line. Such a glue can potentially be present on each face of a tile.}
    \label{fig:dissolve}
\end{figure}

\paragraph*{Message Following}
We show how to pass a message through a sequence of tiles such that after the message has been passed, a second message can be passed through the exact same sequence of tiles in the same order. 
For example, signals propagate a $g$ message through a sequence of tiles $\{\T_i\}_{i=0}^n$ (not necessarily distinct). 
We then propagate a $br$ message through a series of glue activations such that this message follows the sequence of tiles $\{\T_i\}_{i=0}^n$ in that order. In this case, we say that the $br$ message \emph{follows} the $g$ message.

Figure~\ref{fig:signal-following-none} shows a $g$ message being passed through a tile. Let $T_G$ denote this tile. This message enters from the south and then may potentially be output through the north, east, or south depending on if collisions occur.
 The goal is to ensure that a second message can be output through exactly that same side (and no others). Other cases where the $g$ message enters through the north, east, or west are equivalent up to rotation. 
 For each possible output signal of the $g$ glue in $T_G$, we define glues on the signal input side of the $T_G$ which are activated by the output $g$ glue being bound. 
 As shown in Figure~\ref{fig:signal-following-none}, the north $g$ glue activates $brn^\prime$, the east $g$ glue activates $bre^\prime$, and the south $g$ glue activates $brs^\prime$. Informally, the activated $brn^\prime$, $bre^\prime$, or $brs^\prime$ glue ``records'' the output side of the $g$ message. In the case shown in Figure~\ref{fig:signal-following-none} where the $g$ message enters from the south, the $brn^\prime$, $bre^\prime$, and $brs^\prime$ glues are sufficient for recording the output side of the $g$ message. In cases where the $g$ message enters through the north, east, or west, a $brw^\prime$ glues is required to record the case where the $g$ message exits through the west side of a tile. The $br$ signal is then propagated using $brn^\prime$, $brs^\prime$, $bre^\prime$, and $brw^\prime$ glues. Figure~\ref{fig:signal-following-forward} depicts the signals and glues for propagating the $br$ signal in the case where the $g$ message enters from the south. In this case the $br$ signal will also enter from the south. 
 The $br$ signal is propagated through $T_G$ as exactly one of the $brn^\prime$, $brs^\prime$, and $bre^\prime$ glues binds to one of the $brn$, $bre$, and $brs$ glues on the output side of a tile to the south of $T_G$ that is propagating $br$. All of the $brn$, $bre$, and $brs$ glues must be activated as the tile to the south of $T_G$ has no ability to know which direction the $g$ message of $T_G$ will take. The $br$ signal passed to $T_G$ will have the same output side as the $g$ signal. For example, if the $g$ message enters from the south and exits through the east, then, as shown in Figure~\ref{fig:signal-following-none}, the glue $bre^\prime$ will be activated; $brn^\prime$ and $brs^\prime$ will remain latent. Then, as the $br$ signal propagates through the tile to the south of $T_G$, $brn$, $bre$, and $brs$ are all activated on the north side of the tile. When $bre$ and the $bre^\prime$ glue on the south edge of $T_G$ bind, this binding event activates the glues $bre$, $brs$, and $brw$ on the east edge of $T_G$, effectively propagating the $br$ signal to the tile to the east of $T_G$. This is shown in Figure~\ref{fig:signal-following-forward}. Notice that there are no signals belonging to $T_G$ that fire when $brs^\prime$ binds. This is because no signals are needed to propagate $br$ to the south of $T_G$. The binding of $brs$ and $brs^\prime$ are enough to propagate $br$ to the south of $T_G$. 
\begin{figure}[htp]
\centering
  \subfloat[][

        An example of signals used to propagate an $g$ message CCW.]{%
        \label{fig:signal-following-none}%
	        \includegraphics[width=1.1in]{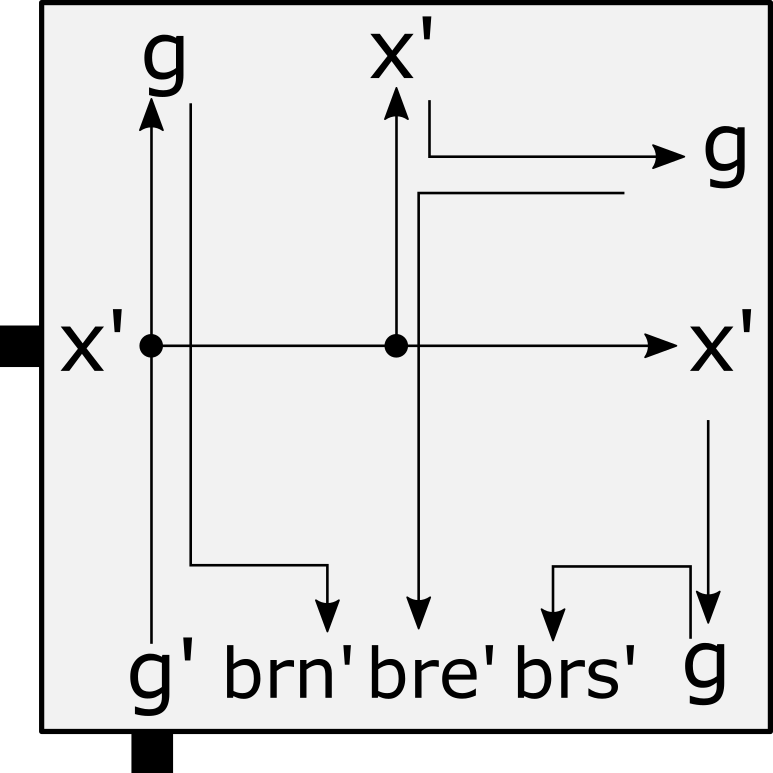}
        }%
        \quad\quad
  \subfloat[][A $br$ message that is following a previously passed $g$ message. 

  ]{%
        \label{fig:signal-following-forward}%
        		\includegraphics[width=1.1in]{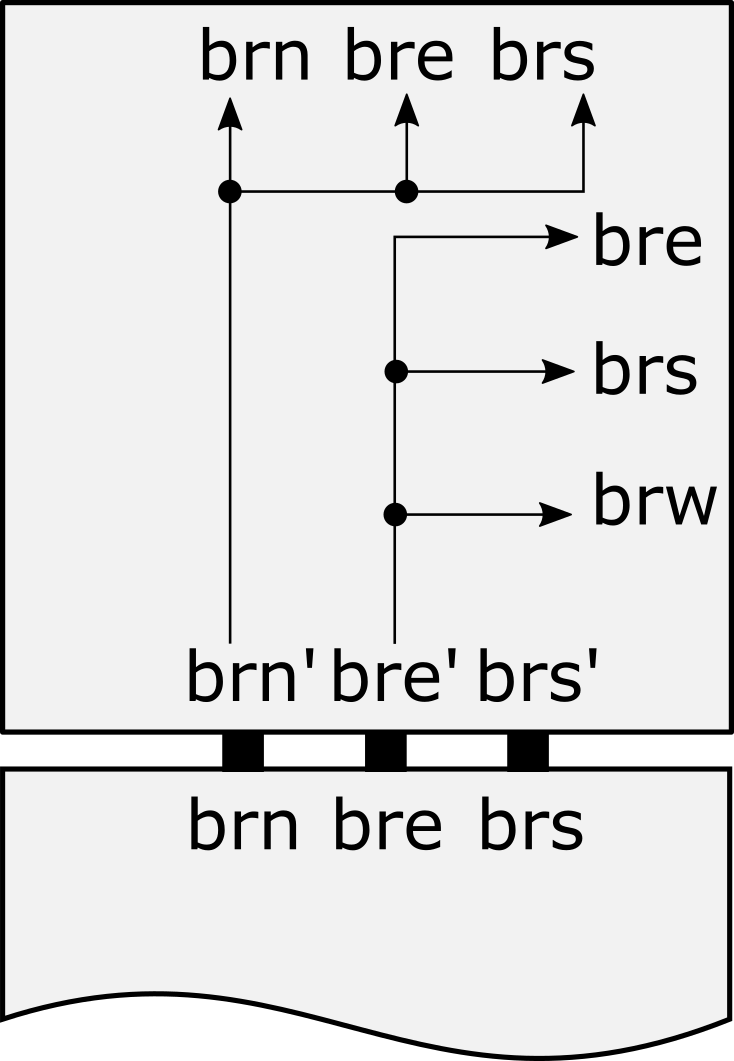}
        }%
  \caption{Tiles which demonstrate signal following.\vspace{-15pt}}
  \label{fig:signal-following-all}
\end{figure}
\section{3D Shape Replication}\label{sec:3d-rep}

In this section, we show that there is a tileset in the STAM$^R$ which is capable of replicating arbitrary shapes. This is stated in Theorem \ref{thm:replicator}, and we prove it by providing modular constructions capable of encoding and decoding arbitrary sets of shapes which are given by Lemma \ref{lem:encoder} and Lemma \ref{lem:decoder}, respectively, and then discussing how they can be combined to replicate shapes.

\begin{theorem}\label{thm:replicator}
There exists a tileset $R$ in the STAM$^R$ which is a universal shape replicator, such that for the systems using $R$ (1) all input assemblies are uniformly covered, (2) the constant $c$ which bounds the size of the junk assemblies equals 4, and (3) they finitely complete with respect to a set of terminal assemblies with the same shapes as the input assemblies.
\end{theorem}

\begin{lemma}\label{lem:encoder}
There exist a shape encoding function $f_e$, and a tileset $E$ in the STAM$^R$ which is a universal shape encoder with respect to $f_e$, such that for the systems using $E$ (1) all input assemblies are uniformly covered, (2) the constant $c$ which bounds the size of the junk assemblies equals 4, and (3) they finitely complete with respect to a set of terminal assemblies which encode the shapes of the input assemblies.
\end{lemma}

\begin{lemma}\label{lem:decoder}
There exist a shape decoding function $f_d$, and a tileset $D$ in the STAM$^R$ which is a universal shape decoder with respect to $f_d$, such that for the systems using $D$ (1) the constant $c$ which bounds the size of the junk assemblies equals 3 
and (2) they finitely complete with respect to a set of terminal assemblies with the same shapes as those encoded by the input assemblies.
\end{lemma}

We now prove Lemmas \ref{lem:encoder} and \ref{lem:decoder}, and consequently Theorem \ref{thm:replicator}, by construction. In the following few sections we describe the process by which an STAM$^R$ system can encode arbitrary shapes. We then show how an STAM$^R$ system can construct arbitrarily shaped assemblies from the encodings produced by the encoding system. Additionally, these systems make use of universal tilesets $E$ and $D$ respectively, meaning that regardless of the shapes to be encoded or decoded, our systems never require additional tiles besides those from $E$ and $D$. These tilesets can then be combined to create a tileset $R = E \cup D$ which is then a universal shape replicator. It should also be noted that constructing the universal encoder and decoder separately allows for additional complex tasks to be performed in the STAM$^R$. For example, tiles are capable of simulating the execution of Turing machines to perform arbitrary computation. As will be briefly discussed later, this means that once shapes have been encoded, it is then possible to manipulate the encodings using simulated Turing machines before the decoding process. Such behavior is clearly much more general than shape replication.



\subsection{Forming a bounding box and electing a corner as ``leader''}\label{sec:bounding-leader}

Here we describe the process by which a set of arbitrary shapes $S = \{s_1, \ldots, s_n\}$ can be encoded in the STAM$^R$ using a universal tileset $E$. It should be noted that we don't explicitly list each tile type in $E$; rather, much like how it is more useful to use pseudo-code instead of compiled machine code when describing an algorithm, we describe the tiles in $E$ implicitly by their functionality, noting that there are many essentially equivalent ways to design tiles which perform the necessary tasks and a discussion of the finer details regarding exactly which signals and glue types are used in each instance would be less informative.

Given our set $S$ of shapes, we define our STAM$^R$ system $\mathcal{E}_S$ to be the triple $(E, \Sigma_S, \tau=2)$ where $\Sigma_S$ is our initial system state containing assemblies of the shapes in $S$. This state consists of all tiles in $E$, each with an infinite count, and additionally consists of a set $A=\{\alpha_1, \ldots, \alpha_n\}$ of uniformly covered, deconstructable assemblies such that the shape of $\alpha_i$ is $s_i$ for $i=1,\ldots,n$. The assemblies of $A$ are called our \emph{shape assemblies} and are made only of tiles from a fixed subset of $E$ called \emph{shape tiles}. Note that the glues and signals defined in these shape tiles are not used to encode any information regarding the structure of our shape assemblies; any shape specific information is inferred during the encoding process and the shape tiles simply contain the necessary glues and signals to perform basic tasks required for the encoding process, none of which are specific to any particular part of the shape assemblies. 
Additionally, we will define tile encoding and decoding functions, $f_e$ and $f_d$ during our construction. Essentially our encoding of a shape consists of a sequence of rows of binary values, each row corresponding to a 1-dimensional slice within the minimal bounding box of our shape, with $1$ representing a location in the shape and $0$ representing a location not in the shape. 

The encoding process described below can be largely broken down into 3 steps. First, a bounding box is constructed around the shape assemblies using special tiles which are distinct from the shape tiles. Then, one of the corners of the box is elected non-deterministically to be the \emph{leader corner} to provide an origin point which will represent the first tile location of our encoding. Finally, from the leader corner, the shape will be disassembled tile-by-tile during which an encoding assembly will be constructed, recording for each disassembled tile whether it is part of the shape or not (i.e. a ``filler'' tile used to assist the construction). During our description of the encoding process, we will follow the process for a single shape assembly $\alpha_i$, but note that all shape assemblies are encoded simultaneously in parallel in $\mathcal{E}_S$.


\subsubsection{Bounding Box Assembly Construction}

The first step in our encoding process begins by forming a bounding box assembly $\beta_i$ through the attachment of special tiles, called \emph{filler tiles}, to $\alpha_i$. These filler tiles cooperatively bind to 2 diagonally adjacent tiles of our shape assembly in order to fill out any concave portions. When a filler tile attaches to an assembly, signals are fired from the newly bound glues which activate additional glues between the filler tile and shape assembly. These new glues ensure that the filler tile is bound with strength 2 on each face to the shape assembly as this will be important during the disassembly process. After the filler tile is firmly attached with 2 strength-2 bonds, signals are then fired within the filler tile which activate strength-1 glues of type $g_f$ on all other faces. These will be used for further filler tile attachment. Figure \ref{fig:fill-convex-site} illustrates the attachment of a filler tile to an assembly and shows how sequential signaling is used to ensure that the filler tile is attached with strength 2 on both of its input faces before activating glues on each of its output faces.

Because filler tiles must be able to cooperatively bind to both shape tiles and other previously attached filler tiles, we need 3 unique types of filler tiles: One which initially presents 2 glues of type $g_x^*$ to bind with 2 shape tiles, one which initially presents 2 glues of type $g_f^*$ to bind with 2 other filler tiles, and one which presents one of each glue to cooperatively bind with a shape tile and a filler tile. Each type of filler tile is otherwise identical. Because we've chosen our binding threshold $\tau=2$, the two initially present glues are sufficient for binding into any location on the assembly with at least 2 adjacent shape or filler tiles. The signals from the binding of these glues then activates additional glues on the same faces which ensures that the filler tile is attached with strength 2 on two separate faces, regardless of whether or not additional filler tiles later bind to this one. This property will be used to guarantee that the assembly stays connected during the disassembly process.

\begin{figure}
    \centering
    \includegraphics[width=0.9\textwidth]{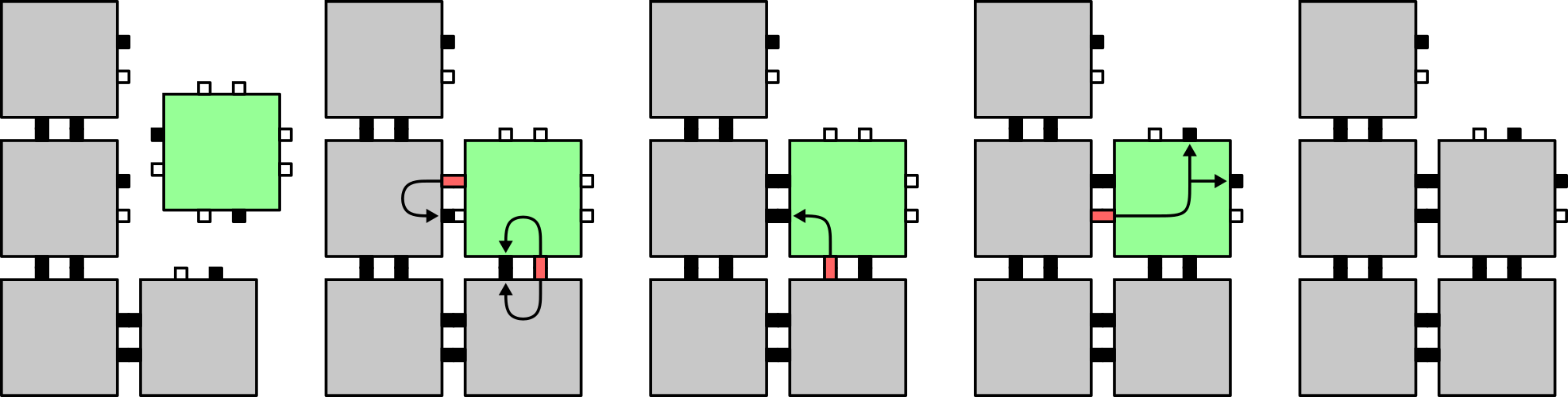}
    \caption{Filler tile binding to a concave site. Once a filler tile attaches cooperatively, signals activate glues on the filler tile and adjacent tiles. These glues ensure that the filler tile is attached with strength 2 on all sides. These glues are activated sequentially and once both are in the \texttt{on} state, signals activate output glues on all other sides of the filler tile. Once these signals activate, the super tile has 1 fewer concave site and the filler tile behaves as though it is just another tile on the supertile. While depicted in 2D for clarity, this occurs in 3D during our construction, but the idea is the same.}
    \label{fig:fill-convex-site}
\end{figure}



Eventually, after sufficiently many filler tiles have attached, there will be no more locations in which another filler tile can attach. There are often many ways in which this can occur for any shape assembly, but the resulting bounding box assembly will always be a minimal bounding box of our shape. It should be noted that its possible that not every location within the bounding box is filled. This can occur if the original shape had enclosed cavities, but can also occur because the attachment of filler tiles can create additional cavities as they attach. This is not a problem and it will always be possible for filler tiles to complete the outer surface of the bounding box. Additionally, this bounding box will be uniformly covered by glues of type $g_x$ and $g_f$.

\subsubsection{Detecting Bounding Box Completion}

In order to continue with the encoding process, we first need to verify that the bounding box is fully formed. This is done by growing a shell of tiles around our assembly. This shell, which we call the \emph{inner shell}, is able to grow to completion only if the assembly is a fully formed bounding box. Figure \ref{fig:inner-shell} illustrates the high-level construction of the inner shell around a fully formed bounding box.

Growth of the inner shell begins with the attachment of corner gadgets to our assembly. We use 2 types of 3D corner gadgets, one which is able to bind to a corner of our assembly presenting 3 glues of type $g_x$ and one which is able to bind to a corner presenting 3 glues of type $g_f$ (note that at $\tau=2$ only two glues are needed for a corner gadget to attach, but any tile allowing a corner gadget to attach must expose all 3). That is, the corner gadgets can attach either to a shape tile or a filler tile on a corner of our assembly. Note that these gadgets exist in our system while the bounding box is being constructed; therefore, it's possible that corner gadgets attach to tiles in our assembly before the bounding box has been fully constructed. Additionally, special care needs to be taken when the bounding box surrounding our shape assembly has at least one side of dimension 1. The details of the inner shell's construction is described below and these various cases are addressed.

\begin{figure}
    \centering
    \includegraphics[width=\textwidth]{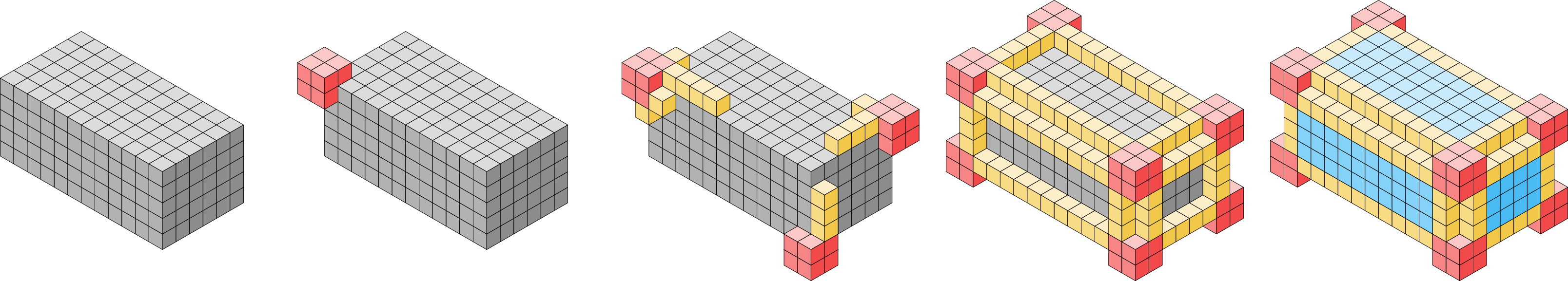}
    \caption{Growth of the inner shell around a bounding box (illustrated in gray). Growth begins by the attachment of corner gadgets (red). Cooperative binding with the corner gadgets and bounding box allow edge tiles to attach (yellow). Cooperation between the edge tiles and the bounding box then allow filler verification tiles (blue) to grow which are used to fill in the faces of the inner shell. The process by which these verification tiles bind to the bounding box ensures that there are no gaps or protrusions on the bounding box surface. }
    \label{fig:inner-shell}
\end{figure}

When a corner gadget attaches to our assembly, signals from the attachment cause strength-1 glues to activate on the faces of the corner gadget which point parallel to the surface of our assembly. These glues will be used to allow cooperative attachment of special \emph{edge tiles} that will attach in a line along the edges of a completed bounding box. The glues activated on the corner gadget can either be of type $g_\text{edge}^L$ or $g_\text{edge}^R$ depending on which face of the corner gadget they reside. Glues of type $g_\text{edge}^L$ indicate that the edge to be grown is a left edge of the bounding prism relative to the direction of growth of the edge while glues of type $g_\text{edge}^R$ indicate a right edge.

Like filler tiles, edge tiles initially have 2 active glues on adjacent faces: one of these glues is either of type $g_\text{edge}^{L*}$ or $g_\text{edge}^{R*}$ so as to be complementary to the glue presented by the corner gadget, and one of type $g_x^*$ or $g_f^*$ so as to also be complementary to a glue on the surface of our assembly. Since any combination of these glues is necessary, there are 4 unique types of edge tiles. Once an edge tile has cooperatively attached to our assembly, signals are fired which activate another glue of type $g_\text{edge}^L$ or $g_\text{edge}^R$ to allow additional edge tiles to cooperatively attach to it and the assembly. Additionally, glues are activated on all other exposed sides of the edge tile which will be used by detector gadgets later. These glues are unique to the specific face of the edge tile so that detector gadgets can distinguish between the interior and exterior sides of an edge as well as the side of the edge tile pointing away from the assembly. Although tiles are allowed to rotate in the STAM$^R$ and don't have fixed orientations, this directionality can be enforced by the relative orientations of the two glues used for the initial binding of a tile. Edge tiles will continue to grow along the surface of our assembly from corner gadgets until they are either blocked by another tile, reach the end of the surface of our assembly, or it is detected that the edge is invalid. 

\begin{wrapfigure}{r}{0.5\textwidth}
    \centering
    \includegraphics[width=0.48\textwidth]{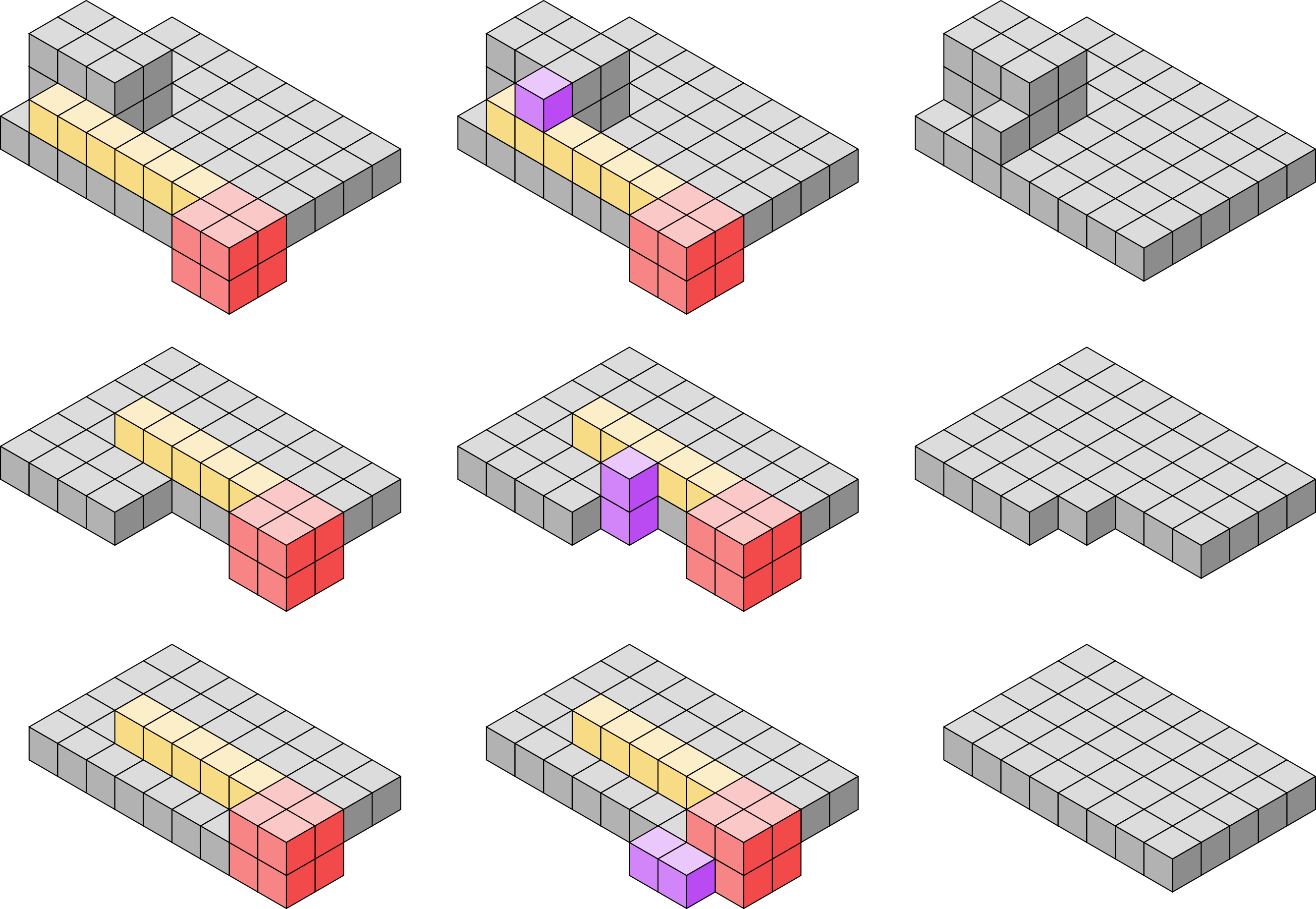}
    \caption{Detecting and resolving invalid edges}
    \label{fig:invalid-edges}
\end{wrapfigure}
For an edge to be valid, there must be no shape or filler tiles adjacent to any edge tiles except for those underneath the edge tiles to which the edge tiles cooperatively attached; additionally, if an edge is a right (respectively, left) edge, then there must not be a shape or filler tile occupying a location diagonally adjacent to the right (resp., left) of the edge tiles making up the edge with respect to the forward growth direction of the edge. Edge tiles which violate these validity conditions can be easily detected using detector gadgets specific to the particular situation as illustrated in Figure \ref{fig:invalid-edges}. Following the attachment of such a detector gadget, a signal is propagated along the edge causing all connected edge tiles and corner gadgets to dissolve. Before this signal is propagated though, signals from the detector gadget ensure that a new filler tile is effectively added to the assembly in a safe location (that is without causing the eventual bounding box to be bigger than necessary). This is done using signals from the detector gadget to convert one of its own tiles or the detected invalid tile into a filler tile. This conversion is done so that we don't risk infinite assembly sequences wherein a corner gadget attaches, attempts to grow an invalid edge, and dissolves repeatedly. Because a filler tile is always effectively added upon detection of an invalid edge, eventually it will be impossible for invalid edges to occur.

In the case where a valid edge is blocked by another tile, then there are 2 possibilities: (1) the edge is blocked by a shape or filler tile, or (2) the edge is blocked by another edge or corner gadget. If a filler tile blocks the path, then like with invalid edges, a detector gadget can cooperatively bind to the blocking tile and the edge tile, convert the edge tile into a filler tile, and propagate a dissolve signal down the remaining edge tiles. If another edge tile or corner gadget blocks the edge, then we need to determine if the blocking tile is part of a valid edge. If the edge is invalid, then it will eventually dissolve and nothing needs to be done. Otherwise, the tile blocking our edge belongs to another valid edge. In this case the meeting point can either be at a corner of our assembly or along the edge of our assembly. Because of the unique glues presented on all sides of an edge tile, these situations can easily be differentiated by detector gadgets. If the meeting point is a corner, then signals from the corresponding detector gadget will cause the corner to convert to a piece of a corner gadget. The remaining corner gadget pieces can then attach and the result will be a corner gadget connected to two incoming edges. If the meeting point is an edge, the detector gadget will fire signals to activate glues between the colliding edge tiles connecting the edge tiles and allowing future signals to propagate between them.

\subsubsection{Dissolving Edge Tiles and Corner Gadgets}\label{sec:dissolving}
Care must be taken when dissolving edge tiles and corner gadgets to avoid erroneous attachments of tiles which have dissolved, but on which not all of the glues have yet deactivated. When dissolving edge and corner tiles, we use a procedure called \emph{careful dissolving} to guarantee safe detachment. To understand this procedure first note that, we make a distinction between those glues which were initially active on a tile before it attached to an assembly, which we call \emph{prior glues}, and those which activated after the initial attachment, called \emph{posterior glues}. Here we make one exception regarding the strength 2 glues present between the outermost corner tile of a corner gadget and its 3 neighboring tiles, these are classified as \emph{corner glues} and will be handled differently. Also, in addition to all of the functional glues present on an edge tile or corner gadget tile, when two edge tiles bind to each other, a strength 1 pair of glues of type $g_d$ and $g_d^*$, called \emph{dissolve helper glues}, are activated between them. Corner gadgets also have these glues activated between their tiles, but this is done in a tree-like structure with the root being the outermost corner tile. This tile shares dissolve helper glues with the 3 corner gadget tiles adjacent to it, and these share dissolve helper glues with the 3 corner gadget tiles which cooperatively bound in between, though only on 1 face each so as to form a tree.

Careful dissolving begins when a detector gadget binds to an edge or corner gadget tile. This binding initiates a \emph{dissolve signal} that propagates along the edge and corner gadget tiles, deactivating all prior glues. Now suppose $\gamma$ is a group of connected edge tiles which have detached from the assembly as a result of these deactivations. By our tile design, prior glues always only bind with with either posterior glues or bounding box glues ($g_x$ or $g_f$), and posterior glues, which are always strength 1, only bind with prior glues. Notice that $\gamma$ can be presenting at most 1 prior glue of strength 1, otherwise it would not have detached from the assembly, though it may have any number of posterior glues and some dissolve helper glues. Because attachment to an assembly requires either a prior glue of strength 2 or two prior glues of strength 2 to bind with posterior glues exposed by a bounding box assembly, $\gamma$ is effectively inert. It is possible that two detached junk assemblies have dissolve helper tiles exposed, but any cooperation between these junk assemblies would require the cooperation of a dissolve helper glue and a prior/posterior glue pair to occur. This may happen, but eventually the prior glue will deactivate and the combined junk will dissolve.

By the connectivity offered by the dissolve helper tiles, even as $\gamma$ further breaks up into smaller assemblies or individual tiles, this property is preserved, since in addition to the dissolve helper glues between each pair of tiles in $\gamma$, any glues holding tiles together form a prior/posterior pair. For a strength 1 cut to exist in $\gamma$, allowing it to break apart, it must be the case that the prior glue deactivates between the tiles, otherwise they will still be held together with at least strength 2. Eventually, we will be left with only individual inert tiles or the 4 tiles that make up the corner of a corner gadget which will also be inert by the same argument. Thus we have a maximum junk size of 4. Careful dissolving is possible so long as the above conditions regarding prior and posterior glues are met. This is true for all gadgets and tiles used during the leader election process, so during the leader election process, when we say that a dissolve signal is propagated, we mean that careful dissolving occurs between those tiles.

\begin{figure}
    \centering
    \includegraphics[width=\textwidth]{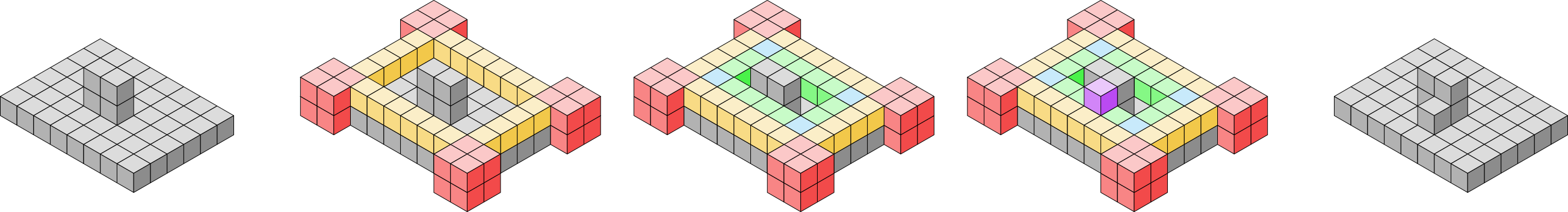}
    \caption{If, as a surface of the inner shell is growing, it is found that there are shape or filler tiles still protruding from that surface of the bounding prism, then a detector gadget will be able to cooperatively bind with the protruding tile and adjacent verification tile. The verification tile will then be converted into a filler tile and the other verification tiles, edge tiles, and corner gadgets will be dissolved. In this illustration, a verification tile is adjacent to a protrusion which is 2 tiles high. There are a few other possible configurations of verification tiles next to protrusions, each requiring a unique detector gadget, but the idea is the same in each.}
    \label{fig:verification-surface-protrusion}
\end{figure}

\subsubsection{Filler Verification}

When the edges growing from 2 corner gadgets meet via edge tiles between them along the surface of a bounding prism, signals between them through the edge tiles activate glues which allow a filler verification process to begin. This process proceeds in iterations growing inwards towards the surface's center and verifies that there are no gaps or protrusions in the surface. If gaps are found, nothing happens until those gaps are filled with filler tiles, after which the verification continues. If protrusions are found, then as illustrated in Figure \ref{fig:verification-surface-protrusion}, detector gadgets are able to cooperatively bind with a verification tile and a shape/filler tile of the protrusion. Signals from this attachment cause the verification tile to become a filler tile and cause all other involved verification tiles, edge tiles, and corner gadgets to dissolve.

The filler verification procedure is as follows. When the edge between two corner gadgets is filled with edge tiles, a signal is able to propagate between the corner gadgets. Once a corner gadget has received signals from its 2 neighboring corner gadgets, glues are activated on the adjacent edge tiles allowing the cooperative binding of a tile called a \emph{verification corner tile}. This verification corner tile attaches diagonally adjacent to the corner gadget within the region bounded by the edge tiles. Additionally, signals from the corner gadgets activate glues on the other edge tiles which allow special \emph{verification edge tiles} to cooperatively bind with the edge tile and surface of the bounding prism. If there is a gap preventing such a binding, it will simply not occur until filler tiles attach to fill the gap. If there is a protrusion, a detector gadget will be able to cooperatively bind with a filler/shape tile on the protrusion and a verification tile. That verification tile will then convert to a filler tile through signals fired from the detector gadget and all other involved edge tiles, verification tiles, and corner gadgets will dissolve. If no protrusion is found, the process repeats with the verification corner tiles acting as the corner gadgets and verification edge tiles acting as the edge tiles. A new iteration of the verification process will begin in the next inner layer of the surface.

\begin{figure}
    \centering
    \includegraphics[width=\textwidth]{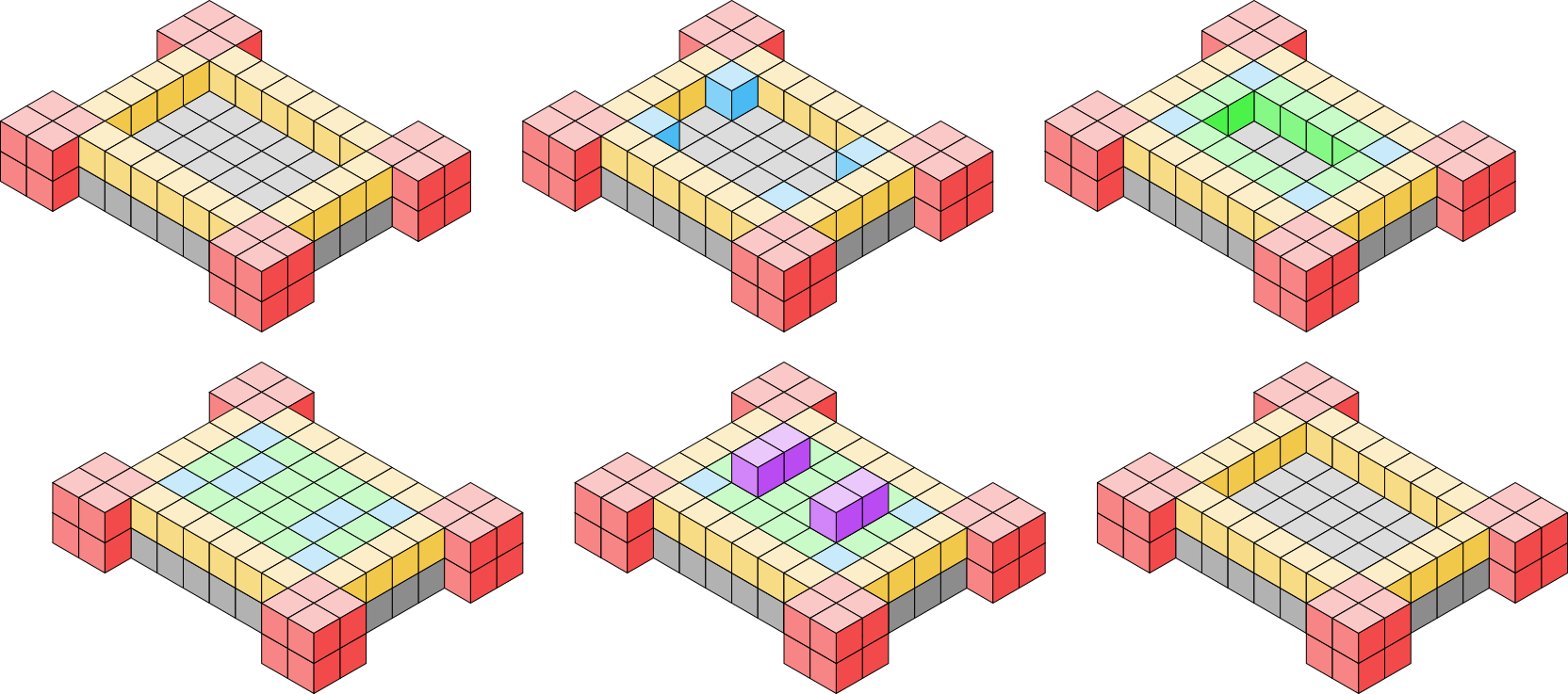}
    \caption{During the surface verification process, tiles attach within the rectangle formed by edge tiles on a surface. These tiles attach in layers growing towards the center of the shape. Once the corners of a layer are adjacent, or in the case of an odd side length when one corner touches three sides of the previous layer, a detector gadget can bind. Signals activated by this binding indicate that the verification process was successful and the verification tiles are dissolved}
    \label{fig:surface-verification}
\end{figure}

This process will continue until the center is reached. This can happen in 2 different ways depending on whether the shortest side of the surface is of even length or odd length. (See Figure \ref{fig:surface-verification}.) If the shortest side of the surface is of even length, then eventually 2 verification corner tiles will be adjacent to each other. A duple detector gadget will be able to cooperatively bind with those tiles indicating that the center has been reached. This will happen on both pairs of adjacent corner verification tiles and once the verification edge tiles attach between them, signals will be able to propagate between the pairs of corner verification tiles. These signals will propagate back along the iterations of the verification tiles and activate glues on the corner gadgets which will allow for the growth of the outer shell to begin on this face of the bounding prism. If the shortest side of the surface is of odd length, the process is similar, but instead of 2 verification corner tiles being adjacent, there will be a single verification corner which is adjacent to either 2 verification corner tiles from the previous iteration, or all 4 if the surface of the bounding prism was a square. In either case, detection gadgets will be able to initiate signals which inform the corner gadgets that verification of this face is complete. Additionally, upon completion, a dissolve signal causes all glues on the verification tiles to turn off and the verification tiles themselves to dissolve.

\subsubsection{Handling Thin Shapes}
The process described above assumes thick shapes, those whose minimum bounding box has no sides of length 1. To handle thin shapes (i.e. those shapes that are not thick), first note that for every corner gadget attached to a thin shape, there will be at least one direction where no edge tile can cooperatively attach to the corner gadget and shape assembly. This can be detected by a detector gadget and upon detection signals will be fired accomplishing 2 tasks: (1) glues will be activated on the corner gadget which allow other corner gadget tiles to attach as if two mirrored corner gadgets were overlapping along the thin edge, and (2) edge tiles running along the thin edge of the assembly from the corner gadget will be dissolved and the outgoing $g_\text{edge}^L$ or $g_\text{edge}^R$ glue from the corner gadget will be deactivated and replaced by a newly activated glue of type $g_\text{edge}^T$. We call corner gadgets that have been modified in this way \emph{extended corner gadgets}. To the glue of type $g_\text{edge}^T$, a different type of tile, called a \emph{thin edge tile}, can cooperatively attach to the assembly and corner gadget. Thin edge tiles behave similarly to regular edge tiles and grow sequentially along the assembly. Upon meeting another thin edge tile, like with normal edge tiles, a detector gadget cooperatively binds and activates glues on the thin edge tiles allowing them to bind with each other if they meet along a thin edge or converting the thin edge tiles into corner gadget tiles if they met at a corner. If the path of the thin edge tiles is blocked by a shape or filler tile, a detector gadget can cooperatively bind and the last thin edge tile is converted to a filler tile and a dissolve signal is propagated down the remaining edge tiles.

In the case where our initial shape assembly is a \emph{thin rod}, having dimensions $1\times 1\times m$, the corner gadgets which bind to the ends of the ends of the rod will be extended twice (or 3 times if $m=1$). Detector gadgets can be used to determine that a corner gadget has been extended more than once and signals from the attachment of these detector gadgets will activate the same glues on the corner gadgets indicating that filler verification is complete for the corresponding $1\times 1$ side of the assembly. 

\subsubsection{Outer Shell Construction}

Whenever the filler verification process is completed on a surface of the bounding prism, signals activate glues on the corner gadgets of that surface which initiate the growth of an outer shell. The glues activated on the corner gadgets exist on the outward pointing faces of the tiles between edge tiles and allow tiles called \emph{outer shell tiles} to bind with strength 2 to these locations as illustrated in Figure \ref{fig:outer-shell}. Once attached, these outer shell tiles present strength-1 glues of type $g_\text{out}$ on all sides except the one that points away from the assembly. Another type of tile, called an \emph{outer edge tile}, is then able to cooperatively bind to these outer shell tiles and the edge tiles from the inner shell. These outer edge tiles also present $g_\text{out}$ glues which further allow other outer edge tiles to cooperatively bind on top of the edge tiles from the inner shell. When two outer edge tiles meet along an edge, detector gadgets can cooperatively bind to the pair causing them to activate glues between each other and bind.

Additionally, special corner gadgets called \emph{outer corner gadgets} bind with 3 outer shell tiles on the corners of the assembly. (Because in our construction $\tau=2$, outer corner gadgets really only cooperatively bind with 2 of the outer shell tiles to attach, but by using sequential signaling, we can ensure that they do not propagate their signals to other outer corner gadgets until they are bound to all 3 outer shell tiles on their respective corner of the assembly.) These outer corner gadgets are different from normal corner gadgets in that they have 12 tiles as illustrated in Figure \ref{fig:outer-shell}. Once an outer corner gadget attaches, signals are propagated along outer shell and outer edge tiles to adjacent outer corner gadgets. 

\begin{figure}
    \centering
    \includegraphics[width=\textwidth]{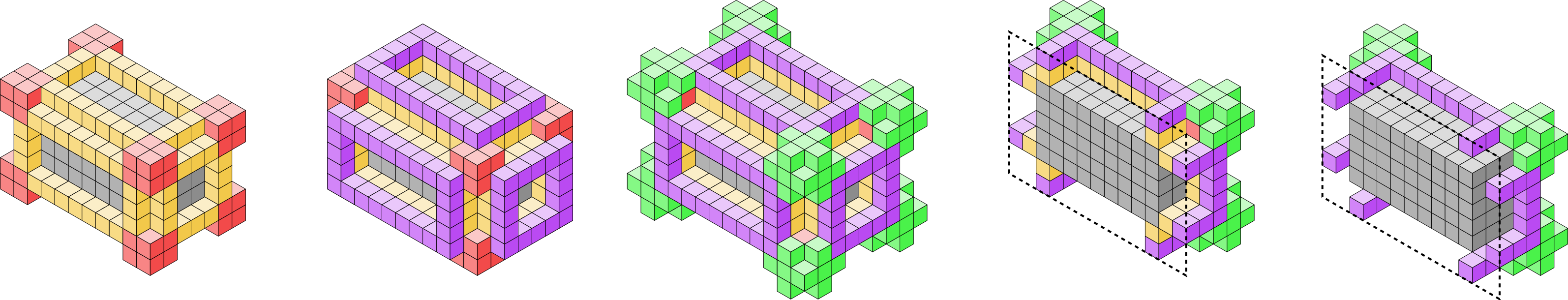}
    \caption{Once filler verification has successfully occurred on a surface of our bounding box, outer shell tiles attach to the edge tiles and corner gadgets on that surface to form a rectangle. Between the corners of these rectangles, outer corner gadgets can cooperatively bind. Once the corner gadgets have attached sufficiently to the outer shell tiles and the necessary connectivity conditions have been met, inner shell tiles are dissolved from between the outer shell and bounding box assembly. Illustrated using a cross-section view, the detachment of these tiles leaves us with a detached bounding box assembly that is too large to fit in the gaps of the outer shell, but too small to touch more than one interior corner of the outer shell simultaneously. Because of this, the bounding box assembly can then bind to an interior corner of the outer shell, but only on one corner, which is then elected leader.}
    \label{fig:outer-shell}
\end{figure}

When an outer corner gadget has received this signal from all 3 of its neighbor outer corner gadgets, a dissolve signal is propagated to the inner shell corner gadget below. This signal prompts that corner gadget and its edge tiles (but not any other corner gadgets) to dissolve and additionally causes glues, called \emph{candidate glues}, of type $g_\text{cand}$ to activate on the corners of the bounding box assembly underneath and glues of a complementary type $g_\text{cand}^*$ to activate on the interior corners of the outer corner gadgets. Because of the condition under which these signals are fired, an outer corner gadget will not signal its underlying inner shell corner gadget to dissolve until all of the outer shell corner gadgets neighbors are bound to the assembly. Consequently, even though the outer shell gadgets cause the inner shell between them and the assembly to dissolve, the outer shell will remain attached to the assembly on at least one corner until all outer corner gadgets have attached. Once the final outer corner gadget attaches however, the inner shell underneath will be able to fully dissolve and we will be left with our bounding box enclosed within but not attached to the outer shell. While the bounding box will be free to move (slightly) within the outer shell, it will be trapped inside of it due to their relative sizes.

Because the corners of the bounding box and interior corners of the outer shell have complementary glues, the corners of the bounding box assembly are able to bind to the interior corners of the outer shell; however, because the interior of the outer shell is larger than the bounding box itself, only 1 corner will be able to touch the outer shell at any given time, and thus to bind. The corner of the bounding box which happens to bind is elected leader and a special glue $g_\text{lead}$ on that corner is activated. Additionally, the binding of the bounding box assembly to the outer shell causes signals to propagate which cause the $g_\text{cand}^*$ glues on the outer shell to deactivate and then cause the outer shell to dissolve. We are then left with a bounding box with 1 corner ``elected as leader'' and containing a $g_\text{lead}$ glue from which the disassembly process can begin.

\subsection{Shape encoding}
Following the process of leader election on a bounding box, we are presented with a single corner with unique glues exposed indicating a leader tile. Here we describe the tiles of $E$ which allow for the the universal shape encoding function $f_e$ to be implemented on the shape contained in a bounding box. We use the term \emph{voxels} to reference the locations of $\mathbb{Z}^3$ in the bounding box, which may contain shape tiles, filler tiles, or no tiles as there may still be cavities within the box.

At a high level, the encoding of a shape is generated by a process which visits each voxel in the bounding box sequentially, and transfers the information of whether the voxel is inhabited by a filler tile or a shape tile to a new encoding assembly $\phi$. The set of all encodings of shapes $S = \{s_1,\ldots,s_n\}$ is $\Phi = \{\phi_1, \ldots, \phi_n\}$ where $\phi_i$ is the encoding of $s_i$ for $i=1,\ldots,n$. 
The first step in the process is for an \emph{encoding corner gadget} (see Figure \ref{fig:gadgets-encoding_gadget}) to bind to the corner elected as leader, and then construct a set of helper tiles around the bounding box.
Deconstruction is then carried out in \emph{slices}, where each slice is the set of voxels contained in a 2D subset of the bounding box. The starting voxel contains the tile elected leader (see Figure~\ref{fig:bounding-prism-slice}) and the orientation of the binding of the encoding corner gadget arbitrarily defines the orientation of the slices. For ease of explanation, once an orientation has been chosen by the attachment of the encoding corner gadget, we choose the $x$ and $y$ directions to correspond to the axes along a slice and the $z$ direction to be the axis perpendicular to $x$ and $y$ into the bounding box from the leader. Specifically, each $xy$ plane of the bounding box constitutes one slice.
The end result of the encoding process is a rectangular prism assembly of height 1 where the each tile corresponds to a unique location of the bounding box in $\mathbb{Z}^3$, and whose glues represent whether or not each location contains a shape tile (represented by a 1), or empty space inhabited by a filler tile or otherwise (represented by 0).
Additionally, information about the order in which tiles were deconstructed is included in $\phi_i$ for purposes of decoding and defining the width of a \emph{row}.
We note that the tiles in this section obey the careful dissolving property in Section~\ref{sec:dissolving}.

\begin{figure}
    \centering
    \includegraphics[width=6cm]{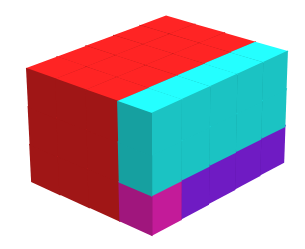}
    \caption{An example bounding box. The teal, fuchsia and purple tiles inhabit the slice of the bounding box of the $xy$ plane where $z=0$. The fuchsia tile, which was elected the ``leader'', is treated as the origin (0, 0, 0). The fuchsia and purple tiles inhabit the first row, where $y=0$. The red tiles demonstrate the remaining tiles of the bounding box. We note that these tile colors are reused in figures throughout the remainder of this section, however take on other meanings in their respective contexts.}
    \label{fig:bounding-prism-slice}
\end{figure}

\subsubsection{Creation of a deconstruction shell}
The first step of the encoding process is for an \emph{encoding corner gadget} (Figure~\ref{fig:gadgets-encoding_gadget}), similar in structure to the corner gadgets utilized in the leader election process, to bind to the leader corner.
We then treat that corner as the origin of our shape, where the directions of the $x$, $y$, and $z$ axes are shown in Figure~\ref{fig:deconstructor-shell_creation-0}.
This reference point and orientation allows us to assign coordinates to each voxel of the bounding box.
Of key importance during the deconstruction process is that the deconstructing supertile remains connected with strength 2 at all times.
It is given that the shape tiles are connected with strength 2, and filler tiles similarly connect to both shape tiles and each other with strength 2.
However, filler tiles are connected to only the 2 tiles which caused their cooperative placement and exterior filler tiles expose only strength 1 $g_f$ glues.
To ensure that during the deconstruction process no tiles prematurely disconnect from the bounding box (and to provide additional functionality during the deconstruction process), \emph{shell tiles} are added which create a shell around the bounding box and utilize the signals demonstrated in Figure~\ref{fig:fill-convex-site} to enable strength 2 connections with the exterior-most fill tiles.
At the end of the creation of the deconstruction shell (which we will also simply refer to as the `shell'), the bounding box will have all tiles on its faces covered, aside from those that are part of the first slice of the bounding box to be encoded.
The shell consists of three parts corresponding to tile types: (1) the shell base, tiles which cover one face of the bounding box, allow for communication between tiles in the shell and allow for cooperative binding of recognizer tiles (to be described), (2) shell slices, which cover 3 faces of the bounding box (aside from the tiles that are part of the first slice of the bounding box) and are removed after each slice of the bounding box is encoded, and (3) a cap, which covers the remaining face and allows for the encoding process to sense when it has completed the decoding process.

\paragraph{Shell Base Formation}

The growth of the shell base is the first step of process and is initialized from the encoding corner gadget; cooperative growth of shell base tiles begins along the $xz$ plane, demonstrated in Figure~\ref{fig:deconstructor-shell_creation-0}.
This growth is initiated by the tile of the encoding gadget in the $(0,-1,0)$ location, which activates glues on its $+x$ and $+z$ faces leading to base tiles being able to cooperatively bind with the encoding corner gadget and the bounding box.
Once bound to the shape, they activate glues similar to the encoding gadget to continue the binding until no feasible binding sites remain.

\begin{figure}
    \centering
    \includegraphics[width=4cm]{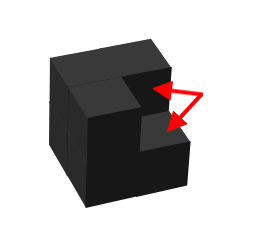}
    \caption{Encoding corner gadget utilized to bind to the elected corner. The faces with arrows pointing towards them are those which begin with glues in the \texttt{on} state, complementary to the leader election glues.}
    \label{fig:gadgets-encoding_gadget}
\end{figure}

\begin{figure}[htp]
    \centering
    \includegraphics[width=10cm]{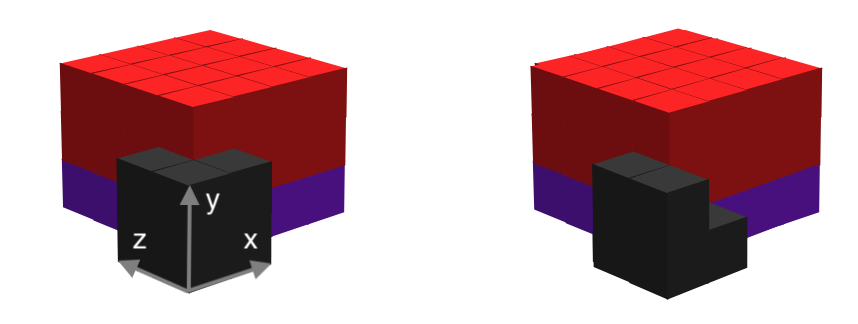}
    \caption{(Left) The encoding corner gadget (black) binding to the leader corner. Purple tiles are deconstruction shell base tiles whose growth is initiated after binding of the encoding corner gadget to the bounding box. Red tiles indicate the bounding box, comprised of both filler tiles and shape tiles. (Right) After initial binding of the encoding corner gadget to the elected corner, glues are deactivated in order to allow for the encoding process to access all voxels in the first slice of the bounding box}
    \label{fig:deconstructor-shell_creation-0}
\end{figure}

\paragraph{Shell Slice Formation}

To ensure the shell is complete before the remainder of the encoding process proceeds, the shell growth process proceeds away from the origin in the $+z$ direction only after shell slice tiles have entirely surrounded an $xy$ plane of the bounding box.
Each shell slice which grows is only a single tile wide.
The growth of the first shell slice tile is enabled by the activation of a strength 1 glue on the encoding corner gadget on the tile in the $(-1,-1,0)$ location along its $+z$ face, and with the adjacent shell base tile.
We note that this growth is initiated at the same time as the shell base tiles, however will not begin until a shell base tile is bound to the bounding box in the appropriate location.
Cooperative binding sites between the growing shell slice and the face of the bounding box allow for shell tiles to be placed in the $+y$ direction until reaching the edge of the bounding box, as shown in Figure~\ref{fig:deconstructor-shell_creation-1}.
A \emph{shell detector gadget} allows for the shell slice tiles to sense they have reached an edge between two faces of the bounding box.
For the growth of shell slice tiles to continue in the $+x$ direction along the adjacent face, a tile must be placed on the $+y$ face of the most recently placed shell slice tile - the binding of the shell detector gadget to the slice tile and a tile of the bounding box activates a strength 2 glue, allowing a second type of slice tile to bind which contains a complementary strength 2 glue, exposing strength 1 glues along all faces adjacent to tile face containing the strength 2 glue.

Shell growth continues until similarly reaching the edge in the $+x$ direction, where a shell detector gadget binds and causes the prior process to be repeated.
Growth of shell slice tiles then continues in the $-y$ direction along the face of the bounding box until overlapping with the shell base tiles; when a shell slice tile binds to a shell base tile, a message returns to the shell slice tile which initiated the growth of the current slice.
Upon sensing this message, a strength 1 glue is activated on the face of shell tile which initiated growth of the current shell slice layer in the $+z$ direction.
The shell growth process continues until reaching the exterior most slice of the bounding box and cooperative growth is no longer possible.

\begin{figure}[htp]
    \centering
    \includegraphics[width=4cm]{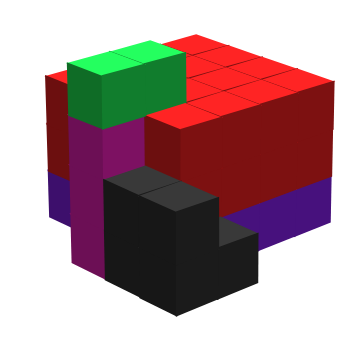}
    \caption{Shell slice tiles (fuchsia) grow along the edge of the bounding box. Growth in the $+y$ direction is initiated from the encoding corner gadget, and continues until reaching the edge of the bounding box. Green tiles are a shell detector gadget, allowing for the shell tiles to sense the edge of the bounding box and activate a strength 2 glue, causing a shell tile with a complementary glue to extend into the $+y$ direction}
    \label{fig:deconstructor-shell_creation-1}
\end{figure}

\begin{figure}
    \centering
    \includegraphics[width=4cm]{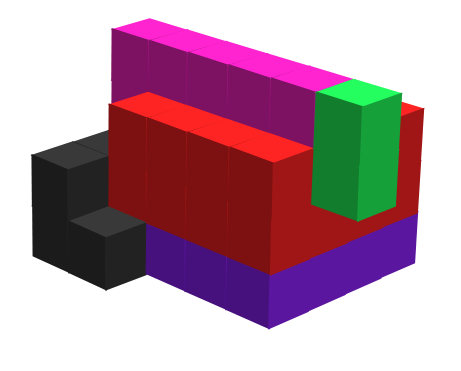}
    \caption{Growth in the $+x$ direction is no longer possible by the shell tiles (fuchsia), and the shell growth duple (green) binds, allowing for shell tiles to continue growth}
    \label{fig:deconstructor-shell_creation-2}
\end{figure}

\paragraph{Shell Cap Formation}

At this point, a 4-tile \emph{capping gadget} binds to an exposed, unique strength 1 glue exposed on the $+z$ face of outermost slice tile and either a $g_f$ or $g_x$ glue on the bounding box (Figure~\ref{fig:deconstructor-shell_creation-3}).
We note that this unique glue is activated alongside the shell slice growth glue, however geometric hindrances prevent the capping gadget from binding at any point but the edge of the bounding prism.
This gadget, similar to the shell detector gadget, causes a strength 2 glue to be activated on the outward-most shell slice tile to place a capping tile. 
This allows for a final set of capping tiles to enclose the remainder of the bounding prism; once the capping tiles complete the shell, a message is sent back to the encoding corner gadget that the encoding process can begin (Figure~\ref{fig:deconstructor-shell_creation-4}).
The encoding process begins with a signal to deactivate the glues which bind the tiles which provided geometric guidance to the encoding corner gadget and activating a new strength 1 glue, $d_{\oplus,0}$.

\begin{figure}[htp]
    \centering
    \includegraphics[width=4cm]{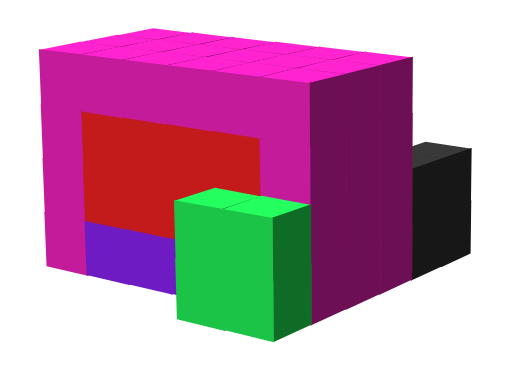}
    \caption{4-tile capping gadget in green binding to exposed shell tiles after all shell slice tiles have been added to the bounding box.}
    \label{fig:deconstructor-shell_creation-3}
\end{figure}

\begin{figure}[htp]
    \centering
    \includegraphics[width=4cm]{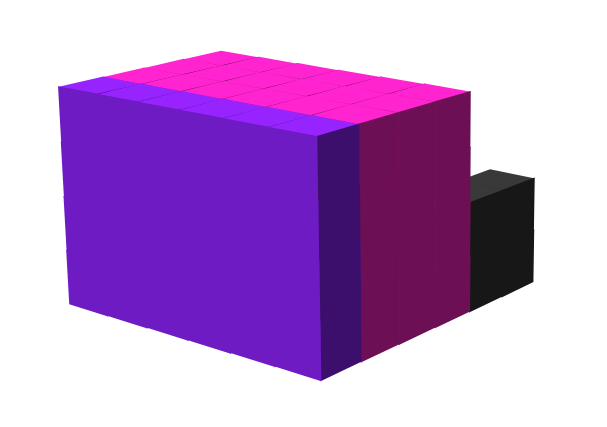}
    \caption{Capping layer fully added to the deconstruction shell}
    \label{fig:deconstructor-shell_creation-4}
\end{figure}

\subsubsection{Encoding Assembly via Bounding Box Deconstruction}\label{sec:encoding-description}
With the deconstruction shell created around the bounding box, we are now able to begin the process of building the encoding structure ($\phi$) by deconstruction.
Before continuing into the details of the encoding process, we provide a description of how the information provided by the location of tiles in a bounding box is encoded into binary values.
Beginning with the origin point $(0,0,0)$, we read the tile type information for each tile in the first row sequentially by incrementing the $x$-coordinate; for example, the second tile read is in the voxel with coordinates $(1,0,0)$.
Once all tiles in the current row have been read, we jump to the next row up. 
For example, in a $3 \times 4 \times 5 \:\: (x\times y \times z)$ bounding box shown by Figure~\ref{fig:deconstructor-encoding_order-0}, the final location in the first layer is $(2,0,0)$.
The next tile encoded is at coordinates $(2,1,0)$.
We then encode tiles heading towards the origin; the next voxel encoded in our example encoding would be $(1,1,0)$.
Upon arriving at the coordinate $(0,1,0)$ (the last of the row moving in this direction) we jump to the next row up, then encoding $(0,2,0)$.
By this process of visiting every tile in a slice in a `zig-zag' pattern, we are able to encode the information regarding any slice of a bounding box sequentially. 

\begin{figure}[htp]
    \centering
    \includegraphics[width=5.5in]{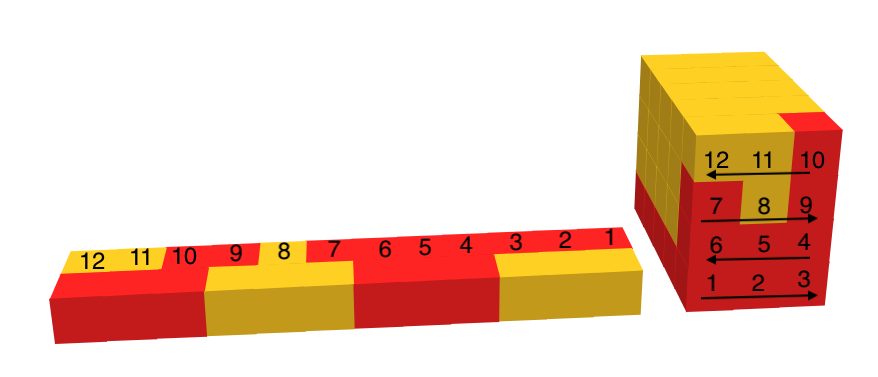}
    \caption{(Right) An example $3 \times 4 \times 5$ shape, (Left) The first two rows of its encoding assembly. The first (closest) row encodes the direction followed for each row of a slice, and the second row encodes the presence of a shape tile or filler tile in each location. Yellow tiles represent `0', red tiles represent `1'. Shape tiles and `+' direction growth are encoded as 0, fill tiles and `-' are encoded as 1. The encodings of additional slices only need a single row each, since the growth direction is shared across rows of consecutive slices.}
    \label{fig:deconstructor-encoding_order-0}
\end{figure}

The very first row of the encoding subassembly contains additional information regarding the direction of the growth in our zig-zag pattern, and as a byproduct we also are able to easily retrieve the width of the rows of tiles.
We compare the $x$ values in the coordinates $(x,y,z)$ between the first tile of a row and the last tile of a row by subtracting the $x$ value between the two such that $\Delta x = x_\text{last} - x_\text{first}$.
If a tile is contained in a row where $\Delta x > 0$ we denote this growth in the positive (`+') direction.
Alternatively, if $\Delta x < 0$ we denote this growth in the negative (`-') direction.
We encode `+' direction growth as a `0', and `-' direction growth as a `1'.
For example, in Figure~\ref{fig:deconstructor-encoding_order-0}, the first row begins growth at  tile 1, the origin $(0,0,0)$ and ends at $(2,0,0)$, leading to  $\Delta x = 2-0 = 2$.
In contrast, the second row begins at $(2,1,0)$ and ends at $(0,1,0)$, leading to $\Delta x = 0 - 2 = -2$.
We can see that the direction tiles placed in front of row 1 are encoded as 0, and encoded as 1 for row 2.
All further slices only add a single tile for each voxel, as the direction for all tiles which have the same $x,\:y$ value in their tuple $(x,y,z)$ is the same (e.g., the tile in $(1,0,0)$ which is the second tile placed in the first slice; the tile in $(1,0,1)$ is the second tile placed in the second slice).

For simplicity, the differentiation between shape and fill tiles is excluded in remaining figures in this section.

\paragraph{First Slice Deconstruction}\label{sec:decon-first-slice}

To encode the information contained in the first slice of the bounding box, one of four \emph{recognizer} tiles, $\text{rec}_0 = \{0^\oplus_0,1^\oplus_0,0^\circleddash_0,1^\circleddash_0\}$, cooperatively bind to a tile in the bounding box and the corner gadget (or the tiles added to the corner gadget, as will be shown shortly).
The recognizer tiles detect either a fill tile with glue $g_f$ or a shape tile if the glue is of type $g_x$.
We note that the activation of the $d_{\oplus,0}$ glue on the encoding corner gadget allows for only two possible tiles to bind to the origin location.
$0^\oplus_0$ tiles start with active $d_{\oplus,0}^*$ and $g_f^*$ glues on adjacent edges, $1^\oplus_0$ tiles start with active $d_{\oplus,0}^*$ and $g_x^*$ glues on adjacent edges.
The two remaining tile types are utilized for `-' direction growth.
The $\text{rec}_0$ tiles contain glues which allow for specific growth patterns unique to the first slice; the recognizer tiles for the remaining slices are demonstrated in Section~\ref{sec:encoder-remaining-slice}.

\begin{figure}[htp]
    \centering
    \includegraphics[width=6cm]{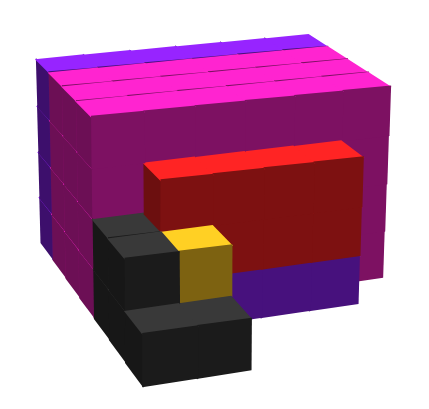}
    \caption{Binding of the first recognizer tile causes additional signals which initiate growth of tiles on the encoding corner gadget}
    \label{fig:deconstructor-first_row-1}
\end{figure}



After this binding occurs, 2 sets of signals are activated. First, the binding with the encoding corner gadget causes the activation of a strength 2 glue on the encoding corner gadget which allows for the growth of an additional layer of tiles in the $-z$ direction adjacent to the encoding gadget, shown in Figure~\ref{fig:deconstructor-first_row-1}.
Secondly, signals are sent to the face of $\text{rec}_0$ tile opposite the bounding prism which allows for growth of two \emph{messenger} tiles; a strength 1 glue is activated on the $-y$ face of the outermost tile (Figure~\ref{fig:deconstructor-first_row-2}).
Messenger tiles contain glues which allow for the recognizer tiles to pass information regarding the direction of growth and the tile type of the shape voxel which they are adjacent to.
This, along with activation of glues from the encoding corner gadget itself allows for cooperative growth of a path along the edge of the encoding corner gadget (Figure~\ref{fig:deconstructor-first_row-3}).
Once the growth of tiles reach the tile of the encoding corner gadget located at $(-1,-1,-1)$, cooperative growth halts. An \emph{encoding detector gadget} (green) is able to bind to the glue on the encoding corner gadget and the outermost encoding tiles placed due to cooperative growth.
This binding of the messenger tile with the encoding detector gadget causes the activation of a strength 2 glue which allows for binding with the first tile of the encoding shape along the $-x$ axis (this tile ends up becoming the nucleation site for decoding as well).
Once the first tile of the encoding structure is added, additional tiles cooperatively bind to the tiles of the encoding structure and the shell slice tiles (but not the shell cap tile).
This growth is visualized in Figure~\ref{fig:deconstructor-first_row-4}.

\begin{figure}[htp]
    \centering
    \includegraphics[width=4cm]{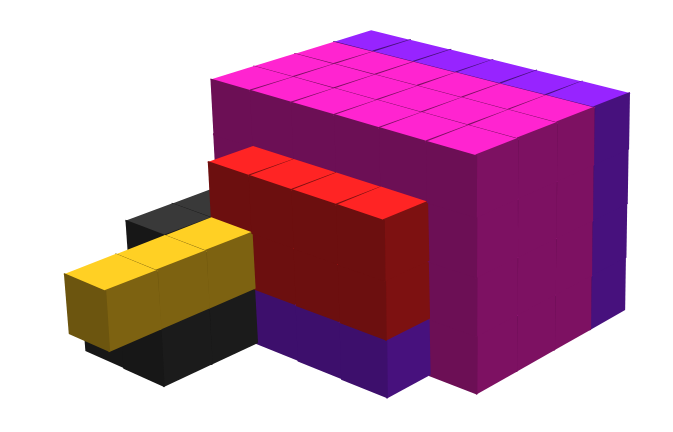}
    \caption{Two messenger tiles, uniquely mapped to the activation of $\text{rec}_0$ tiles allow for growth to extend out past the tiles of the encoding corner gadget for purposes of cooperative growth. Note that strength 1 glues are activated on 4 faces of the outermost yellow tile, as we cannot guarantee in which rotation the tile will bind}
    \label{fig:deconstructor-first_row-2}
\end{figure}

\begin{figure}[htp]
    \centering
    \includegraphics[width=8cm]{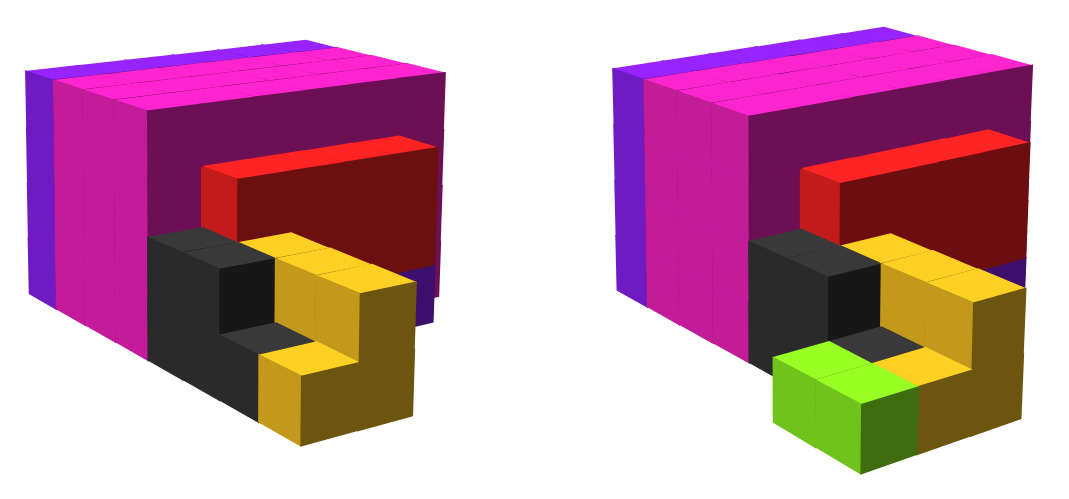}
    \caption{(left) Enabled by the outwards growth of the recognizer tiles shown in Figure~\ref{fig:deconstructor-first_row-2}, tiles are able to cooperatively grow outwards. (right) An encoding detector gadget (green) can then attach to exposed glues from the recognizer tile growth and the encoding corner gadget, allowing for both the encoding corner gadget and recognizer tiles to `sense' that we have reached the outermost edge}
    \label{fig:deconstructor-first_row-3}
\end{figure}

\begin{figure}[htp]
    \centering
    \includegraphics[width=8cm]{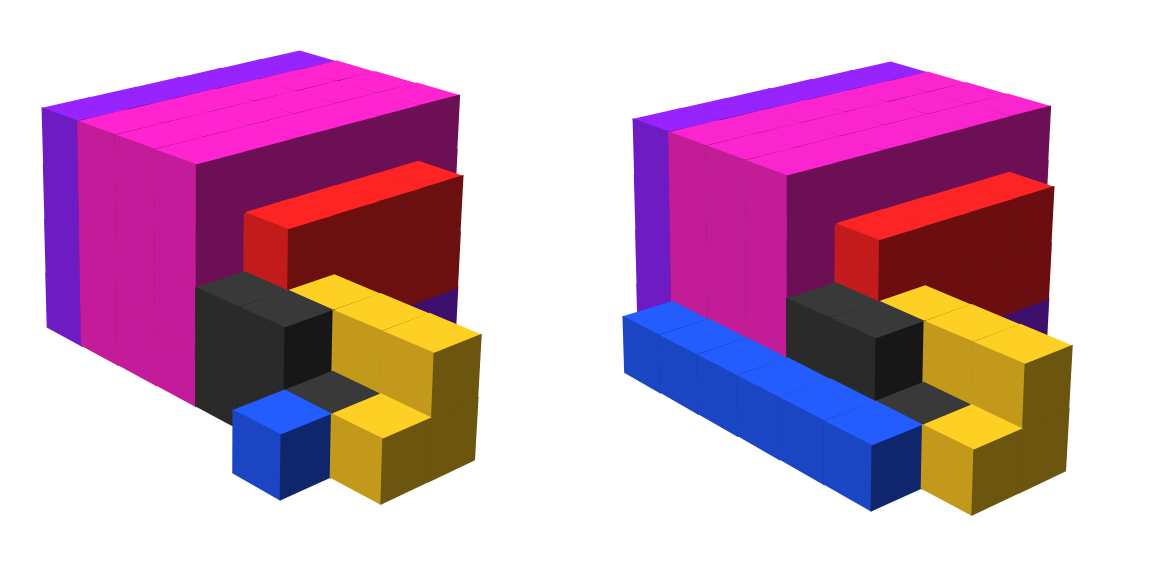}
    \caption{(Left) The first tile of the encoding structure (blue) is bound to the encoding corner gadget, (Right) cooperative growth of tiles with the first row of shell tiles}
    \label{fig:deconstructor-first_row-4}
\end{figure}

After the encoding structure tile attaches to the encoding corner gadget, the first tile of the encoding structure exposes a strength 2 glue along its $-z$ face, allowing for binding of a messenger tile which redirects growth in the $+y$ direction.
Three more tiles are added in succession - a helper tile with a strength 2 glue to allow for growth in the $+y$ axis, a directionality encoding tile and a 0/1 encoding tile.
The three tiles are placed in this order, growing in the $+z$ direction as pictured in Figure~\ref{fig:deconstructor-first_row-5}.
We have now encoded the information of the tile type which inhabits (0, 0, 0), along with the direction of growth.
Once the 0/1 encoding tile and the directionality encoding tile bind to the encoding structure, a message is passed backwards through the messenger tiles towards the recognizer tile, deactivating glues and turning into size 1 junk (i.e., dissolving the tiles) as the message propagates along the edge of the encoding corner gadget.
The purpose of deactivation is to allow for reuse of the same path along the encoding corner gadget.
This leaves only the tiles on the encoding structure, and the first messenger tile which attached to the recognizer.
Upon reaching the recognizer tile, it exposes a glue in the $-y$ axis to signal to the encoding corner gadget that this recognizer has successfully encoded its adjacent voxel.
After binding with this glue, the encoding corner gadget signals for the addition of two tiles in the $+x$ direction (using glues on the shell tiles of the $xz$ plane for cooperation) which activate the $d_{\oplus,0}$ glue in the $+y$ direction, allowing for the next recognizer tile to be placed.
The prior process is then repeated, which then creates a series of message tiles to grow back to the encoding structure (Figure~\ref{fig:deconstructor-first_row-6}), using many of the same voxel locations.
Additionally, an $r$ glue is activated on the recognizer tile's direction of growth (in this case, $+x$) in order to allow for recognizer tiles to detect when they need to activate $d_{\oplus,0}$ or $d_{\circleddash,0}$ glues.

\begin{figure}[htp]
    \centering
    \includegraphics[width=8cm]{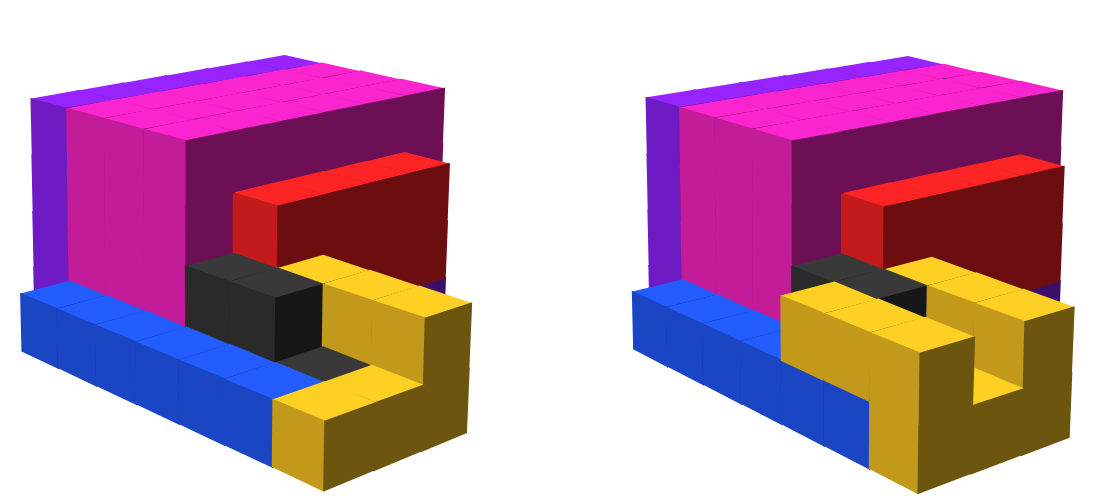}
    \caption{(left) A messenger tile binds to the first row of the encoding structure, activating a strength 2 glue to allow for cooperative placement of the encoded direction and tile type. (right) the first tile placed on the encoding structure is an encoding of direction, and the second is the tile encoding the type of the tile (i.e. shape or filler)}
    \label{fig:deconstructor-first_row-5}
\end{figure}

\begin{figure}[htp]
    \centering
    \includegraphics[width=8cm]{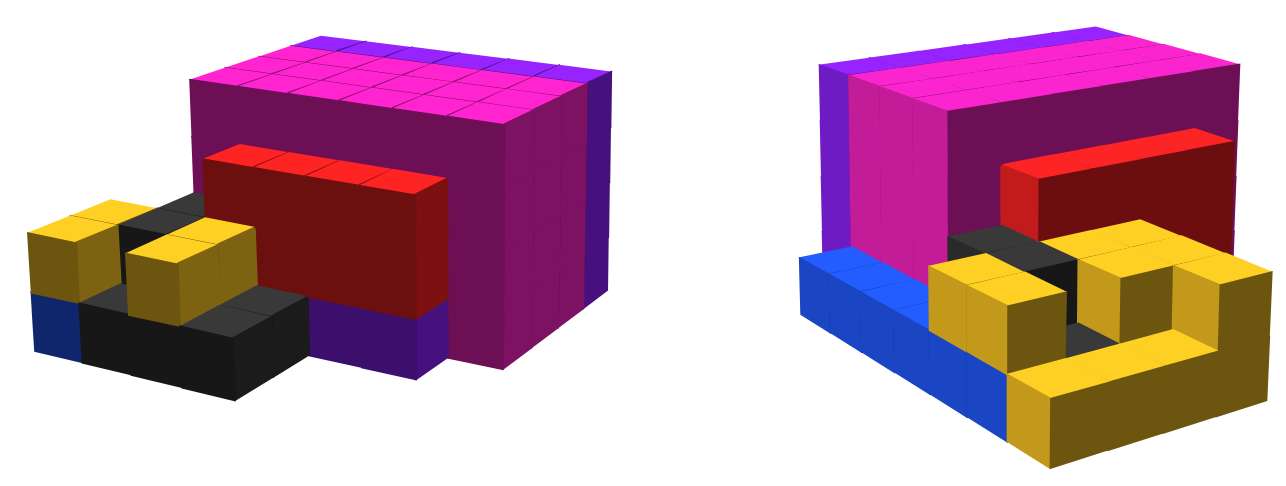}
    \caption{(Left) Resulting structure after deconstruction of messenger tiles, (Right) Addition of next tile in shape reuses the edge alongside the corner gadget for cooperative growth}
    \label{fig:deconstructor-first_row-6}
\end{figure}


This process repeats until recognizer tiles have encoded all information of the first row of the shape.
Once the final tile of the row has been placed, there exists no tile for which the tiles which extend the encoding gadget can bind to.
Instead of cooperative binding allowing for the addition of a recognizer tile, a \emph{row completion gadget} binds to the $r$ glue exposed and either a fill or shape tile.
The tile which bound to the row completion gadget activates a $d_{\circleddash,0}$ glue which allows for cooperative binding with the row above after the $r$ glue is bound, as shown in Figure~\ref{fig:deconstructor-first_row-7}.
Since the first row is `+' direction, the row growth then changes to `-' direction .
We note that 2 different versions of this row completion gadget exist to terminate `+' and `-' direction growth - the glues present are the same, but the glue locations are mirrored.
Upon binding of the `+' direction recognizer tile, the row completion gadget detects the type of tile above the row completed by activating a glue in the $+y$ direction and the $-x$ direction.
This allows for the binding of a \emph{row detector gadget} if an additional row needs to be encoded.
We will describe the case where the row detector gadget is unable to bind shortly.
If the row completion gadget senses an additional row due to the binding of a row detector gadget, the $r$ glue holding the row completion gadget to the direction `0' tile then deactivates, leaving it free to dissolve.
Message tiles mapped to the `-' direction recognizer tiles (teal) allow for expanding of the encoding structure similar to the first row and `+' direction recognizer tiles; a recognizer tile binds to a tile on the bounding box, messenger tiles allow for the growth of a path of tiles along the edge of the encoding gadget and then extend the encoding gadget and encode both the direction of growth and tile type.
Figure~\ref{fig:deconstructor-first_row-8} demonstrates this process, along with cooperative growth on top of the prior row.

\begin{figure}[htp]
    \centering
    \includegraphics[width=10cm]{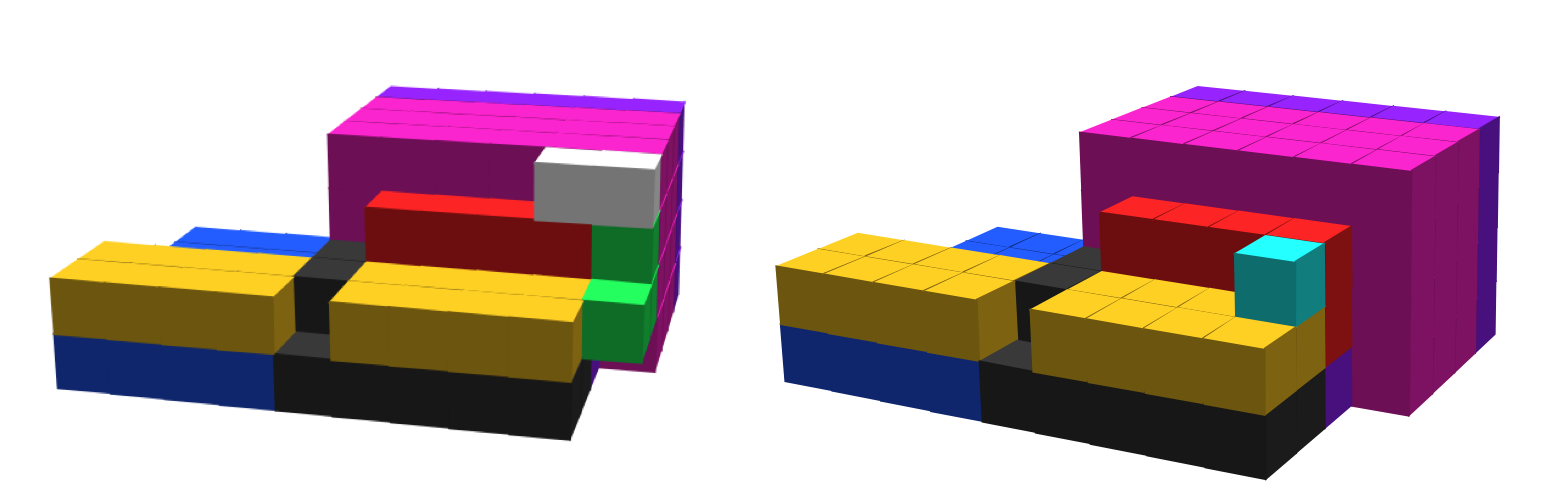}
    \caption{(Left) Row completion gadget (green) binds to supertile upon completion of the encoding of the first row. Row detector gadget (white) indicates to the detector gadget that an additional row needs to be encoded. (Right) Signals allow for growth to continue with a recognizer tile of direction `1'.}
    \label{fig:deconstructor-first_row-7}
\end{figure}

\begin{figure}[htp]
    \centering
    \includegraphics[width=10cm]{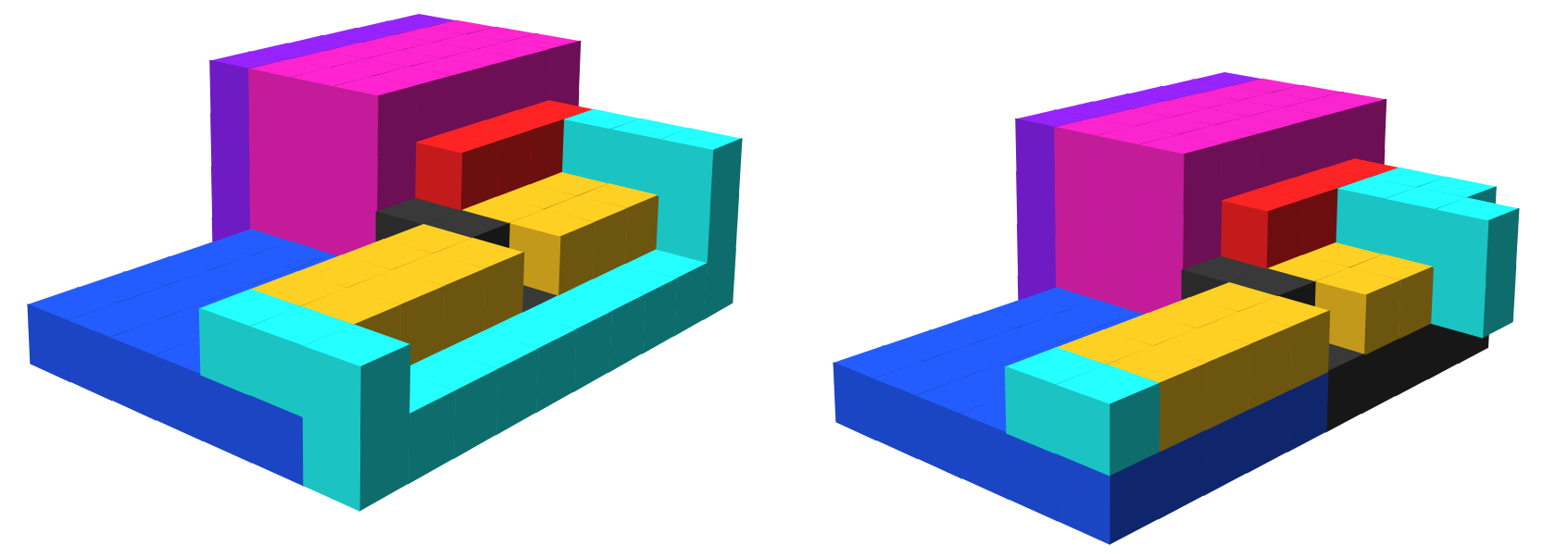}
    \caption{(Left) Growth of direction `1' messenger tiles directly mimics that of direction `0'. (Right) Direction `1' tile recognition occurs in the opposite direction}
    \label{fig:deconstructor-first_row-8}
\end{figure}

At some point, a row completion gadget will bind to a location where there exists no row above the previously encoded row.
This condition indicates that the slice has been completely encoded.
To detect this situation the row completion gadget has a glue which allows for cooperative binding of a \emph{slice completion gadget} only if the topmost tile of the gadget is exposed; this only occurs in the situation illustrated in Figure~\ref{fig:deconstructor-first_row-9}.
After binding of the slice completion gadget, the gadget activates a glue in the $+z$ direction that, when binding to complementary glues on the shell tiles, sends messages which dissolve (1) the shell in the next slice, (2) the recognizer tiles of the current slice, and (3) the slice of the shape itself.

\begin{figure}[htp]
    \centering
    \includegraphics[width=10cm]{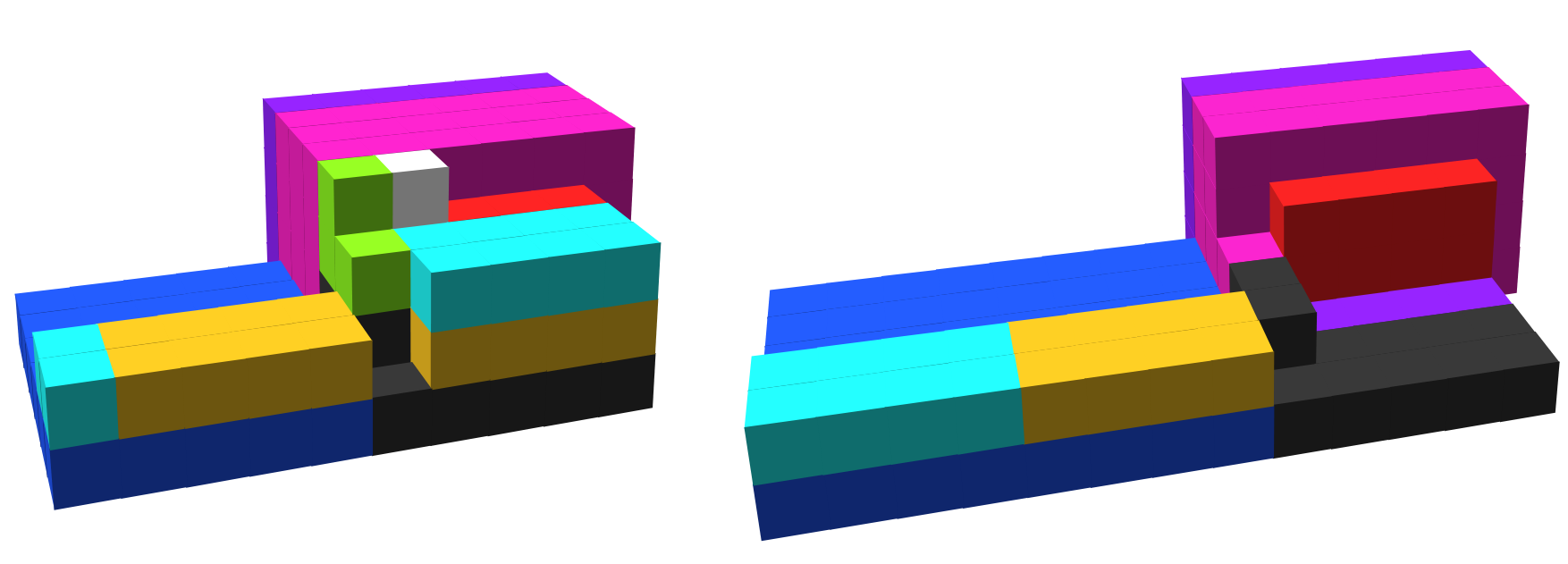}
    \caption{(Left) The row completion gadget has its topmost tile above the topmost row of the current slice, allowing for the slice completion gadget (white) to bind to the row completion gadget to indicate slice has been fully encoded. (Right) Beginning state of next slice growth after all tiles involved in encoding the current slice are turned into size 1 junk.}
    \label{fig:deconstructor-first_row-9}
\end{figure}

\paragraph{Remaining Slice Deconstruction and Termination}
\label{sec:encoder-remaining-slice}
After the encoding of the first slice has completed, we must then deconstruct the remaining slices with similar, but slightly modified dynamics.
This is due primarily to the fact that the encoding structure now contains directionality information, which remains constant across slices.
Instead of growing along the edge of the encoding corner gadget and the encoding structure, messenger tiles grow `over' themselves - they stay in the same $xy$ plane.

We add a new set of tiles $\text{rec} = \{0^\oplus,1^\oplus,0^\circleddash,1^\circleddash\}$ which allows for modified message tile growth in order to encode voxel information on the encoding structure.
We note that the base fill tiles expose glues complementary to these tile types to allow for cooperative binding of rec tiles of type $d_\oplus^*$ (as they are responsible for first row growth, which is in the `+' direction).
This allows for tiles of type $0^\oplus$ or $1^\oplus$ to bind to the first row, depending upon the tile in the slice (i.e. if its a shape or filler tile).
The growth dynamics of the messenger tiles differ significantly from the messenger tiles which are mapped to the $\text{rec}_0$ tiles.
As demonstrated in Figure~\ref{fig:deconstructor-intermediate_rows-0}, for `+' growth recognizer tiles a strength 2 glue activates to bind a messenger tile to the recognizer tile in the $+y$ direction.
Strength 1 glues are activated on all faces of this messenger tile to allow for cooperative binding of additional messenger tiles to continue in the $-x$ axis towards the encoding structure.
Once the messenger tile can no longer cooperatively bind to the encoding corner gadget, a \emph{messenger detector gadget} is able to attach to the messenger tile and the encoding corner gadget, activating signals allowing the growth of messenger tiles to place a tile encoding on the encoding structure.
After placement of this encoding tile on the encoding structure, a message is returned to the recognizer tile indicating that the tile has been encoded, allowing for messenger tiles to dissolve and signal to the base tile that encoding is complete, activating a glue to have its neighbor turn \texttt{on} a $d_\oplus^*$ glue.
This process continues across the first row, as shown in Figure~\ref{fig:deconstructor-intermediate_rows-1}.

\begin{figure}[htp]
    \centering
    \includegraphics[width=10cm]{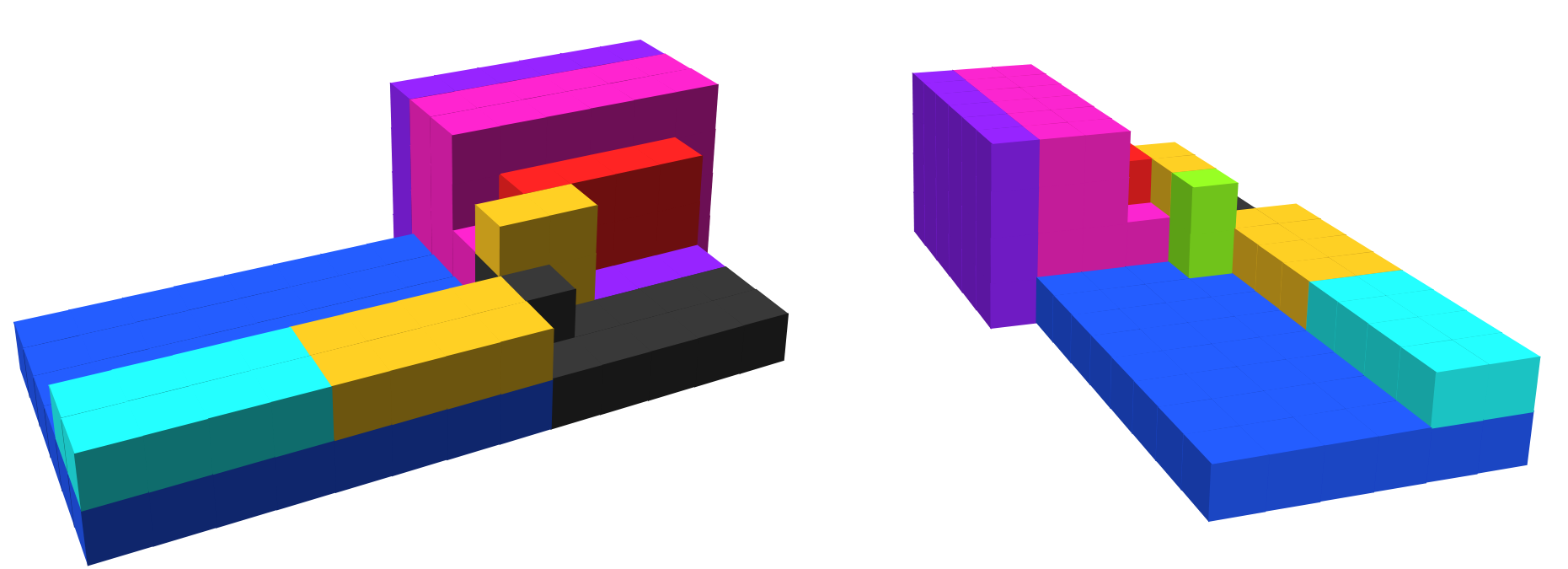}
    \caption{(Left) Direction `0' growth requires the ability to grow over previously placed tiles. (Right) Similar to the growth of the encoding structure, we require a messenger detection gadget (green) to enable the messenger tiles to sense when they have grown to the edge of the current encoding.}
    \label{fig:deconstructor-intermediate_rows-0}
\end{figure}

\begin{figure}[htp]
    \centering
    \includegraphics[width=6cm]{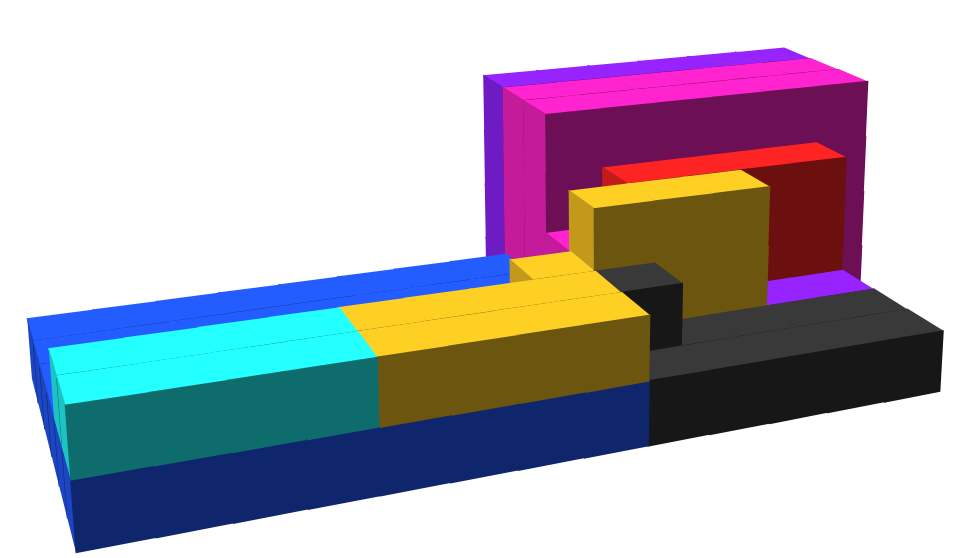}
    \caption{The second recognizer tile binds to bounding box, causing growth in the $-x$ axis to place an encoding tile on the encoding structure.}
    \label{fig:deconstructor-intermediate_rows-1}
\end{figure}

At the end of the growth of a row, we use the alternate form of the row completion gadget (i.e., glues present on $+x$ face of gadget, instead of $-x$) utilized in Section~\ref{sec:decon-first-slice} to sense the completion of a row by binding to the last recognizer tile and the bounding prism.
This causes the recognizer tile which bound to the row completion gadget to activate a $d_\circleddash$ glue in the $+y$ axis, allowing for the reversal of growth direction (Figure~\ref{fig:deconstructor-intermediate_rows-2}).

\begin{figure}
    \centering
    \includegraphics[width=10cm]{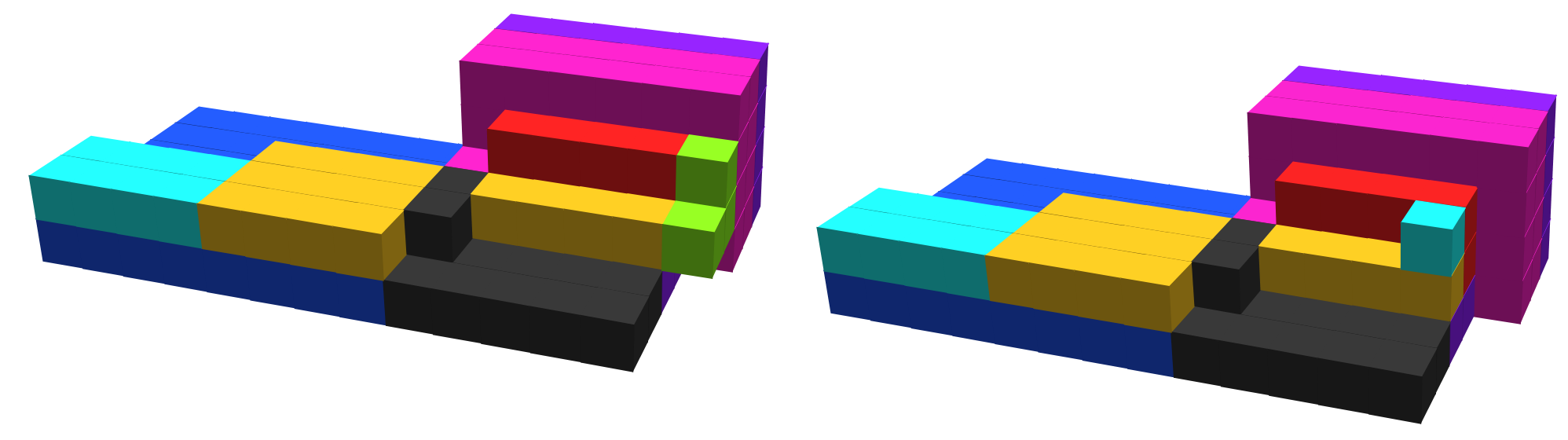}
    \caption{(Left) After the last tile in the row has been successfully encoded, a row completion gadget (green) is able to bind and enable the activation of a $d_\circleddash$ glue. (Right) The first negative direction tile (teal) binds to the top of the last recognizer tile of the prior row.}
    \label{fig:deconstructor-intermediate_rows-2}
\end{figure}

The `-' direction recognizer tile is able to utilize only cooperative binding to place its messenger tiles (instead of relying on a strength 2 glue to grow in the $+y$ axis first) in the $-x$ axis, cooperatively growing along the top of the prior row.
This process continues until the binding of a messenger detection gadget, resulting in a placement of a tile on the encoding structure (Figure~\ref{fig:deconstructor-intermediate_rows-3}).

\begin{figure}[htp]
    \centering
    \includegraphics[width=6cm]{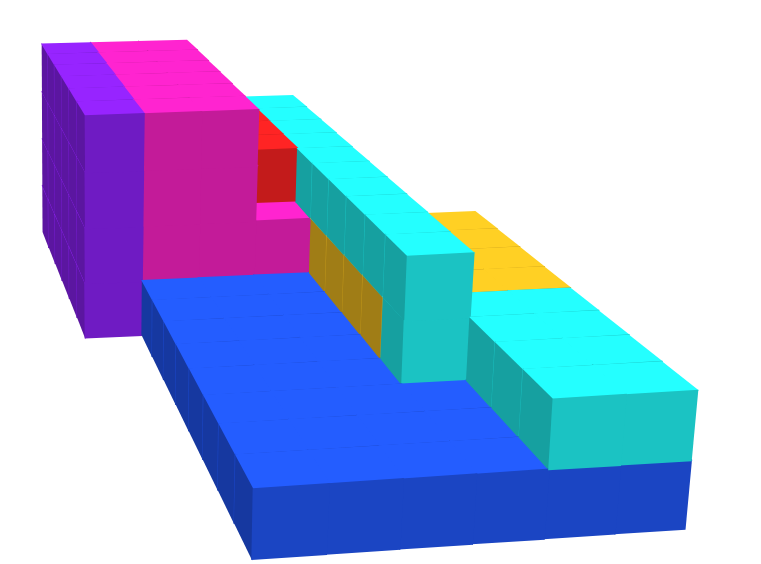}
    \caption{After the detection gadget binds, the negative direction tile messengers (teal) place a tile on the encoding structure.}
    \label{fig:deconstructor-intermediate_rows-3}
\end{figure}

Once a row completion gadget binds to the final recognizer tile along with a \emph{slice completion gadget} (see Figure~\ref{fig:deconstructor-intermediate_rows-5}), the tiles which comprise the current slice of the bounding prism, its recognizer tiles and the shell of the next slice are all dissolved.
We note an edge case where a voxel may be missing a tile from the bounding prism generated (see Figure~\ref{fig:deconstructor-intermediate_rows-6}).
This case arises in situations where there exists either some width 1 cavity (similar to the bent cavity in Figure~\ref{fig:bent-cavity}) and the binding of a filler tile blocks diffusion for other filler tiles, or an in enclosed cavity which is unreachable by filler tiles before deconstruction.
Since this encoding tileset also includes the tiles which generate the bounding prism, there exist filler tiles present to be attach into such a location.
As cooperative binding is required between a face of the bounding prism and a face of either a base tile or recognizer tile, the encoding process will not progress until a filler tile attaches to that location and a $g_f$ glue is exposed (Figure~\ref{fig:deconstructor-intermediate_rows-6}). 

\begin{figure}[htp]
    \centering
    \includegraphics[width=6cm]{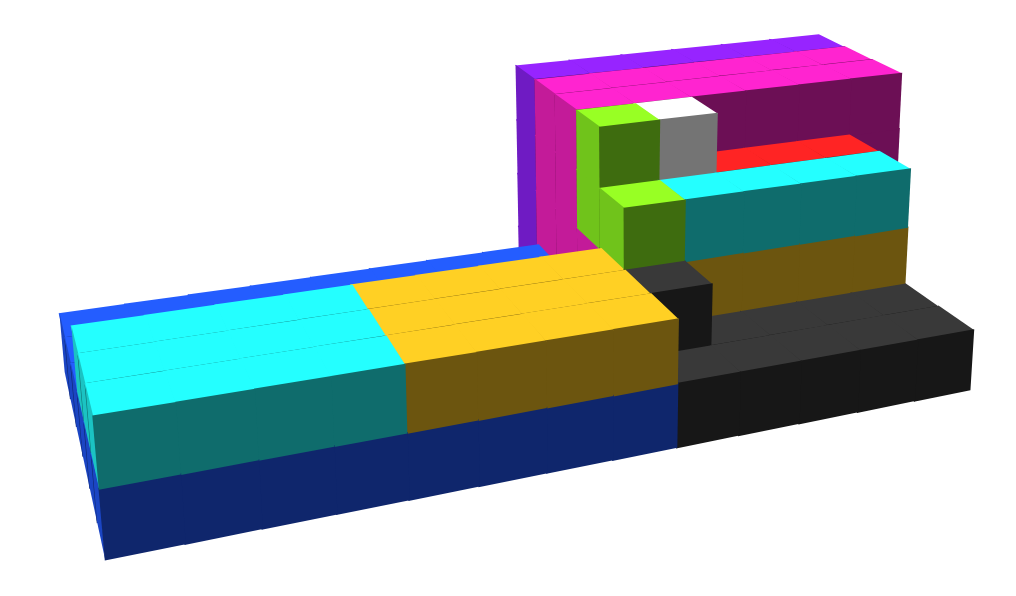}
    \caption{Slice completion gadget (white) binds after row completion gadget binds to the final row of a slice, identical to the process for first slice.}
    \label{fig:deconstructor-intermediate_rows-5}
\end{figure}

\begin{figure}[htp]
    \centering
    \includegraphics[width=13cm]{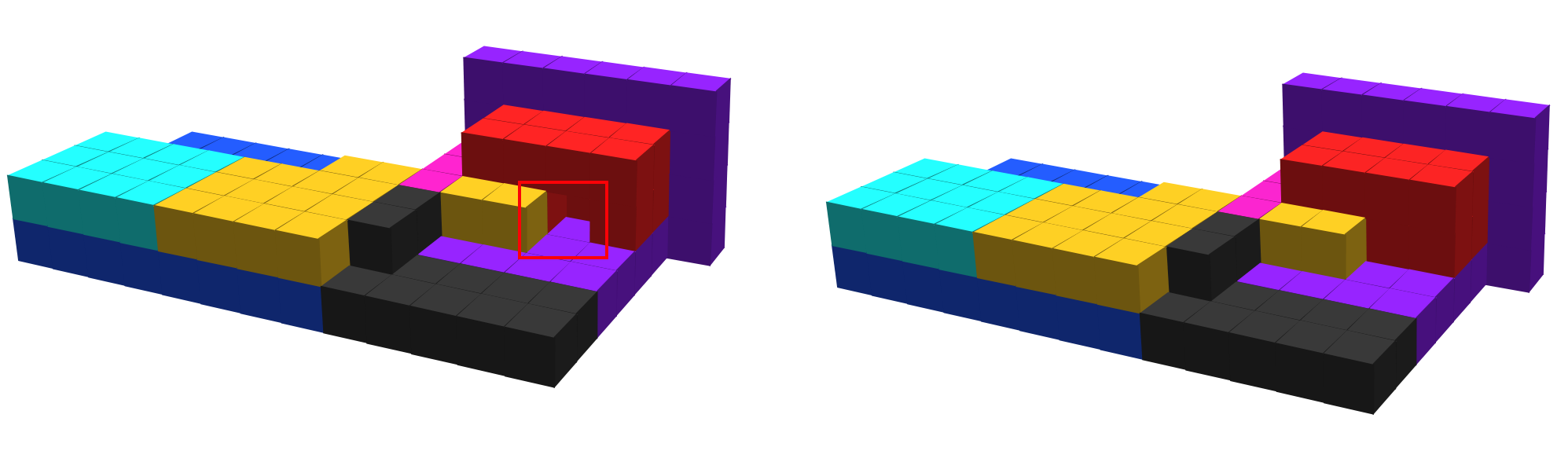}
    \caption{(Left) A tile missing from the bounding prism undergoing the encoding process, highlighted by a red box. We note that this exact void location would not be possible in a valid bounding prism, however it is presented for explanatory value. (Right) Encoding halts until a filler tile binds in the void, ensuring that encoding process does not skip a voxel.}
    \label{fig:deconstructor-intermediate_rows-6}
\end{figure}

Once this process reaches the final slice, we end up with an exposed set of tiles in the bounding prism which are able to be encoded utilizing the same mechanics as any other intermediate slice.
The key difference is that instead of slice shell tiles being exposed in the $+z$ axis, the next set of exposed tile are those of the capping layer.
The encoding process proceeds as normal, including the binding of the row completion gadget and slice completion gadget as seen in Figure~\ref{fig:deconstruction-final_row-0}.
After the capping tiles bind with the row completion gadget indicating that the final slice has been encoded, in addition to the slice, messenger and recognizer tiles dissolving into size 1 junk, a cascade of signals is sent outwards from the capping tiles to dissolve the remainder of tiles involved in the encoding process.
This includes the base, remaining slice tiles, capping tiles, the encoding corner gadget, and the encoding structure upon which messenger tiles placed the encoding of the shape.
Upon receiving the dissolve signal, we note that the encoding corner gadget sends a signal to the first tile in the encoding structure which encodes a voxel (i.e., it is set back from the direction row of the encoding structure).
The complement of this glue is found on all tiles in the encoding of the first slice, however only this outermost tile has this glue exposed. 
This signal causes a strength 2 $g_0$ glue to be activated, allowing for a location for the decoding process to begin.

\begin{figure}[htp]
    \centering
    \includegraphics[width=5.5in]{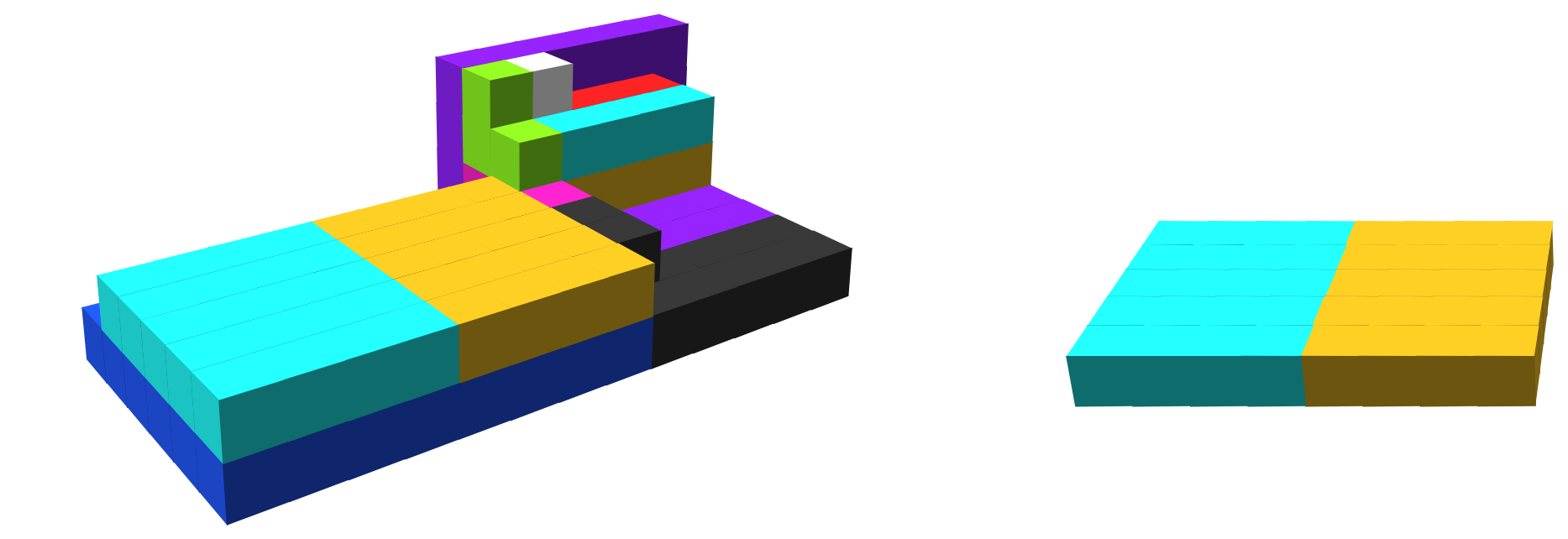}
    \caption{(Left) The final slice after encoding has completed - the binding of the row completion and slice completion gadgets (green and white, respectively) activate glues to signal to the capping layer that encoding is complete. (Right) At the end of the dissolution of all ``helper'' tiles, all that remains is the rectangular prism of depth 1, with a glue encoding the location of each voxel of the input shape and a strength 2 glue indicating the first tile in that encoding, plus a set of disconnected junk tiles.}
    \label{fig:deconstruction-final_row-0}
\end{figure}

Beginning with the creation of a bounding box and leader election around a uniformly coated shape $s$ in Section~\ref{sec:bounding-leader}, at the end of the assembly sequence for the tileset $E$ we have produced a terminal supertile $\phi$ which represents an encoding of the the shape using the encoding function $f_e$, with a maximum junk size of 4.
The STAM$^R$ system $\mathcal{E}_S = (E, \Sigma_S, \tau=2)$ finitely completes, as each of the sub-constructions to carry out the encoding $f_e$ require a finite number of steps (and thus, finite tile count) to complete. 
The final property which must hold is that regardless of the number of distinct shapes of input assemblies, the shapes of all will be correctly replicated.
By our construction, there are never exposed glues on the surfaces of any pair of assemblies that each contain an input assembly that would allow them to bind to each other.
Since junk assemblies produced by any assembly sequence are also unable to negatively interact with other assemblies, a system whose input assemblies have multiple shapes will behave simply as the union of individual systems which each have one input assembly shape, creating terminal assemblies of all of (and only) the correct shapes.
This proves Lemma~\ref{lem:encoder}.

\subsection{Shape Decoding}
We now describe the tileset $D$ which functions as a universal shape decoder.
The STAM$^R$ system for shape decoding is defined as $\mathcal{D}_\Phi = \{D,\Sigma_\Phi,\tau=2\}$.
$\Sigma_\Phi$ includes infinite copies of the tiles of $D$ and the set of encoding structures generated from $\mathcal{E}_S$, $\Phi = \{\phi_1,\ldots,\phi_n\}$.
The shape decoding process and tile types required can be broken into 3 main sets of tiles. We describe the process for a single $\phi\in\Phi$ and note that the process proceeds identically for each encoding simultaneously.
First, \emph{base tiles} initiate the decoding process by binding to $\phi$ at a unique starting location.
They then grow a subassembly outward from the encoding assembly which is guaranteed to be connected to it by at least strength 2 throughout the decoding process.
Second, we construct the \emph{shape} and \emph{filler} tiles (which are unique to the decoder's tileset, and separate from the similarly named tile types of the encoder) and describe how encoded information allows for an assembly sequence of shape tiles guaranteed to be connected to their neighbors in the decoded shape.
Third, we have a set of tiles called \emph{decoder tiles} which read the encoding and allow for the sequential placement of shape and filler tiles based on their location in the encoding.
Similar to the concerns regarding the shape becoming disconnected and splitting into multiple disconnected assemblies in the deconstruction process, decoding must proceed in a manner that allows for the growth of a slice which guarantees strength 2 connection to the encoding structure and growing shape, and also prevents a filler tile from becoming `trapped' in an enclosed volume.
The prevention of filler tiles becoming `trapped' in an enclosed volume drives a significant portion of the complexity of this process when combined with the need for strength 2 attachment of all shape tiles at steps in an assembly sequence.

In the tileset $D$, we use a decoding process which places tiles in the exact same order as the encoding process built the encoding assembly $\phi$ as presented in Section~\ref{sec:encoding-description}.
Two pieces of information are explicitly encoded in $\phi$.
The bulk of the tiles in the encoding correspond to identifying if a location in a shape corresponds to empty space, or a tile of the shape.
The second piece of information, provided in the first row of the encoding, is the the direction of growth; this can be utilized in two manners.
First, the direction of growth provides to the system the types of tiles to be utilized to reach the point encoded, as growth processes vary significantly between `+' direction growth (encoded as a 0) and `-' direction growth (encoded as a `1').
Secondly, when the direction of growth encoded changes from 0 to 1 or 1 to 0, this indicates to the system when a tile is to be placed into a new row.
This information is required to ensure that we can grow a slice such that each tile is guaranteed to be connected to its neighbor, but also so tile faces are assigned with the appropriate glues.
We note that the tiles in this section obey the careful dissolving property in Section~\ref{sec:dissolving}.

We first present the details of tile attachment.

\subsubsection{Fill and Shape Tile Attachment Details}\label{sec:shape-attachment-details}

In this section, we demonstrate a template for tiles which allows for the decoding process to be carried out, allowing for connections between all shape tiles and their neighbors within a slice.
Additionally, we provide examples of gadgets which allow for the growth of consecutive slices of a shape encoding without causing filler tiles (which are not part of the final shape, but may be temporarily required to ensure a strength 2 connection between portions of shapes where a cut may exist in the binding graph of the partially decoded shape) to be stuck in an enclosed volume of a shape.
At a high level, these tiles ensure three properties: (1) each tile is, at a minimum, connected to its neighbor in an encoding, (2) shape tiles are connected to all adjacent shape tiles with strength 2, and (3) before any tile is added to a new slice, if the tile in the same $x, y$ coordinates of the prior slice is a filler tile, that filler tile must be removed before placement of the next tile occurs.
While we demonstrate how these properties are carried out in the current section, we prove their correctness in Section~\ref{sec:universal-decoding}.

\paragraph{Tile Type Identification}
We demonstrate the filler tiles required to carry out the decoding of a shape, based upon the requirements for incrementally building a slice utilizing the `zig-zag' process.
First, each filler tile has 6 variants to handle growth along a row (also called `normal row growth') and the change of a row for both directions of growth (see Figure~\ref{fig:decoder-decoding_analysis-0}).
The two tiles of normal row growth (either $+x$ or $-x$ direction) are typically used for the majority of growth.
There exist two tiles which either grow in the $+y$ direction or turn $+y$ direction growth into $+x/-x$ direction growth; this leads to four total tiles when considering both directions of growth.
Shape tiles have 12 variants to also account for the type of tile of its neighbor in the previous slice.
See Figure~\ref{fig:decoder-decoding_analysis-1} for examples of the signals necessary on shape tiles which must bind to a shape tile in the prior slice.
To determine which of these 18 possible tile variants is utilized in any given voxel, the assembly sequence of the tileset $D$ takes information from a variety of sources - direction tiles, decoding tiles, direction change detectors, and neighbor detection gadgets.

\begin{figure}[htp]
    \centering
    \includegraphics[height=6cm]{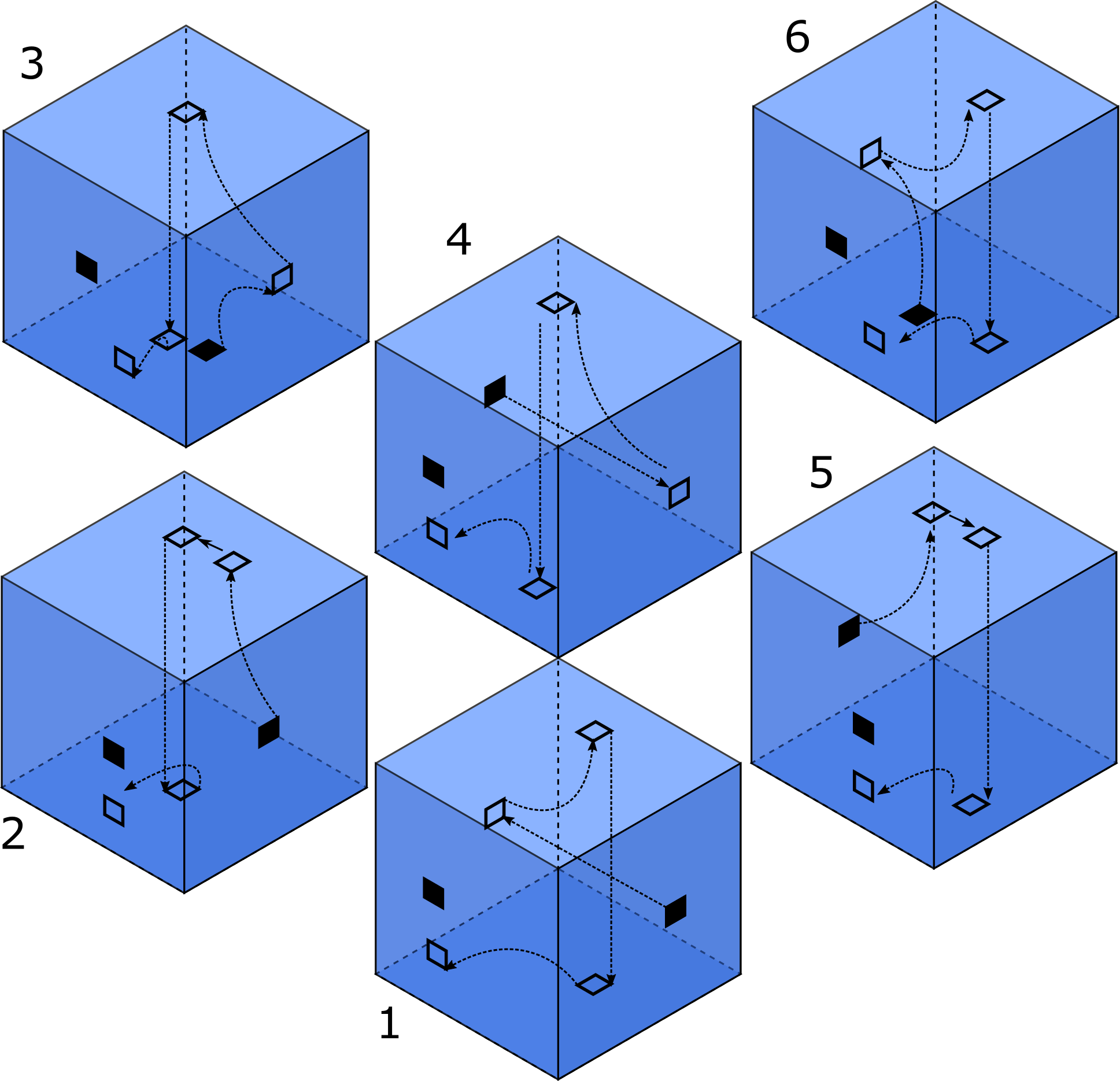}
    \caption{An example of normal row growth and direction change tiles used by the decoding process to build a slice - these tile types map to both shape and fill tiles. (1) and (4) are standard row growth tiles for `+' and `-' direction growth, respectively. (2) and (5) are row end tiles for `+' and `-' direction growth; they open cooperative binding sites which allow for tiles (3) and (6) to bind and change the direction of growth. Signal activation arrows demonstrate the order in which faces of shape tiles are determined to be either bound to a neighboring shape tile or have a fill tile adjacent to the face.}
    \label{fig:decoder-decoding_analysis-0}
\end{figure}

\begin{figure}[htp]
    \centering
    \includegraphics[height=4cm]{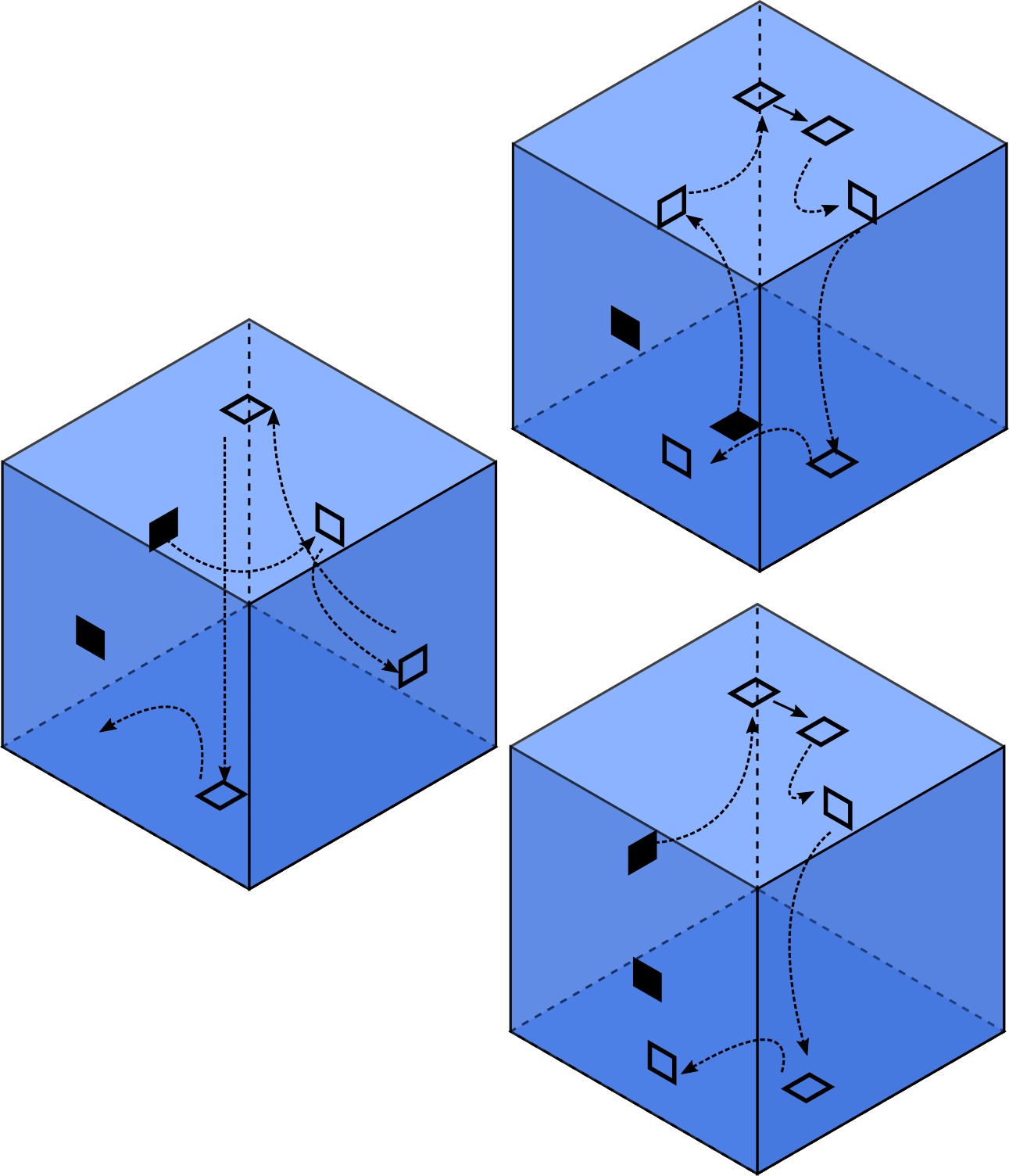}
    \caption{An example of shape tiles which have shape tiles as neighbors in the prior slice. Note that signals must pass through the face adjacent to the neighbor in the prior slice before binding to the next tile in the slice.}
    \label{fig:decoder-decoding_analysis-1}
\end{figure}

We utilize Figure~\ref{fig:generator-tile_specifics-0} to analyze the tile types which contribute information to the determination of the final tile type placed at any given voxel, aside from neighbor detection gadgets.
The \emph{direction tile} provides three pieces of information - the location of the voxel, and the tile growth direction (`+' or `-' growth, defined by 0 or 1 glues, respectively), and the growth direction of the prior tile placed.
The location of the voxel is simply tracked by the location of the direction tile which has active glues allowing for binding with \emph{decoding tiles} (tiles which bind to the locations encoding shape information on the original encoding structure), and the direction of growth is defined by the value of the first row of the encoding (see Figure~\ref{fig:deconstructor-encoding_order-0} for additional details on growth direction).
Each direction tile is either placed directly on top of the first row encoding the direction of growth (as in Figure~\ref{fig:generator-tile_specifics-0}), or is placed due to cooperative binding which passes the directional information among directional tiles with the same $x$ coordinate (i.e., tiles grow in the same direction as their neighbor in prior slices).
Additionally, the direction tile determines whether the prior direction tile was `+' or `-' growth direction by glue bindings.
The \emph{direction change detectors} bind to the current direction tile and the direction tile for the succeeding voxel - this, along with being bound to the prior direction tile allows for the current direction tile to expose a glue which encodes for both the direction of growth and determine if the tile is at the beginning or end of a row.
If the direction tiles of either adjacent tile contain a growth direction different from that of the current direction tile, the current tile is at the end or beginning of a row.
The \emph{decoding tiles} provide the information as to whether the tile in the current encoding of voxel location is either a shape of fill tile.
The binding of a decoding tile to the encoding supertile is enabled by cooperative binding with the direction tile.
All the information gathered by both the direction tile and the direction change detectors map to the activation of one of six possible glues, corresponding to the six tiles in Figure~\ref{fig:decoder-decoding_analysis-0}.
The decoding tile placed now contains the information regarding the growth direction of the tile and whether the tile is a shape or a filler tile.


\begin{figure}[htp]
    \centering
    \includegraphics[width=6cm]{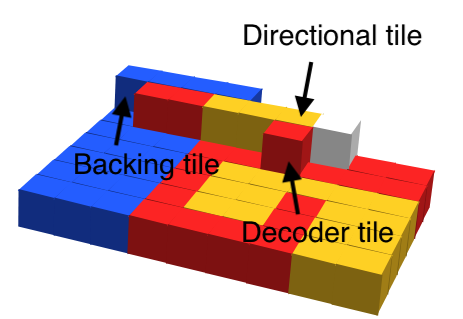}
    \caption{An example of the information which is gathered from the encoding structure. The directional tile gathers information regarding the growth type of tile location encoded. The direction change detector gadget (white) which detects that growth type `0' shifts in growth type `1', indicating a change of row and necessitating a direction change tile. The decoder tile, once glues are available to cooperatively bind to the encoding structure and the directional tile, determines that the tile in the current location is a shape tile}
    \label{fig:generator-tile_specifics-0}
\end{figure}

Shape tiles take an additional piece of information - whether or not the tile in the same $(x,y)$ coordinate in the prior slice (i.e., if $(x,y,z)$ is the location of current tile to be placed, its neighbor in the prior slice is $(x,y,z-1)$) is a filler tile or shape tile.
A shape tile cannot be immediately connected to a filler tile in the prior slice and remain in place, as that filler tile must be removed to prevent it being trapped in an enclosed cavity.
This information cannot be learned at the initial binding location shown in Figure~\ref{fig:generator-tile_specifics-0}.
As such, the decoding tiles expose glues to enable tile growth to the voxel of the tile.
This final piece of information is determined by the binding of one of three \emph{neighbor detection gadgets}.

\begin{figure}[htp]
    \centering
    \includegraphics[width=6cm]{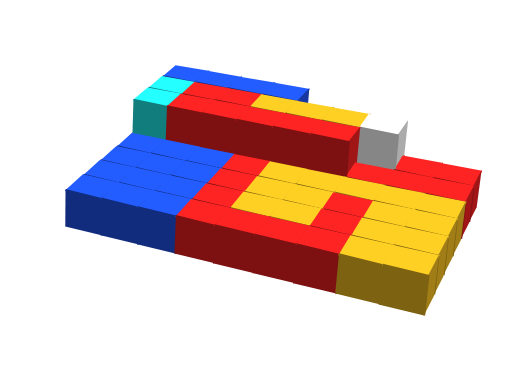}
    \caption{Continuation of the example in Figure~\ref{fig:generator-tile_specifics-0}. After the decoding tile has determined all information regarding the tile to be placed from the encoding (a shape tile which is at the end of a row), the decoding tile initiates growth of tiles which allow for the information regarding the tile to reach its voxel - the additional red tiles grown from the encoding structure. The final piece of information which dictates the type of tile to place is the tile type which is present in the slice prior. A neighbor detection gadget (teal) is utilized to cooperatively bind to the decoding tile }
    \label{fig:generator-tile_specifics-1}
\end{figure}

When the growth of the decoding tiles reaches the location for placement of a tile (the process by which this occurs is detailed in following sections), the neighbor detection gadget cooperatively binds with the decoding tiles and the neighbor of the tile to be placed in the current location.
If a shape tile is detected, the gadget detaches and activates a glue to place a tile which requires that binding of the neighbor occurs before growth of the slice can continue (Figure~\ref{fig:decoder-decoding_analysis-1}).
If a fill tile is detected in the prior slice, we utilize a shape tile which pre-encodes the information that face of the neighbor contains a $g_x$ glue, as the fill tile must be removed before the shape tile can detect the glues on the fill tile (Figure~\ref{fig:decoder-decoding_analysis-0}).
Additionally, we initiate a process of guaranteeing removal of the fill tile that requires a duple be used for the removal process.
We also have a special neighbor detection gadget for the first slice, where the neighbor tile is a \emph{backing} tile (used to enable strength 2 connections between all slices, it is described further in Section~\ref{sec:decode-row-1}).
Due to the neighbor detection gadget sensing a backing tile, the shape tile to be placed will pre-encode the $g_x$ glue.
The binding of the neighbor detection gadget to a backing tile causes the growth of an additional backing tile.

\paragraph{Detecting Neighbors and Removing Fill Tiles}
We present an example of the deconstruction process necessary for fill tile removal in the decoding of a shape.
The supertile described is a continuation of the examples in Figures~\ref{fig:generator-tile_specifics-0},~\ref{fig:generator-tile_specifics-1}.
First, a fill tile neighbor detection gadget cooperatively binds to the decoding tiles growing outward from the encoding and the fill tile of the prior slice (Figure~\ref{fig:generator-neighbor_removal-0}).

\begin{figure}[htp]
    \centering
    \includegraphics[width=6cm]{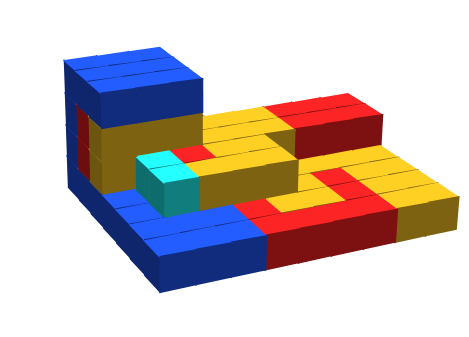}
    \caption{The detector tile initiates the placement of a fill tile in the next voxel location. This allows for cooperative binding of a neighbor detection gadget (teal) to the fill tile placed in the prior slice}
    \label{fig:generator-neighbor_removal-0}
\end{figure}

After this binding occurs, the fill detector gadget binds with strength 2 to the fill tile.
This binding additionally causes all the remaining glues on the fill tile to be set to the \texttt{off} state; once this glue deactivation occurs, the 3-tile unit will detach from the growing supertile and become junk (Figure~\ref{fig:generator-neighbor_removal-1}).
Detachment of the size-3 junk allows for cooperative binding to place the tile encoded by the decoding tile such that it has not blocked the removal of the fill tile in the prior location (Figure~\ref{fig:generator-neighbor_removal-2}).
As provided by the construction, a strength 2 connection exists between any remaining tiles in the slice.

\begin{figure}[htp]
    \centering
    \includegraphics[width=6cm]{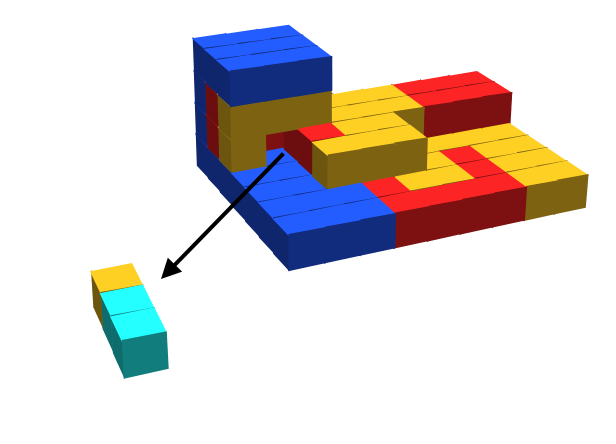}
    \caption{After the fill detector gadget (teal) binds to the fill decoding with strength 2, this causes the fill tile to detach from its slice. Once all glues on the fill tile have deactivated, the size-3 junk is able to detach from the supertile.}
    \label{fig:generator-neighbor_removal-1}
\end{figure}

\begin{figure}[htp]
    \centering
    \includegraphics[width=6cm]{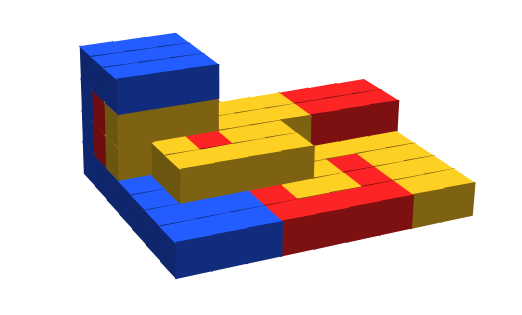}
    \caption{The new fill tile in the current slice is allowed to cooperatively bind once the fill detector gadget junk detaches from the supertile.}
    \label{fig:generator-neighbor_removal-2}
\end{figure}

\paragraph{Slice Incorporation}
We refer to the process by which tiles bind to their neighbors as slice incorporation; this process occurs in a similar manner for both type fill and shape tiles, however shape tiles may need to additionally bind to a neighbor in the prior slice.
First, a tile binds to its predecessor.
This is enabled by the two starting active glues, as shown in Figure~\ref{fig:decoder-decoding_analysis-0} by the solid black squares.
One glue is provided by the decoding tiles, and the other is provided by the neighbor; these map uniquely to a single tile.
Once binding occurs to the predecessor and the tile is a shape tile and has a neighbor in the prior slice, it then binds to its neighbor (Figure~\ref{fig:decoder-decoding_analysis-1}).
At this point, growth can continue in the slice and a glue is exposed; the tiles is a shape tile, it exposes at $s$ type glue, $f$ type if it is a fill tile.
The binding of this tiles successor activates a glue in the $+y$ direction.
Once the $+y$ direction glues bind, we then pass a signal in the $-y$ direction.
As shown in both Figures~\ref{fig:decoder-decoding_analysis-0} and~\ref{fig:decoder-decoding_analysis-1}, the $+y/-y$ face between tiles which change rows utilizes two separate set of glues, as tile growth occurs in the $+y$ direction before signaling slice growth completion.
Finally, once bound in the $-y$ direction we activate glues in the $+z$ direction, allowing for growth of the next slice.
In this sequence of glue activation, we guarantee that the topmost row of a slice will be bound fully to all neighbors in the slice before glues are activated allowing for new growth.
As such, in order for the first tile in a new slice to be placed, it must be connected with strength $\geq 2$ to the encoding structure via the topmost layer.

We note that with shape tiles, each tile contains the information to be connected to its neighbors and expose surface glues in any exterior location or internal location adjacent to fill tiles.
These exterior glues can become active immediately, or be activated at some later point by the action of some sort of gadget binding to the surface and causing signals be passed through the entire structure.
If the glues begin in the \texttt{on} state, we must take care such that if we present a replicating system (per Theorem~\ref{thm:replicator}) that they do begin the encoding process while decoding is taking place.
For that reason, in this construction we do not immediately activate the $g_x$ glues of an encoded shape.
The shape resulting from our tileset is terminal once all extraneous fill tiles and base tiles have detached from the encoding.
These shape tiles begin with strength 1 glues along all exterior edges of type $g_a$; these have no complement in either tileset involved in replication.
However, we can define an \emph{activation corner gadget} which contains two $g_a^*$ glues and is able to bind to the inactive shape tiles.
Upon binding of the activation corner gadget to the shape tile, glues bound to the activation corner gadget initiate a cascade of signaling to all other tiles in the shape which deactivate $g_a$ glues and in their place activate $g_x$ glues

Having a process to connect the tiles in a slice together, we now present the remainder of the tiles utilized to place shape tiles in the appropriate location and validate completion of the encoded shape.

\subsubsection{Base Creation}
Before continuing, we first provide an example shape (and its encoding) which will be used throughout the remainder of this section.
Figure~\ref{fig:shape-encoding} demonstrates the encoding of the shape provided in Figure~\ref{fig:shape-example}.
This shape and the encoding of the shape are used throughout the section as an example.

    \begin{figure}[htp]
        \centering
        \includegraphics[width=4cm]{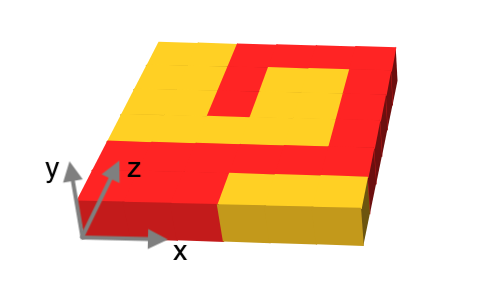}
        \caption{Encoding of the initial shape. Red voxels expose `1' glues on their surface, and yellow indicate exposed `0' glues. The first row indicates direction of growth for tiles in the same $+z$ column. The orientation of the axes for growth (identical to the orientation during encoding) is defined as shown.}
        \label{fig:shape-encoding}
    \end{figure}
    
    \begin{figure}[htp]
        \centering
        \includegraphics[width=4cm]{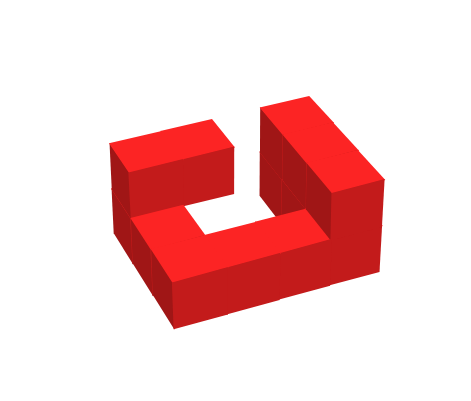}
        \caption{The shape which will be decoded from the encoding}
        \label{fig:shape-example}
    \end{figure}

We demonstrate the set of tiles which create a base, the initial set of tiles which, when combined with an encoding of a shape, nucleate growth and serve as a foundation for the remainder of the growth process.
We note that this encoding in a rectangular plane is convenient for our purposes (and prevents a massive increase in the number of tiles and signals required), however this entire process could be completed with only `0' and `1' tiles encoded in a line.

This encoding supertile begins with a strength 2 $g_0$ glue exposed, allowing for the tile $t_0$ to bind (Figure~\ref{fig:generator-base-0}).
Once $t_0$ is bound, it begins the process of growing the base by activating signals which cause uniquely mapped tiles to bind with the purpose of finding the width of the shape, demonstrated in Figure~\ref{fig:generator-base-1}.

\begin{figure}[htp]
    \centering
    \includegraphics[width=4cm]{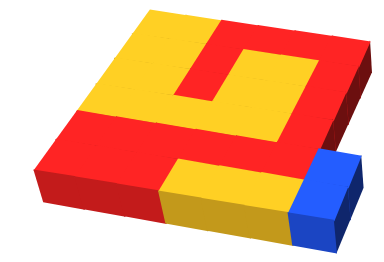}
    \caption{Initial binding of $t_0$ to encoding supertile, with the second tile included (base tiles indicated by blue)}
    \label{fig:generator-base-0}
\end{figure}

\begin{figure}[htp]
    \centering
    \includegraphics[width=4cm]{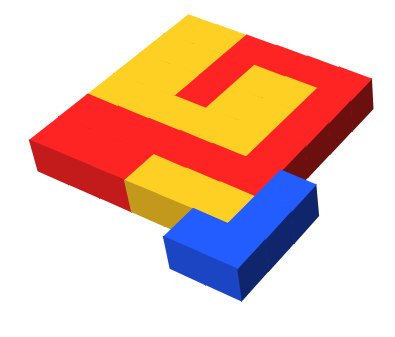}
    \caption{Extending initial base tiles (blue) to begin reading the width of the shape.}
    \label{fig:generator-base-1}
\end{figure}

We first determine the width of the shape.
Since each row alternates direction, we can utilize this information to construct a set of tiles which are able to identify the width of the base required for decoding.
A set of \emph{counting tiles} are able to add tiles to the existing supertile which define a base the width of the shape.
This counting process operates by cooperatively adding one tile to attach to the width-detection tiles.
The first row is able to utilize signals passed through the unique tiles which initiated growth to cause the addition of at least a width-1 base (Figure~\ref{fig:generator-base-2}).
We note that the encoding will be guaranteed to contain more than two tiles in any row due to the tiles added in the process of leader election.
The tile encoding the second location of the base then activates a strength 2 glue which allows for the binding of a counting tile (Figure~\ref{fig:generator-base-3}).
This tile enables cooperative growth along the edge of the currently exposed counting tiles.

\begin{figure}[htp]
    \centering
    \includegraphics[width=4cm]{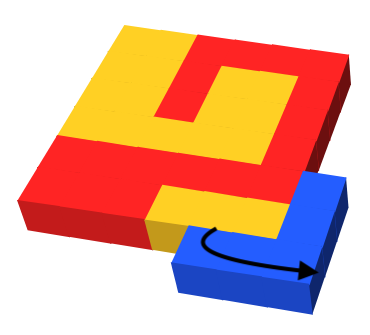}
    \caption{The first counting tile extends the width of the base by 1 voxel. Since we have used unique tiles up to the point, we are able to pass a message through to cause the addition of two general base tiles.}
    \label{fig:generator-base-2}
\end{figure}

\begin{figure}[htp]
    \centering
    \includegraphics[width=4cm]{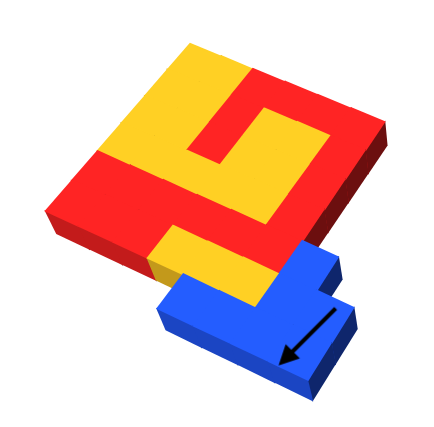}
    \caption{Newly placed tiles initiate a message which cause a strength 2 counting glue to be exposed.}
    \label{fig:generator-base-3}
\end{figure}

Once the counting tiles reach the end of the existing growth, one of two possible \emph{counting detectors} is able to bind to the new growth of counting tiles and the encoding structure (Figure~\ref{fig:generator-base-4}).
The two counting detectors have glues which sense either a `0' direction tile or a `1' direction tile.
Since the initial row is of direction `0', the counting process will be sent a signal along the new growth to both extend the width of the base by 1 tile and dissolve the prior placed counting detectors into size 1 junk in order to allow for the counting process to repeat (Figure~\ref{fig:generator-base-5}).
Otherwise, if a direction `1' tile is sensed, we have found the beginning of the second row and can terminate the counting process. 
Once this counting process is completed, we activate glues on the initial base tiles to cooperatively fill in the remainder of the tiles (Figure~\ref{fig:generator-base-6}).
The cooperative filling is determined to be complete by the binding of a \emph{base completion gadget} (Figure~\ref{fig:generator-base-7}), returning a signal to $t_0$ that causes another set of signals to be propagated that enable the placement of a base tile in the $+z$ direction.

\begin{figure}[htp]
    \centering
    \includegraphics[width=4cm]{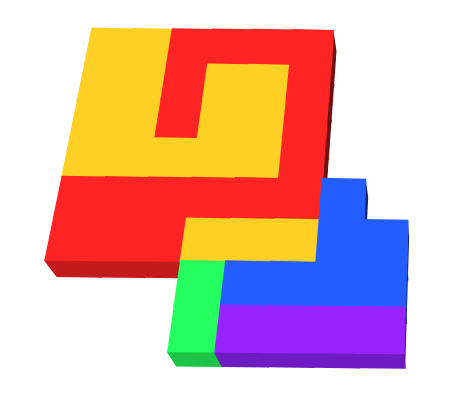}
    \caption{Cooperative growth along blue base tiles allows for counting tiles (purple) to reach the furthestmost tile. A duple (green) allows for the counting row to sense when it must extend the base by an additional tile by cooperatively binding to both the furthest counting tile and a `0' on the encoding supertile. Messages are sent to extend both the base tiles counting the current base width and to extend the width of the base by 1.}
    \label{fig:generator-base-4}
\end{figure}

\begin{figure}[htp]
    \centering
    \includegraphics[width=4cm]{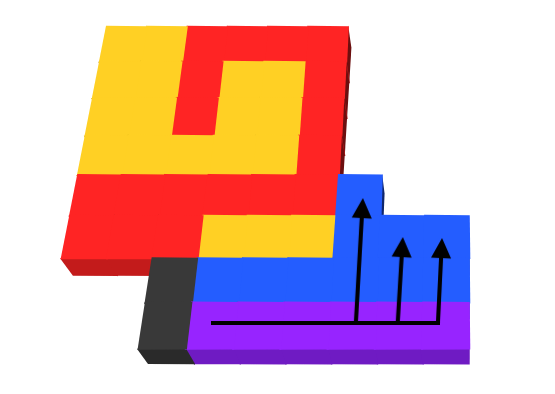}
    \caption{Once the counting row reaches the `1' tiles, this indicates the base is of the correct width. This is sensed by a counting duple (black) which cooperatively binds to both the counting row and the `1' glue.}
    \label{fig:generator-base-5}
\end{figure}

\begin{figure}[htp]
    \centering
    \includegraphics[width=4cm]{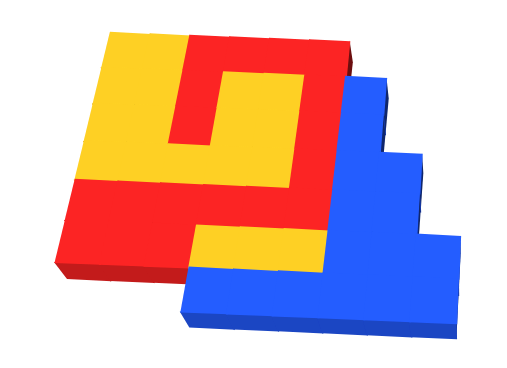}
    \caption{After binding of the black counting duple, the counting tiles dissolve and a signal is sent to begin cooperative growth of the remainder of the base adjacent to the encoding}
    \label{fig:generator-base-6}
\end{figure}

\begin{figure}[htp]
    \centering
    \includegraphics[width=4cm]{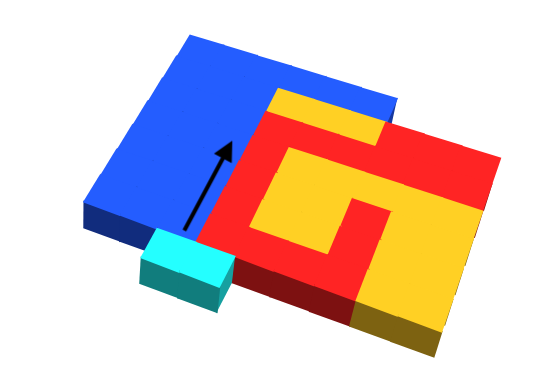}
    \caption{Base completion duple (white) allows for the base to detect when tiles have extended the base along the entire edge of the initial encoding supertile. A message returns to the initial tiles placed once all tiles of the row adjacent to the encoding have been placed in the base.}
    \label{fig:generator-base-7}
\end{figure}

\subsubsection{Row 1 Tile Placement}\label{sec:decode-row-1}
Once the base is complete, a signal is sent to begin the decoding process of the first row.
Figure~\ref{fig:generator-growth_dir0-row0-0} demonstrates how this signal allows for a strength 2 glue to be exposed in the $+y$ axis, allowing for a base tile to generate cooperative binding on top of the first directional tile.
Unlike other directional tiles, the directional tile of first tile of the first row encodes the information that a row change tile is to be utilized, without the need for sensing the directional tile prior (as there is no prior directional tile).
Once the directional tile binds, it then activates a glue allowing for the cooperative binding of a decoder tile that determines if the origin tile is a shape or fill tile.
Additionally, this binding causes a signal to be passed backwards through the base tile most recently placed such that it initiates the growth of a \emph{backing tile}.
Backing tiles serve two main purposes; first, to indicate to tiles of the first slice that they are adjacent to an exterior edge, and any shape tile must encode exterior glues on its $-z$ face.
Second, backing tiles allow for the tiles in the topmost row of a slice to bind along their top edge with strength 2 connections.
The process by which this second item proceeds is outlined in Section~\ref{sec:decoder-slice-completion}.

Once the decoder tile determines which type is to be placed, a glue is exposed in the $+x$ direction to enable growth of the decoding tiles.
Due to the current decoding tile being the first tile of the row, we can guarantee that at this point a neighbor detection gadget
must bind to the recently placed backing tile and the decoding tile (Figure~\ref{fig:generator-growth_dir0-row0-1}).
This binding of the neighbor detection gadget with the backing tile additionally causes the backing tile to activate a glue allowing for cooperative binding of another backing tile with the base in the $+x$ direction.
The decoding tile now contains all the information regarding the tile type to place after binding with the neighbor detection gadget.
A strength 2 glue allows for the growth of an additional decoder tile (mapping to the tile type indicated in the encoding assembly); this enables cooperative binding of the tile type mapped between itself and the base tiles (Figure~\ref{fig:generator-growth_dir0-row0-2}).
After the base tile and decoded tile of the shape are connected with strength 2, signals are sent back through the decoder tiles towards the directional tile which initiated growth.
Upon passing this signal to the decoder tile's predecessor, all decoder tiles not bound to the directional tile dissolve into size 1 junk (Figure~\ref{fig:generator-growth_dir0-row0-3}).
The decoder tile adjacent to the directional tile activates a glue indicating for the next directional tile to be placed, thus allowing for the placement of an additional decoding tile.

\begin{figure}[htp]
    \centering
    \includegraphics[width=4cm]{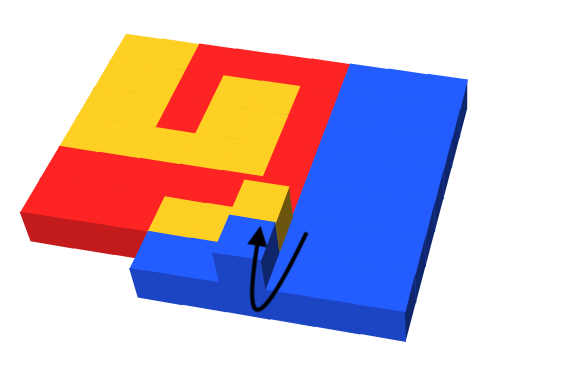}
    \caption{The initial tile, once messages have been received that the base is complete, initiates a signal which causes a base tile allow cooperative binding of the first directional tile. Additionally, a glue initiating growth of the first backing tile is exposed in the $+x$ direction}
    \label{fig:generator-growth_dir0-row0-0}
\end{figure}

\begin{figure}[htp]
    \centering
    \includegraphics[width=4cm]{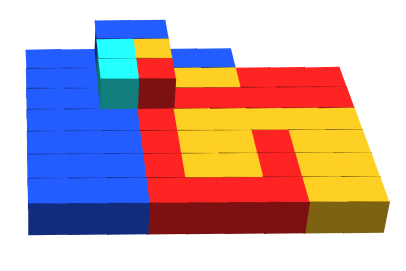}
    \caption{Once cooperative binding has occurred to dictate the decoding tile, glues are activated on the $+x$ face of the decoding tile, allowing for cooperative growth and binding of neighbor detection tiles. A strength 2 glue is exposed upon binding with the neighbor decoding tile, allowing for a decoding tile to be added which cooperatively places the tile encoded.}
    \label{fig:generator-growth_dir0-row0-1}
\end{figure}

\begin{figure}[htp]
    \centering
    \includegraphics[width=4cm]{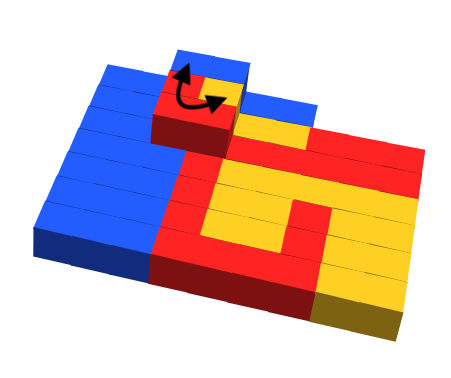}
    \caption{Once the row-1 tile binds to the base, it exposes a glue in the $-z$ axis that allows for the cooperative binding of the fill/shape tile encoded by the first location. Once cooperative binding occurs, a second glue is activated which allows for a strength 2 connection between the shape/fill tile most recently placed and its predecessor (in this case the base tile - the first row of which contains glues and signals which allow for binding in this manner). Additionally, when the detection duple binds to the backing tile, a signal is sent to activate glues in both the $+x$ and $+y$ directions which allows for a second tile to bind}
    \label{fig:generator-growth_dir0-row0-2}
\end{figure}

\begin{figure}[htp]
    \centering
    \includegraphics[width=4cm]{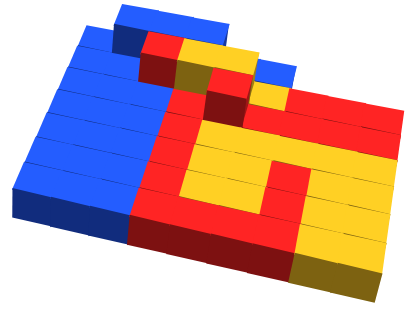}
    \caption{A message is passed backwards along the binding edges such that the direction tile activates a glue which allows for the next directional tile to bind. Additionally, the decoding tiles placed in support of the prior encoded location of the shape deactivate all glues and become junk to allow for the next tile of the encoding to be placed utilizing the same path of voxels.}
    \label{fig:generator-growth_dir0-row0-3}
\end{figure}

\begin{figure}[htp]
    \centering
    \includegraphics[width=6cm]{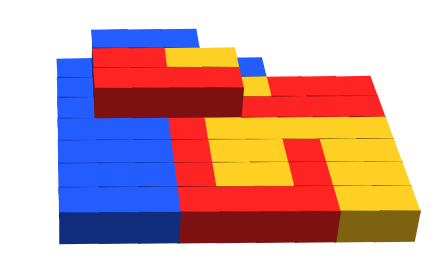}
    \caption{Placement of encoded tiles continues, with decoding tiles re-utilizing the same set of voxels to grow voxels further away from the origin.}
    \label{fig:generator-growth_dir0-row0-4}
\end{figure}

Tile additions continue also utilizing the direction change decoding demonstrated in Section~\ref{sec:shape-attachment-details}, until the final tile of the row is reached.
At this point growth continues by the standard process of directional tiles allowing cooperative binding with the encoding structure, however switched to direction `1' growth.
In order to enable the placement of encoding tiles via direction `1' growth, the backing tiles must be present in the new row to allow for binding of neighbor detection gadgets.
A \emph{backing growth detector} (see Figure~\ref{fig:generator-growth_dir0-row0-5})binds to the most recently placed backing tile and the base (or backing) tile in the row prior.
Binding of the backing growth detector allows for a strength 2 glue to be turned \texttt{on} to enable the growth of a backing tile in the $+y$ direction (Figure~\ref{fig:generator-growth_dir0-row0-6}).


\begin{figure}[htp]
    \centering
    \includegraphics[width=6cm]{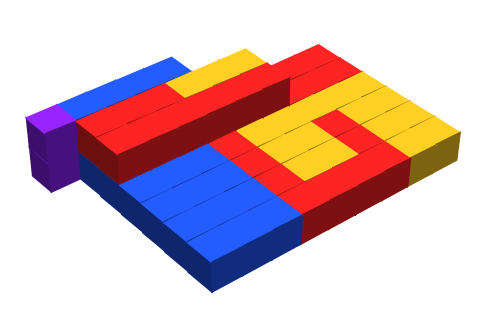}
    \caption{Backing growth detector (purple) binds to the outermost backing tile and the base to signal to the backing tile to activate a strength 2 glue in the $+y$ direction. Note that for following rows, the backing growth detector will bind with two backing tiles}
    \label{fig:generator-growth_dir0-row0-5}
\end{figure}

\begin{figure}[htp]
    \centering
    \includegraphics[width=6cm]{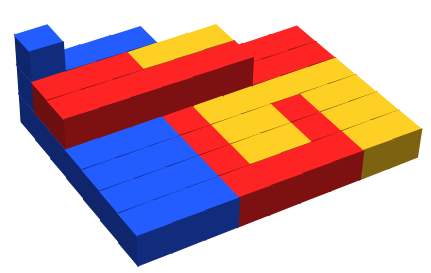}
    \caption{Binding of the next backing tile in order ready growth, allowing for binding of neighbor detection tiles.}
    \label{fig:generator-growth_dir0-row0-6}
\end{figure}

\subsubsection{Row 2n Tile Placement}\label{sec:decode-row-2n}
For each even numbered row, tiles grow in the `-' direction; that is, the first tile in the some row 2n is placed above the last tile of the prior row (2n - 1)   
For row 0 growth, each additional tile placed took us further away from the origin point (e.g., incrementing the $x$ value in the ($x,\: y,\:z$) position tuple).
In the case of `-' direction growth, tiles of the slice are placed at the furthest-most $x$ value of the slice and decrement to 0.
While the decoding tiles of the first row bind initially to decoding tiles, the most recently placed tile and directional tiles, the decoding tiles of `-' direction growth cooperatively bind with the prior decoding tile and a base tile.
Growth occurs in two cases; in the case of the first tile of a row of direction 1 growth, tiles bind until they reach the furthest-most base tile.
When reaching the outermost base tile, a \emph{direction `1' detection gadget} binds with the outermost base tile and the furthest placed decoding tile (Figure~\ref{fig:generator-growth_dir1-0}).
At this point, a glue is activated on the decoding tile's $+y$ face, allowing for cooperative growth to continue.
This allows for cooperative growth along the previously placed tiles until no longer possible, at which point a neighbor detection gadget is able to bind to the decoding tile and the neighbor tile (in this case a backing tile, see Figure~\ref{fig:generator-growth_dir1-1}).

\begin{figure}[htp]
    \centering
    \includegraphics[width=4cm]{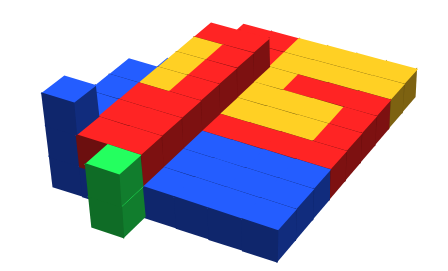}
    \caption{Cooperative binding for direction `1' tile growth of the first tile in row 2 extends to the edge of the base. A direction `1' detection gadget (green) attaches to the base and the growing row, indicating the edge has been reached. Once the direction `1' detection gadget is bound, a glue activates on the $+y$ face of the tile, allowing for cooperative growth in the $+y$ direction on the currently grown structure.}
    \label{fig:generator-growth_dir1-0}
\end{figure}

\begin{figure}[htp]
    \centering
    \includegraphics[width=4cm]{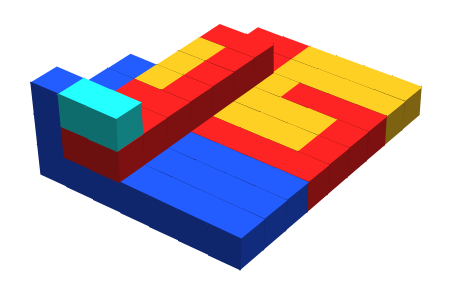}
    \caption{The binding of the neighbor detection gadget allows for a strength 2 glue to activate in the $+y$ direction, allowing for a tile with glues mapping to the decoding tile type (in this case, a shape tile which has a $g_x$ glue encoded on its back side) to cooperatively bind to the prior tile placed.}
    \label{fig:generator-growth_dir1-1}
\end{figure}

Similarly, this allows for both the placement of the encoded tile and the extension of the backing tiles; upon the placement of the encoded tiles, a signal is sent to dissolve all decoding tiles not involved in growth in the $+y$ direction into size 1 junk.
The next directional tile is added, allowing for the binding of the next decoding tile and the growth to place the tile dictated by the encoding structure.
To sense when the growth of the decoding tiles in the $+x$ direction has reached its furthest-most point, the remaining decoding tile which originally redirected growth in the $+y$ direction enables a glue similar to that present on the direction `1' detector gadget.
We note this does not cause interactions between multiple encoding processes going on in parallel, as the presence of the base tiles and the directional row offset any possible growing decoding tile (Figure~\ref{fig:generator-growth_dir1-2}).
Once the neighbor detector gadget binds, it grows in the $+y$ direction and places its encoded tile (Figure~\ref{fig:generator-growth_dir1-3}).
This repeats until all tiles of the row have been added.

\begin{figure}[htp]
    \centering
    \includegraphics[width=4cm]{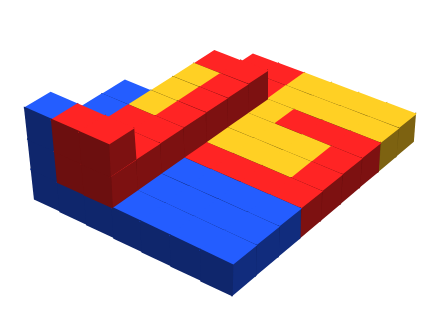}
    \caption{After binding of neighbor detection gadget, shape tiles are placed.}
    \label{fig:generator-growth_dir1-2}
\end{figure}

\begin{figure}[htp]
    \centering
    \includegraphics[width=4cm]{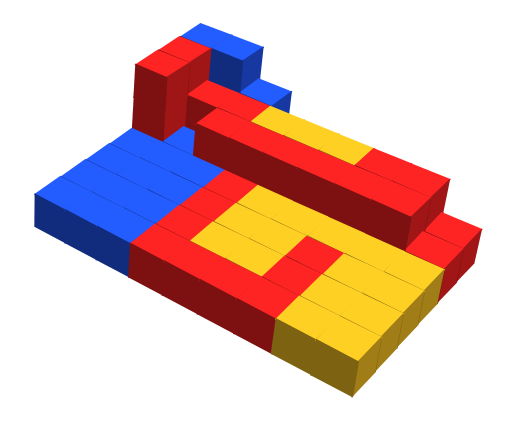}
    \caption{Mid-growth of the second tile in the encoding of row 2. Note that all but one horizontal tiles are deactivated in direction `1' growth, this is in order for collision to occur and correctly place remaining tiles.}
    \label{fig:generator-growth_dir1-3}
\end{figure}

\begin{figure}[htp]
    \centering
    \includegraphics[width=4cm]{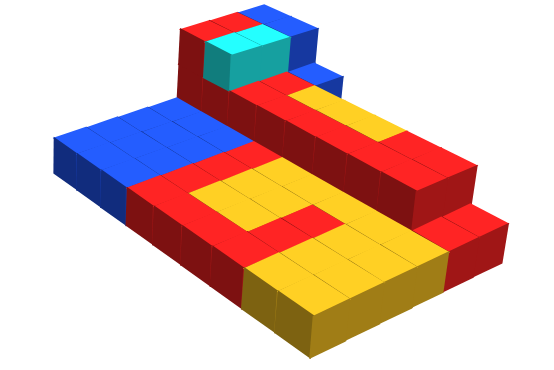}
    \caption{Neighbor detector gadget binds to the furthest-most placed decoding tile of the second decoding tile after colliding with the prior decoding tile growth. This leads to placement of encoded tile and growth of backing. This process repeats for all remaining direction `1' tiles in the row.}
    \label{fig:generator-growth_dir1-4}
\end{figure}

At the end of this row, the backing tiles must grow in the $+y$ direction again.
For row 2, the current backing gadget will not work as there exists a base tile hindering growth (which is necessary for future signals to be sent).
A modified, one-tile gadget is utilized for this specific case.
Additionally, once the row is complete after the placement of a direction change tile, all remaining decoding tiles are dissolved into size 1 junk allow for growth of direction `0' tiles of the following layer.

\subsubsection{Row 2n + 1 Tile Placement}
While growth of row 1 was in direction `0', it is a special case due to the fact that it placed tiles in voxels with the same coordinate in the $y$ axis as the decoding tiles by a set of tiles unique to the first row.
For remaining odd-numbered rows, we must carry out a similar growth in the $+y$ direction before as placing the encoded tile as demonstrated by the row 2 growth example, but incrementing $x$ values.
We note that the example figures in this section do not directly correspond to the encoding provided in Figure~\ref{fig:shape-encoding}, however these are presented to provide the reader with examples of how this process would occur in an encoding which does contain at least 3 rows.
Decoding tiles of some odd valued row grow by cooperatively binding with the decoding tile and previously placed directional tiles, as with the row 1 tiles.
However, upon binding with a shape or a fill tile they activate a glue in the $+y$ direction.
This glue attempts to allow for growth of decoding tiles in the $+y$ direction, leading to the binding of a neighbor detection gadget and the placement of the encoded tile (Figures~\ref{fig:generator-growth_dir0-duple},~\ref{fig:generator-growth_dir0-0}).
Similarly to even numbered direction `1' row growth, decoding tiles are dissolved into size 1 junk to allow for reuse of voxels.
In contrast, all but the bottom-most decoding tile are removed, and glues are activated allowing remaining decoding tiles to sense that a tile has already been placed in the current location (Figure~\ref{fig:generator-growth_dir0-row2n-2}). 
In the case when the decoding tile activates its glue in the $+y$ direction and binds to a tile, it continues growth in the $+x$ direction until finding an open location to grow (Figure~\ref{fig:generator-growth_dir0-row2n-3}).

\begin{figure}[htp]
    \centering
    \includegraphics[width=4cm]{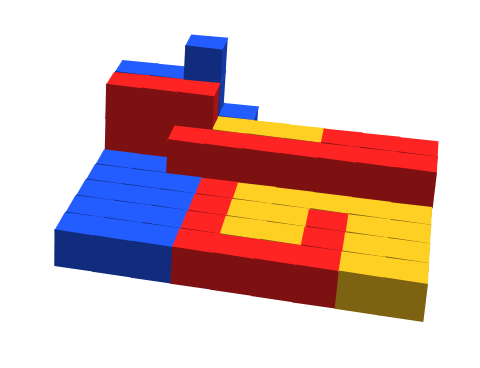}
    \caption{As the direction `0' tile (first tile of row 3) initiates growth, when a tile is cooperatively placed on a base tile it immediately activates a glue in the $+y$ direction. Since a path exists for tiles to grow in that direction, they grow until no cooperative location is available. }
    \label{fig:generator-growth_dir0-duple}
\end{figure}

\begin{figure}[htp]
    \centering
    \includegraphics[width=4cm]{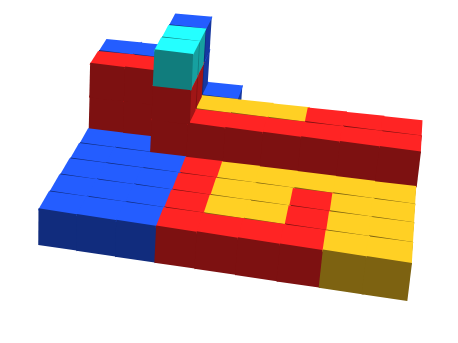}
    \caption{At this point, a detector gadget (teal) binds and indicates that growth has reached the point for the placement of the voxel encoded.}
    \label{fig:generator-growth_dir0-0}
\end{figure}

\begin{figure}[htp]
    \centering
    \includegraphics[width=4cm]{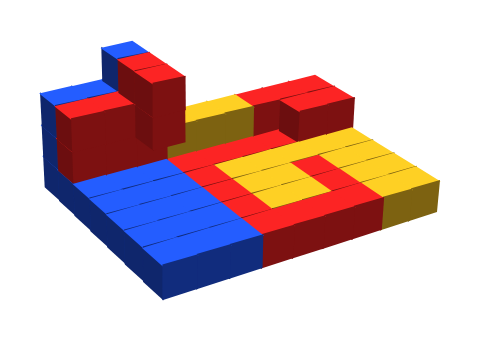}
    \caption{As signals are passed backwards through the tile growth, all horizontal tiles are deactivated. This allows for the direction `0' voxels to sense prior placed tile locations from the same row. Note that tiles growing along the $+y$ axis are retained initially.}
    \label{fig:generator-growth_dir0-row2n-2}
\end{figure}

\begin{figure}
    \centering
    \includegraphics[width=6cm]{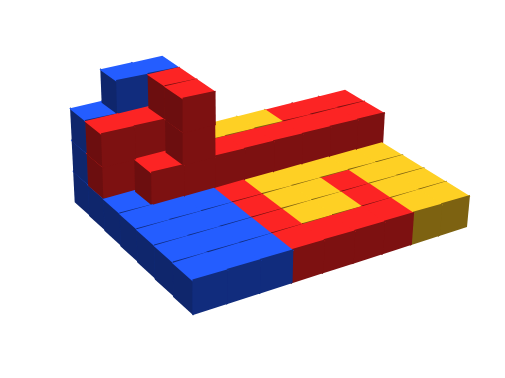}
    \caption{As the tiles which encode the second tile of row 3 grow to their placement location, upon first cooperative binding with the base they attempt to grow in the $+y$ direction. The signal `bounces', and the growth continues along the base. Since the second location has not been placed, the $+y$ direction of growth is free to take place.}
    \label{fig:generator-growth_dir0-row2n-3}
\end{figure}

\subsubsection{Slice Completion}\label{sec:decoder-slice-completion}
Once the directional tiles reach the end of the encoding of the final row within the structure, a \emph{slice completion gadget} binds to the end of the encoding and the directional tile. 
At this point, a message is returned through the current row of directional tiles which enabled growth of the slice (Figure~\ref{fig:generator-slice++0}).
Once the message is received by the first directional tile, it carries out two operations - the first being unique to the first slice.
In order for the growth of the next slice, we must be able to guarantee the shape tiles in the slice are connected to either the shape which has grown, or are connected to the newly growing slice.
To guarantee connection of all tiles of the first slice persist even after filler tile removal, we must create strength 2 connections between the encoding structure and all tiles of the first slice.
This is accomplished by extending the growth of the backing tiles, which allows for all tiles to be connected via strength 2 to the encoding structure.
The message is sent through the base tile which initiated growth of the first slice, into the adjacent backing tiles.
After backing tiles receiving the message, strength 2 glues are activated on all the $+y$ direction faces of the currently placed backing tiles.
Only the topmost layer of backing tiles will allow for cooperative placement of the new backing tiles on top of the newly created slice.
The newly placed backing tiles opens up cooperative binding locations for the backing tiles to then bind with the top row of the slice (Figure~\ref{fig:generator-slice++1}).
This allows for the tiles in the topmost row of the slice to activate glues for binding to their neighbor in the $-y$ direction.
Once bound to the neighbor in the $-y$ direction, the tiles are then able to activate glues which allow for neighbor detection gadgets to bind, allowing for the growth of a new slice.

\begin{figure}[htp]
    \centering
    \includegraphics[width=6cm]{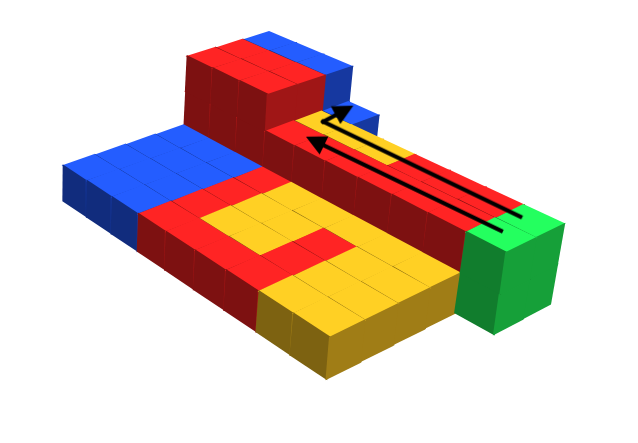}
    \caption{The slice completion gadget (green) binds to the outermost directional tile and decoding tiles, signaling for dissolution of decoding tiles and extension of backing tiles}
    \label{fig:generator-slice++0}
\end{figure}

\begin{figure}[htp]
    \centering
    \includegraphics[width=6cm]{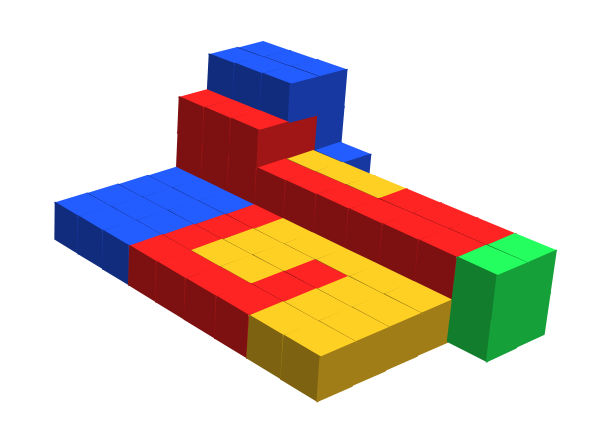}
    \caption{Backing tiles activate strength 2 glues, allowing for cooperative growth along the top of the first slice}
    \label{fig:generator-slice++1}
\end{figure}

In addition to the growth of the backing tiles, a signal is sent to place a new directional tile.
This directional tile takes the information of the first row of directional tiles and cooperatively binds with both 0 and 1 tiles on the encoding structure; its purpose is to simply pass forward the directional information and allow for the tile placement process to continue in the next slice.
In addition to the directional tile exposing a directional glue, we also expose a terminating glue ($g_t$) which is used in the detection of the completion of the final slice.
Once the growth of the new directional tile occurs alongside the creation of the top row of backing tiles, growth of the new slice can begin with starting conditions shown in Figure~\ref{fig:generator-slice++2}.

\begin{figure}[htp]
    \centering
    \includegraphics[width=6cm]{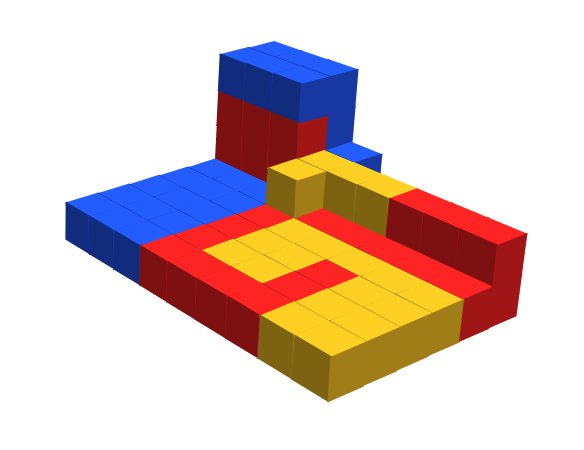}
    \caption{First directional tile of the second slice is ready to begin growth.}
    \label{fig:generator-slice++2}
\end{figure}

\subsubsection{Detaching From Base}
Slice growth proceeds via the previously described process until reaching the final slice.
Once the final slice is placed, a slice completion gadget binds allowing for the placement of a directional tile, as per any other row.
However, the exposed terminating glue allows for the attachment of the \emph{decoder completion detector} with the outermost edge of the encoding structure (Figure~\ref{fig:generator-final_slice-0}).
Upon binding of the decoder completion detector, a glue is activated to allow for the growth of \emph{decoder completion tiles} which cooperatively bind to the outermost slice layer.
Binding of the decoder completion tiles occurs such that only attachments between shape tiles activate glues for cooperative growth, and filler tiles must form a strength 2 duple with the decoder completion tiles.
Once bound as a duple, the filler tiles send glue deactivation signals to their remaining active glues.

\begin{figure}[htp]
    \centering
    \includegraphics[width=7cm]{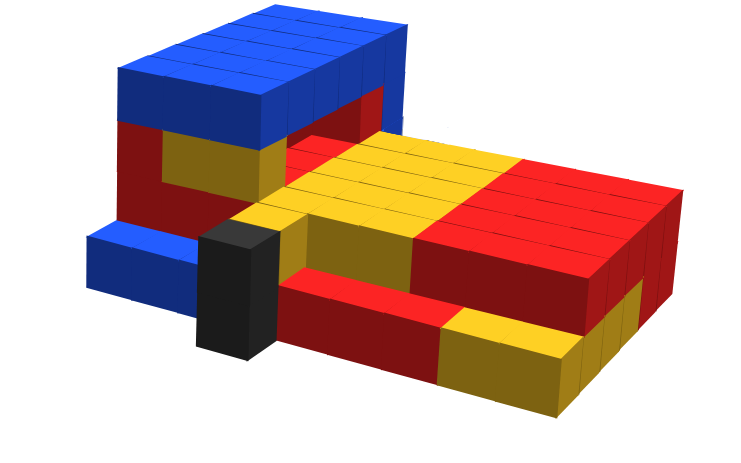}
    \caption{At the completion of the final row, the decoder completion detector (black) is able to bind with the outermost directional tile and cause growth of decoder completion tiles which remove remaining fill tiles.}
    \label{fig:generator-final_slice-0}
\end{figure}

Once a decoder completion tile binds with the outermost backing tile above the top row of a slice, it sends a dissolve message to all the base and backing tiles in the same $yz$ plane (Figure~\ref{fig:generator-final_slice-1}) to turn them into size 1 junk.
The base tiles, upon receiving this dissolve message, also initiate a message to dissolve the remaining tiles placed as part of the assembly sequence into size 1 junk, including the initial binding tile $t_0$.
The initial binding tile then signals to the encoding structure to dissolve into size 1 junk, and the only terminal assembly remaining is the shape assembly produced by the decoding process.

\begin{figure}[htp]
    \centering
    \includegraphics[width=7cm]{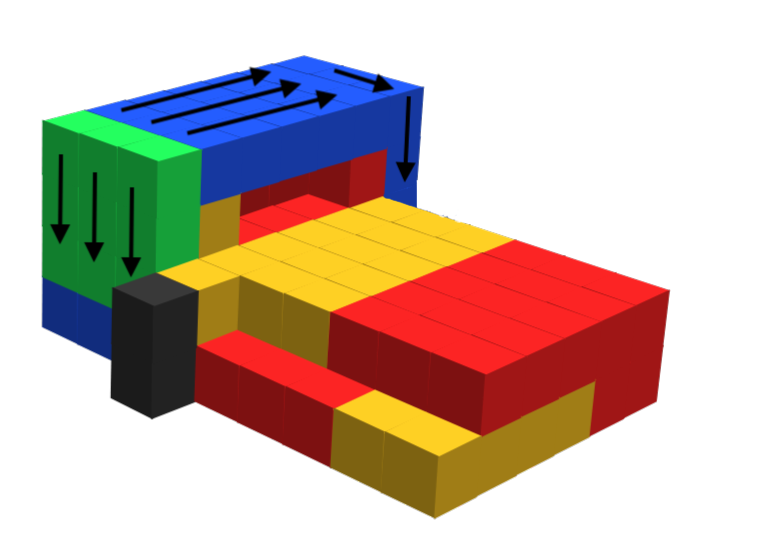}
    \caption{After the decoder completion tiles (green) bind to the final slice, this sends deactivation signals to the fill tiles and bind to the backing tiles, a dissolve message is sent to the remaining tiles involved in the decoding process.}
    \label{fig:generator-final_slice-1}
\end{figure}

\subsubsection{Proof of Universal Shape Decoding Correctness}\label{sec:universal-decoding}

Here we briefly summarize the decoding process and show that during this process, the shapes which were encoded in the set of input encoding assemblies $\Phi$ are correctly assembled. We first consider the decoding process of a single encoding assembly $\phi\in\Phi$ and note that a similar process happens for all encoding assemblies simultaneously without interfering with one another.

Our decoding process begins by building a base of tiles connected to $\phi$. This base holds the shape as it's being constructed and is used to help ensure the connectivity of the shape as it's being constructed. The decoding process is performed in iterations, where during each iteration a row of $\phi$ is scanned tile-by-tile and a corresponding 2D slice of the shape is constructed. Each slice is constructed starting from the bottom (smallest $y$ coordinate) to the top (largest $y$ coordinate), with tiles attaching in a zig-zag manner, as illustrated in Figure \ref{fig:deconstructor-encoding_order-0}. Each slice of the assembled shape corresponds to a unique $z$ coordinate so for convenience we call the slice whose $z$ coordinate is $i$, $\sigma_i$. As each slice is assembled, tiles are placed in each location of the slice, even those locations that will not be part of the final shape, though these will be removed during the assembly of the next slice.

The first slice $\sigma_1$ can be assembled naively, but during the assembly of each following slice, tiles which will not be part of the final shape on the previous slice must be removed. This is done as follows. Suppose that slice $\sigma_i$ ($i > 1$) is currently being assembled. Before a tile $t_i$ is placed in a location $(x, y, i)$, a gadget is used to determine the type of the tile $t_{i-1}$ at location $(x, y, i-1)$ (i.e. the tile with the same $x$ and $y$ coordinates on the previous slice). If this $t_{i-1}$ is part of the final shape, then $t_i$ is placed and signals are used to activate strength 2 glues between $t_i$ and $t_{i-1}$; otherwise, if $t_{i-1}$ is not part of the final shape, it is removed before $t_i$ is placed. Regardless of the type of tile $t_{i-1}$, when $t_i$ is placed, glues are activated which connect $t_i$ to all adjacent tiles on the same slice. Once the final slice is assembled, a final zig-zag pass is made in the next $z$ coordinate which removes all tiles from the last slice which are not part of the final shape.

It is also important to note that the base, on which the shape is being assembled, also forms a ceiling above the slices being assembled. This ceiling helps ensure that tiles on the top row of each slice are able to remain attached to the assembly during construction. It should be clear that during this decoding process (1) each tile that belongs to the final shape is placed in its correct location, and (2) that those tiles of a slice which are not part of the final shape will be removed from the assembly during the assembly of the next slice. However, because tiles are removed during the process, we must show that none of these removals can cause parts of the assembly to unintentionally detach. We state this as Lemma \ref{lem:decoder-correctness}.


\begin{lemma}\label{lem:decoder-correctness}
Let $\phi$ be an encoding assembly which encodes the shape $s$. During the decoding process above, as slice $\sigma_i$ ($i>1$) is being assembled, no tile in slices $\sigma_1,\ldots,\sigma_{i-1}$ which are part of the final shape assembly can detach.

\end{lemma}

\begin{proof}
    To prove this, we first note that all tiles in the slice $\sigma_1$ which will be part of the final shape assembly are bound to each neighboring tile in the slice, meaning that there is no risk of detachment until tiles are removed in later slices. We use induction on the $z$ coordinate of the slices to show that this holds. Therefore, assume the hypothesis holds for slices $\sigma_1,\ldots,\sigma_{k-1}$ and consider what happens as the slice $\sigma_k$ assembles. Before the assembly of $\sigma_k$, the only slice containing tiles that may need removal are in slice $\sigma_{k-1}$ since during the assembly of a slice, all tiles which are not part of the final shape assembly are removed from the previous slice.
    
    As slice $\sigma_k$ is being assembled, if all of the tiles in $\sigma_{k-1}$ are part of the final shape assembly, then nothing will be detached and the proof is complete. Assume then that there is some tile in slice $\sigma_{k-1}$ which is not part of the final shape assembly and thus needs to be removed. Assembly of $\sigma_k$ will continue until we reach such a tile, say $t$ at coordinates $(x_t,y_t,z_t=k-1)$. Gadgets will detect that $t$ needs to be removed before a tile, say $t'$, is placed in coordinates $(x_t,y_t,z_t+1=k)$. When $t$ is detected, $\sigma_k$ will be assembled up to the location of $t'$ meaning that there will be a tile in every location of $\sigma_k$ below $y$ coordinate $y_t$ as well as all locations at $y$ coordinate $y_t$ to either the left or right of $t'$ depending on the parity of the $y$ coordinate in the zig-zag growth procedure for $\sigma_k$.
    
    To ensure that the detachment of $t$ does not cause any other tiles to detach, we must look at all neighbors of $t$ in the assembly. 1 of these neighbors will be $t'$ itself and this tile will be attached to all of its neighbors in $\sigma_k$ so we don't have to consider that one. If $t$ has a neighboring tile in slice $\sigma_{k-2}$, then notice that this tile must (1) be a tile belonging to the final shape assembly since it was not removed during the assembly of slice $\sigma_{k-1}$, and (2) have at least 1 other neighboring tile in $\sigma_{k-2}$ or $\sigma_{k-3}$ to which it is attached since otherwise the shape being encoded would have disconnected parts which we don't allow. Therefore, the removal of $t$ would not cause this tile to detach.
    
    We now consider the 4 potential neighbors of $t$ in the slice $\sigma_{k-1}$. For the neighbor below $t$, say $t_{-y}$, we again note that, because shape $s$ cannot have any disconnected components, $t_{-y}$ must have at least one neighbor other than $t$ which is part of the final shape assembly. Because the current slice $\sigma_k$ has grown up to the $y$ coordinate of $t$, any such neighbor of $t_{-y}$ must already exist in the assembly is attached to $t_{-y}$ with strength 2. Therefore, the removal of $t$ will not cause $t_{-y}$ to detach. 
    
    Now consider the neighbors of $t$ with the same $y$ and $z$ coordinates, call these $t_{-x}$ and $t_{+x}$. Notice that because slices are grown in a zig-zag manner, the growth of the current slice $\sigma_k$ will be such that one of these already has a neighboring tile in $\sigma_k$ and one does not. Without loss of generality, suppose that at the current row of slice $\sigma_k$ attachments are happening from the $-x$ direction to the $+x$ direction so that $t_{-x}$ already has a neighbor in $\sigma_k$ and $t_{+x}$ does not. Because any neighbor of $t_{-x}$ that exists must have been placed by now, the detachment of $t$ will not cause $t_{-x}$ to detach for the same reason as $t_{-y}$. Now, For $t_{+x}$ it may be the case its only neighbor that is part of the final shape assembly is in slice $\sigma_k$ and has not yet attached. Still notice that because $\sigma_k$ has not yet finished growth, no tiles have yet been removed from $\sigma_{k-1}$ with a $y$ coordinate greater than $t_y$. This means that $t_{+x}$ still has neighboring tiles to which it is attached. This is even true if $t_y$ is at the top of the slice since the base contains a ceiling above the assembly to which the tiles are attached. Therefore, even if $t$ is removed, $t_{+x}$ will remain attached to the assembly. The same argument applies to $t_{+y}$, the neighbor above $t$.
    
    By the assembly procedure up to this point, it is therefore safe to remove tile $t$, place $t'$ and continue with the assembly of slice $\sigma_k$. Since this holds for any tile which needs to be removed from slice $\sigma_{k-1}$, the assembly of $\sigma_k$ will complete without any tiles that are part of the final shape assembly detaching. 
\end{proof}

From here, it's clear that the assembly of the slices of the shape can complete without erroneous detachment. Since all tiles that are part of the final shape assembly have been added during the slice construction and since all tiles which are not part of the final shape assembly have been removed from their respective slices, it's clear that the decoding process successfully assembles our final shape assembly.

Given the set of input encoding structures $\Phi = \{\phi_1,\ldots,\phi_n \}$, the STAM$^R$ system $\mathcal{D}_\Phi = \{D,\Sigma_\Phi,\tau=2\}$ produces a set of terminal supertiles $S = \{s_1,\ldots,s_n\}$ in parallel with a maximum junk size of 3.
$\mathcal{D}_\Phi$ finitely completes, as for the production of the set of shapes $s \in S$ from input encoding structures $\Phi$ a finite number of tiles are required for each encoding structure to produces a terminal assembly.
We can guarantee this as each encoding produce a single terminal shape, as the encoding of the shape dissolves into size 1 junk after the terminal shape has decoded.
By our construction, there are never exposed glues on the surfaces of any pair of assemblies that each contain an input encoding that would allow them to bind to each other.
Since junk assemblies produced by any assembly sequence are also unable to negatively interact with other assemblies, a system whose input assemblies have multiple shapes will behave simply as the union of individual systems which each have one input assembly shape, creating terminal assemblies of all of (and only) the correct shapes. 
This proves Lemma~\ref{lem:decoder}.

Now that we have shown the existence of universal encoding and universal decoding tilesets, we have the basis to demonstrate a universal shape replicator.
We generate a new STAM$^R$ tileset $R = E \cup D$ and STAM$^R$ system $\mathcal{R}_S = \{R, \Sigma_S, \tau=2\}$, where $\Sigma_S$ consists of an infinite number of copies of each tile type from $R$ and an infinite number of copies of each uniformly covered assembly from the set $S = \{s_1,\ldots,s_n\}$, whose shapes are any arbitrary set of shapes.

Recall that during the encoding process, the encoding corner gadget is bound to the encoding structure while it is being built. Once the entire encoding process finishes and the corner gadget receives a 'dissolve' signal, it first activates a glue to signal to the first tile placed in the encoding structure that it should turn $\texttt{on}$ the \emph{initiator glue} which is the glue initially bound to by the tiles of $D$. Thus, exactly when an encoding of some $s_i$, $\phi_i$, is completed by the tiles of $E$, decoding that $\phi_i$ will begin by the tiles of $D$, resulting in a terminal assembly with the same shape as $s_i$. We make a slight modification to the tile of the encoding structure that exposes the initiator glue, and the initial decoding tile which attach to it, the \emph{initiator tile}. We make two copies of the initiator tile, which we will call $t_{1}$ and $t_{2}$.  The first, $t_{1}$, will bind to the initiator glue and cause the decoding process to proceed exactly as before. However, when the original initiator tile would have detected completion of the decoding process and sent a `dissolve' signal to the first tile of the encoding structure, $t_1$ instead sends a signal that tells that tile to activate a glue that will allow $t_{2}$ to attach, and then $t_1$ will detach. This will effectively cause the encoding to produce a decoded structure and then have all of the `helper' tiles dissolve, leaving the encoding structure able to bind to $t_{2}$ which then initiates the regular decoding process, and when it receives the signal telling it that has completed, $t_2$ does pass the `dissolve' signal to the first tile of the encoding structure. In this way, each encoding structure causes two copies of the decoded assembly to be produced, and then dissolves. 

By our construction, the only glues required to be shared between the two tilesets are the glues encoding 1 and 0 on the encoding structure, and the previously mentioned glues on the encoded assembly which initiate the decoding process. 
The glues for 0/1 are shared by multiple tiles in both $E$ and $D$.
All tiles in $D$ which have the the 0/1 glue (or its complement) are required to be placed by cooperation with a non 0/1 glue.
Additionally, each tile in $D$ has at most one face which contains strength 1 0/1 glue. 
Since no other glues are shared between $E$ and $D$ it is not possible for strength 2 binding to occur between (super)tiles in $E$ and $D$ aside from the binding of $\phi$ with the initiator tiles of $D$. 
Since junk assemblies produced by any assembly sequence are also unable to negatively interact with other assemblies, a system whose input assemblies have multiple shapes will behave simply as the union of individual systems which each have one input assembly shape, creating terminal assemblies of all of (and only) the correct shapes.

The maximal junk size of $R$ is 4, driven by the junk size of $E$.
We can say that $\mathcal{R}_S$ finitely completes with respect to the set of assemblies created from the shape tiles of $D$ in the shape of each assembly in $S$, as the tileset $R$ operates such that any input shape $s_i$ is encoded into an intermediate structure $\phi_i$, $\phi_i$ is then decoded into two copies of $s^\prime_i$, 
an assembly which contains tiles in the exact same locations as $s$ (up to rotation and translation). 
As deconstruction leads to the production of a single structure $\phi_i$, and $\phi_i$ is only able to be decoded to $s^\prime_i$ two times, we can place a finite bound on the number of each tile type required to produce each terminal assembly $s^\prime$. (This largely follows from the fact that encoding systems using $E$ finitely complete with respect to the set of encoding assemblies, and that decoding systems using $D$ finitely complete with respect to the set of assemblies whose shapes are encoded.) 
Therefore, $R$ also finitely completes, with respect to the set of assemblies with the same shape as the input assemblies, and Theorem~\ref{thm:replicator} is proven.

Note that the condition that a single encoding structure $\phi_i$ leads to the production of exactly two target assemblies  $s^\prime_i$ is imposed to allow for the universal shape replicator to technically be able to replicate shapes from an arbitrarily large set of input assembly shapes without the potential to `starve' the encodings of one shape so that they never produce decoded copies (and thus the replicator would not finitely complete with respect to the full set of terminal assembly shapes). If only one input assembly shape was provided as input, it would instead be possible to just remove the dissolve signals from the encoding structure and allow each to initiate the production of an unbounded number of decoded copies. It would also be trivial to add tiles that make copies of the encoded structures that can each initiate the decoding process, leading to exponential replication.

\section{Universal Shape Encoding, Decoding, and Replication in the STAM}\label{sec:STAM}

As previously mentioned, our use of the STAM$^R$ instead of the standard STAM for the previous results was intended to allow for the input assemblies to be more generic. That is, a single uniform glue can cover their entire surfaces rather than having glues that are direction specific, which is implicitly the case with glues in the STAM (as well as the aTAM and 2HAM, as commonly defined) since tiles are not allowed to rotate in those models and therefore glues with complementary labels but in non-opposite directions can't bind. Giving tiles the ability to rotate, meaning that glues are not specific to directions, made aspects of the shape encoding problem more difficult to solve, especially the ``leader election'' process to select a corner of the bounding box to be the location of the origin. Nonetheless, the constructions can be easily modified to work in the STAM. To do this we can simply define rotated versions of each of our tiles, one for each of the 24 possible rotations. The behavior of these tiles will be identical to the behavior of the tiles in the STAM$^R$ which can easily be seen by forming the trivial bijection between individual tiles in the STAM tileset and rotated instances of those tiles in the STAM$^R$ tileset. This induces a bijection between assemblies formed by the tiles in both, and this bijection clearly preserves the dynamics of the system as any binding of assemblies possible in one corresponds to a binding of the corresponding assemblies in the other. Thus we have an isomorphism between our systems defined on these tilesets with the same input shape assemblies. Additionally, the leader election process is essentially unnecessary in the STAM version with rotated tiles since we could just choose say the top, northeastern most tile of the bounding box assembly as leader once the filler verification has finished. In principle, despite the STAM tileset requiring many rotated copies of the tiles necessary for the bounding box construction, we wouldn't need rotated copies of any other tiles if the same corner was always elected leader. 

Also, it can be argued that the STAM$^R$ is in a sense more physically realizable than the STAM if only for the fact that the STAM requires glues to implicitly encode their orientations. When implementing tiles physically using DNA, where glues are often made of single stranded DNA exposed on the sides of some more rigid DNA structure, several copies of each glue (often one for each of the 6 directions) are needed. Because there are only so many fixed length sequences of nucleotides, requiring that several sequences correspond to the same glue is expensive. This is not only because those sequences can no longer be used for different glues, but also because several similar sequences become unusable as glue sequences must be sufficiently orthogonal to mitigate erroneous binding. Consequently, our choice of a non-standard model of tile assembly does not weaken our results, but rather strengthens them both theoretically and, to some extent, practically.
\section{Beyond Shape Replication}\label{sec:extensions}

The constructions used to prove Theorem \ref{thm:replicator} were intentionally broken into separate, modular constructions proving Lemmas \ref{lem:encoder} and \ref{lem:decoder} and thus providing a universal shape encoder and a universal shape decoder. This is not only useful for proving their correctness, but also for allowing for computational transformations to be performed on the encodings of input shapes in order to instead produce output shapes based on those transformations. Like even the much simpler aTAM, the STAM (and STAM$^R$) are Turing universal, meaning any arbitrary computer program can be executed by systems in these models. Thus, given any program that can perform a computational transformation of the points of a shape and output points of another shape, tiles that execute that program (for instance, by simulating an arbitrary Turing machine in standard ways, e.g. \cite{jSADS,jCCSA}) can receive as input the binary encodings of arbitrary shapes (after their creation by the universal encoder), transform them in any algorithmic manner, and then assemblies of the shapes output by those transformations can be produced (using the universal shape decoder).

\begin{figure}[ht]
    \centering
    \begin{subfigure}{0.35\textwidth}
        \centering
        \includegraphics[width=1.0\linewidth]{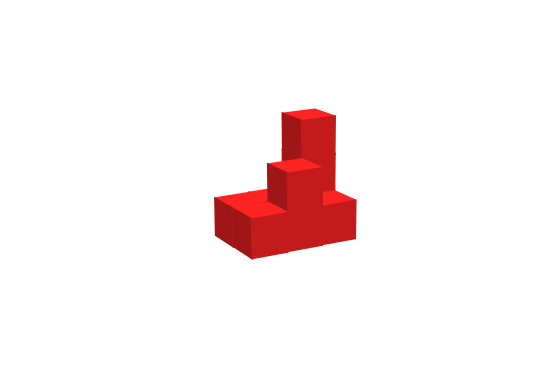}
    \caption{\label{fig:extensions-example}}
    \end{subfigure}
    \begin{subfigure}{0.25\textwidth}
        \centering
        \includegraphics[width=1.0\linewidth]{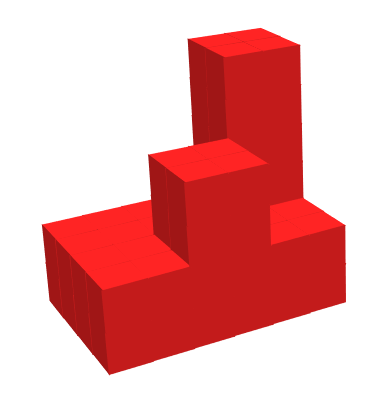}
    \caption{\label{fig:extension-example-scaled}}
    \end{subfigure}
    \begin{subfigure}{0.35\textwidth}
        \centering
        \includegraphics[width=0.95\linewidth]{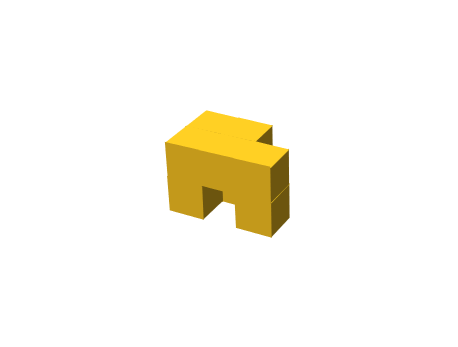}
    \caption{\label{fig:extension-example-complement}}
    \end{subfigure}

    \caption{(a) An example shape, (b) The same shape at scale factor $2$, (c) A shape which is complementary to the top surface of the shape in (a).\label{fig:extension-examples}}
\end{figure}

Due to space constraints, we don't go into great detail about the opportunities that such constructions provide. Instead, we mention just a few of the possibilities (and depict some in Figure \ref{fig:extension-examples}) while noting that the possibilities are technically infinite:

\begin{enumerate}
    \item Scaled shapes: a system could be designed to produce assemblies that have the shapes of input assemblies scaled by either a built-in constant factor (including negative, to shrink the shapes), or instead with another type of input assembly that specifies the scaling factor, allowing for a ``universal scaler''.
    
    \item Inverse shapes: a system could be designed to produce assemblies that have the inverse, i.e. complementary, shapes of the input assemblies (assuming the complements are connected, and restricting to some bounding box size since the complement of any finite shape is infinite).
    
    \item Pattern matching: a system could be designed to inspect input assembly shapes for specific patterns and to either produce assemblies that signal the presence of a target pattern, or instead assemblies that are complementary to, and can bind to, the surfaces of assemblies containing those patterns. 
\end{enumerate}

Although such constructions are highly theoretical and quite complex, and thus unlikely in their current forms to be practically implementable, they provide a mathematical foundation for the construction of complex, dynamic systems that mimic biological systems. One possible example is an ``artificial immune system'' capable of inspecting surfaces, detecting those which match (or fail to match) specific patterns, and creating assemblies capable of binding to those deemed to be foreign, harmful, or otherwise targeted. As mentioned, there are infinite possibilities.

\section{Impossibility of Shape Replication Without Deconstruction}\label{sec:imposs}

In this section, we prove that in order for a system in the STAM$^R$ to encode and/or replicate shapes which have enclosed or bent cavities (see Definitions \ref{def:enclosed-cavity} and \ref{def:bent-cavity}), the input assemblies must have the potential for tiles to be removed. To do so, we first utilize a theorem from \cite{SelfReplicationArxiv}.

\newcounter{tempcounter2}
\setcounter{tempcounter2}{\value{theorem}}

\setcounter{theorem}{3}

\begin{theorem}[from \cite{SelfReplicationArxiv}]\label{thm:need-to-deconstruct}
    Let $U$ be an STAM* tileset such that for an arbitrary 3D shape $S$, the STAM* system $\mathcal{T} = (U,\sigma_S,\tau)$ with $\dom \sigma_S = S$, $\mathcal{T}$ is a shape self-replicator for $S$ and $\sigma_S$ is non-porous. Then, for any $r \in \mathbb{N}$, there exists a shape $S$ such that $\mathcal{T}$ must remove at least $r$ tiles from the seed assembly $\sigma_S$.
\end{theorem}

\setcounter{theorem}{\value{tempcounter2}}
    
Theorem 4 from \cite{SelfReplicationArxiv} applies to the STAM*. However, the STAM$^R$ is simply a restricted version of the STAM* which only allows tiles to be a single shape, that of a unit cube, and which does not allow flexible glues. Since all assemblies in the STAM$^R$ are non-porous (i.e. free tiles cannot pass through the tiles of an assembly or the gaps between bound tiles) and the STAM$^R$ has more restrictive dynamics than the STAM*, the proof of this impossibility result, which shows the impossibility of self-replicating assemblies with enclosed cavities without removing tiles, suffices to prove the following corollary (stated using the terminology of this paper) as well.\footnote{The proof can be found in \cite{SelfReplicationArxiv}, and we omit duplicating it here due to space constraints.} Note that this proof holds even if the input assemblies are not uniformly covered.

\begin{corollary}\label{cor:enclosed}
There exist neither a universal shape encoder nor a universal shape replicator in the STAM$^R$ for the class of shapes with enclosed cavities whose assemblies are not deconstructable.
\end{corollary}



\begin{figure}[ht]
    \centering
    \begin{subfigure}{0.20\textwidth}
        \centering
        \includegraphics[width=1.0\linewidth]{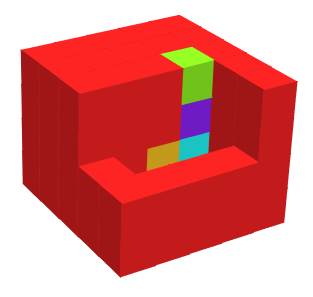}
    \caption{\label{fig:impossible-well-left}}
    \end{subfigure}
    \begin{subfigure}{0.20\textwidth}
        \centering
        \includegraphics[width=1.0\linewidth]{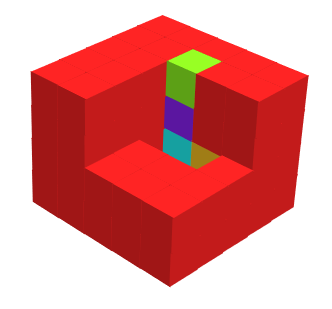}
    \caption{\label{fig:impossible-well-right}}
    \end{subfigure}
    \begin{subfigure}{0.26\textwidth}
        \centering
        \includegraphics[width=1.0\linewidth]{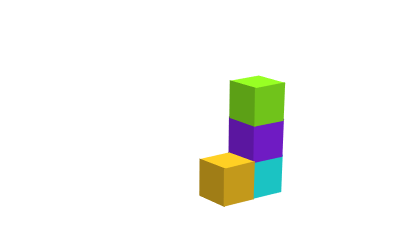}
    \caption{\label{fig:well-assembly-left}}
    \end{subfigure}
    \begin{subfigure}{0.24\textwidth}
        \centering
        \includegraphics[width=0.95\linewidth]{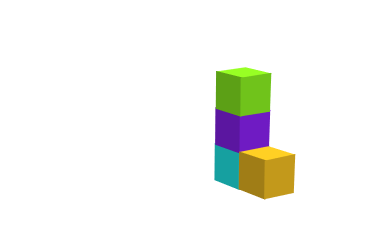}
    \caption{\label{fig:well-assembly-right}}
    \end{subfigure}

    \caption{(a) and (b) Partial depictions of a pair of shapes which cannot be correctly encoded/replicated without a deconstructable input assembly. Each consists of a $5 \times 5 \times 4$ cube with a 4-cube-long bent cavity. For each, the green, purple, blue, and yellow locations indicate the empty locations that make the bent cavity. The rest of the $5 \times 5 \times 4$ cube locations would be filled in with red cubes (some have been omitted to make the cavity locations visible). (c) and (d) The shapes of assemblies that could grow into the bent cavities.\label{fig:bent-cavities}}
    
\end{figure}

    

    

Our next theorem deals with shapes having bent cavities.

\begin{theorem}\label{thm:imposs}
There exist neither a universal shape encoder nor a universal shape replicator in the STAM$^R$ for the class of shapes with bent cavities whose input assemblies are uniformly covered but are not deconstructable.
\end{theorem}

We prove Theorem \ref{thm:imposs} by contradiction. Therefore, let $f_e$ be a shape encoding function and assume $E$ is a universal shape encoder with respect to $f_e$, and let $c$ be the constant value which bounds the size of the junk assemblies. (Nearly identical arguments will hold for a universal shape replicator.) Define the shapes $s_1$ and $s_2$ as shown in Figures \ref{fig:impossible-well-left} and \ref{fig:impossible-well-right}, i.e. each is a $5 \times 5 \times 4$ cube with a bent cavity that goes into the cube to a depth of 3, then turns one of two directions for each. Note importantly that the well is offset from the center of the cube such that $s_1$ and $s_2$ are not rotationally equivalent. Since $E$ is assumed to be a universal shape encoder, there must exist two STAM$^R$ systems $\mathcal{E}_1 = (E,\sigma_1,\tau)$ and $\mathcal{E}_2 = (E,\sigma_2,\tau)$, where $\sigma_1$ consists of infinite copies of tiles from $E$ and infinite copies of uniformly covered assemblies in the shape of $s_1$, and $\sigma_2$ consists of infinite copies of tiles from $E$ and infinite copies of uniformly covered assemblies in the shape of $s_2$.

$\mathcal{E}_1$ must produce terminal assemblies which encode shape $s_1$ but must not produce terminal assemblies which encode shape $s_2$, since no assembly of shape $s_2$ is included in its input assemblies. Similarly, $\mathcal{E}_2$ must produce terminal assemblies which encode shape $s_2$ but not $s_1$. Let $\vec{\alpha}$ be an assembly sequence in $\mathcal{E}_1$ which results in a terminal assembly encoding shape $s_1$. We now show that every action of $\vec{\alpha}$ must be valid, in the same ordering, in $\mathcal{E}_2$ but using an input assembly of shape $s_2$. This is because the exact same glues will be exposed by the input assemblies of shapes $s_1$ and $s_2$ in the same relative locations with the slight difference of relative rotations of the innermost locations of the bent cavities of each from the adjacent cavity locations. Assuming that, in $\vec{\alpha}$, tiles attach into all locations of the bent cavity (if only the location shown in yellow remains empty the same argument will hold, and if both the locations shown in yellow and blue remain empty then there is absolutely no difference in any aspect of the assembly sequence in $\mathcal{E}_2$ and the argument immediately holds), this results only in the relative orientations of at most the bottom two tiles being turned 90 degrees relative to the tile immediately above them (i.e. the tile in the purple location in Figure \ref{fig:bent-cavities}). 
Since tiles in the STAM$^R$ are rotatable, with no distinction for directions, there is no mechanism for tiles in the purple locations of assemblies shown in Figures \ref{fig:well-assembly-left} and \ref{fig:well-assembly-right} from distinguishing from each other (via tile types, glues, or signals). Tiles of the same types which bind into those locations in $\vec{\alpha}$ must also be able to do so in the assembly sequence of $\mathcal{E}_2$ using the exact same glues and firing the exact same signals (if any). 
Thus $\vec{\alpha}$ must be a valid assembly sequence in $\mathcal{E}_2$ as well. This means that an assembly encoding the shape of $s_1$ is also created as a terminal assembly in $\mathcal{E}_2$. Note that if the constant $c$ is greater than the size of the shapes $s_1$ and $s_2$ (i.e. $5*5*4 - 4 = 96$), then we can simply increase their dimensions until they are larger than $c$ (but still contain the same bent cavities) and the argument still holds and the incorrectly produced assemblies cannot be considered ``junk'' assemblies. This is a contradiction that $E$ is a universal shape encoder with respect to $f_e$ and constant $c$. Since no assumptions were made about $E$ other than it being a universal shape encoder, no such $E$ can exist. By slightly altering the argument for a universal shape replicator $R$ (instead of universal encoder $E$) and generating terminal assemblies of shapes $s_1$ and $s_2$ (rather than assemblies encoding those shapes), the same argument holds to show that no universal shape replicator exists, and thus Theorem \ref{thm:imposs} is proven.

\bibliography{tam,experimental_refs}

\begin{thebibliography}{10}

\bibitem{RNaseSODA2010}
Zachary Abel, Nadia Benbernou, Mirela Damian, Erik~D. Demaine, Martin~L.
  Demaine, Robin Flatland, Scott~D. Kominers, and Robert~T. Schweller.
\newblock Shape replication through self-assembly and {R}{N}{A}se enzymes.
\newblock In {\em SODA 2010: Proceedings of the Twenty-first Annual ACM-SIAM
  Symposium on Discrete Algorithms}, pages 1045--1064, Austin, Texas, 2010.
  Society for Industrial and Applied Mathematics.

\bibitem{SelfReplicationArxiv}
Andrew Alseth, Daniel Hader, and Matthew~J. Patitz.
\newblock Self-replication via tile self-assembly.
\newblock Technical Report 2105.02914, Computing Research Repository, 2021.
\newblock URL: \url{{https://arxiv.org/abs/2105.02914}}.

\bibitem{SelfReplicationDNA}
Andrew Alseth, Daniel Hader, and Matthew~J. Patitz.
\newblock {Self-Replication via Tile Self-Assembly (Extended Abstract)}.
\newblock In Matthew~R. Lakin and Petr \v{S}ulc, editors, {\em 27th
  International Conference on DNA Computing and Molecular Programming (DNA
  27)}, volume 205 of {\em Leibniz International Proceedings in Informatics
  (LIPIcs)}, pages 3:1--3:22, Dagstuhl, Germany, 2021. Schloss Dagstuhl --
  Leibniz-Zentrum f{\"u}r Informatik.
\newblock URL: \url{https://drops.dagstuhl.de/opus/volltexte/2021/14670}, \href
  {https://doi.org/10.4230/LIPIcs.DNA.27.3}
  {\path{doi:10.4230/LIPIcs.DNA.27.3}}.

\bibitem{Versus}
Sarah Cannon, Erik~D. Demaine, Martin~L. Demaine, Sarah Eisenstat, Matthew~J.
  Patitz, Robert~T. Schweller, Scott~M. Summers, and Andrew Winslow.
\newblock Two hands are better than one (up to constant factors): Self-assembly
  in the 2{HAM} vs. a{TAM}.
\newblock In Natacha Portier and Thomas Wilke, editors, {\em STACS}, volume~20
  of {\em LIPIcs}, pages 172--184. Schloss Dagstuhl - Leibniz-Zentrum fuer
  Informatik, 2013.

\bibitem{chalkUniversalShapeReplicators2017}
Cameron Chalk, Erik~D. Demaine, Martin~L. Demaine, Eric Martinez, Robert
  Schweller, Luis Vega, and Tim Wylie.
\newblock Universal shape replicators via {{Self-Assembly}} with {{Attractive}}
  and {{Repulsive Forces}}.
\newblock In {\em Proceedings of the {{Twenty-Eighth Annual ACM-SIAM
  Symposium}} on {{Discrete Algorithms}}}, pages 225--238. {Society for
  Industrial and Applied Mathematics}, January 2017.
\newblock \href {https://doi.org/10.1137/1.9781611974782.15}
  {\path{doi:10.1137/1.9781611974782.15}}.

\bibitem{AGKS05g}
Qi~Cheng, Gagan Aggarwal, Michael~H. Goldwasser, Ming-Yang Kao, Robert~T.
  Schweller, and Pablo~Moisset de~Espan\'{e}s.
\newblock Complexities for generalized models of self-assembly.
\newblock {\em SIAM Journal on Computing}, 34:1493--1515, 2005.

\bibitem{DDFIRSS07}
Erik~D. Demaine, Martin~L. Demaine, S{\'a}ndor~P. Fekete, Mashhood Ishaque,
  Eynat Rafalin, Robert~T. Schweller, and Diane~L. Souvaine.
\newblock Staged self-assembly: nanomanufacture of arbitrary shapes with
  ${O}(1)$ glues.
\newblock {\em Natural Computing}, 7(3):347--370, 2008.

\bibitem{2HAMIU}
Erik~D. Demaine, Matthew~J. Patitz, Trent~A. Rogers, Robert~T. Schweller,
  Scott~M. Summers, and Damien Woods.
\newblock The two-handed assembly model is not intrinsically universal.
\newblock In {\em 40th International Colloquium on Automata, Languages and
  Programming, ICALP 2013, Riga, Latvia, July 8-12, 2013}, Lecture Notes in
  Computer Science. Springer, 2013.

\bibitem{RNAPods}
Erik~D. Demaine, Matthew~J. Patitz, Robert~T. Schweller, and Scott~M. Summers.
\newblock {Self-Assembly of Arbitrary Shapes Using {R}{N}{A}se Enzymes: Meeting
  the Kolmogorov Bound with Small Scale Factor (extended abstract)}.
\newblock In Thomas Schwentick and Christoph D{\"u}rr, editors, {\em 28th
  International Symposium on Theoretical Aspects of Computer Science (STACS
  2011)}, volume~9 of {\em Leibniz International Proceedings in Informatics
  (LIPIcs)}, pages 201--212, Dagstuhl, Germany, 2011. Schloss
  Dagstuhl--Leibniz-Zentrum fuer Informatik.
\newblock URL: \url{http://drops.dagstuhl.de/opus/volltexte/2011/3011}, \href
  {https://doi.org/http://dx.doi.org/10.4230/LIPIcs.STACS.2011.201}
  {\path{doi:http://dx.doi.org/10.4230/LIPIcs.STACS.2011.201}}.

\bibitem{IUSA}
David Doty, Jack~H. Lutz, Matthew~J. Patitz, Robert~T. Schweller, Scott~M.
  Summers, and Damien Woods.
\newblock The tile assembly model is intrinsically universal.
\newblock In {\em Proceedings of the 53rd Annual IEEE Symposium on Foundations
  of Computer Science}, FOCS 2012, pages 302--310, 2012.

\bibitem{evans2014crystals}
Constantine~Glen Evans.
\newblock {\em Crystals that count! {P}hysical principles and experimental
  investigations of {D}{N}{A} tile self-assembly}.
\newblock PhD thesis, California Institute of Technology, 2014.

\bibitem{jSignals3D}
Tyler Fochtman, Jacob Hendricks, Jennifer~E. Padilla, Matthew~J. Patitz, and
  Trent~A. Rogers.
\newblock Signal transmission across tile assemblies: 3{D} static tiles
  simulate active self-assembly by 2{D} signal-passing tiles.
\newblock {\em Natural Computing}, 14(2):251--264, 2015.

\bibitem{STAMshapes}
Jacob Hendricks, Matthew~J. Patitz, and Trent~A. Rogers.
\newblock Replication of arbitrary hole-free shapes via self-assembly with
  signal-passing tiles.
\newblock In Cristian~S. Calude and Michael~J. Dinneen, editors, {\em
  Unconventional Computation and Natural Computation - 14th International
  Conference, {UCNC} 2015, Auckland, New Zealand, August 30 - September 3,
  2015, Proceedings}, volume 9252 of {\em Lecture Notes in Computer Science},
  pages 202--214. Springer, 2015.
\newblock URL: \url{http://dx.doi.org/10.1007/978-3-319-21819-9\_15}, \href
  {https://doi.org/10.1007/978-3-319-21819-9\_15}
  {\path{doi:10.1007/978-3-319-21819-9\_15}}.

\bibitem{DirectedNotIU}
Jacob Hendricks, Matthew~J. Patitz, and Trent~A. Rogers.
\newblock Universal simulation of directed systems in the abstract tile
  assembly model requires undirectedness.
\newblock In {\em Proceedings of the 57th Annual IEEE Symposium on Foundations
  of Computer Science (FOCS 2016), New Brunswick, New Jersey, USA {\rm October
  9-11, 2016}}, pages 800--809, 2016.

\bibitem{JonoskaSignals1}
Nata\v{s}a Jonoska and Daria Karpenko.
\newblock Active tile self-assembly, {P}art 1: Universality at temperature 1.
\newblock {\em International Journal of Foundations of Computer Science},
  25(02):141--163, 2014.
\newblock \href {https://doi.org/10.1142/S0129054114500087}
  {\path{doi:10.1142/S0129054114500087}}.

\bibitem{ke2012three}
Yonggang Ke, Luvena~L Ong, William~M Shih, and Peng Yin.
\newblock Three-dimensional structures self-assembled from {DNA} bricks.
\newblock {\em Science}, 338(6111):1177--1183, 2012.

\bibitem{SignalsReplication}
Alexandra Keenan, Robert~T. Schweller, and Xingsi Zhong.
\newblock Exponential replication of patterns in the signal tile assembly
  model.
\newblock In David Soloveichik and Bernard Yurke, editors, {\em DNA}, volume
  8141 of {\em Lecture Notes in Computer Science}, pages 118--132. Springer,
  2013.

\bibitem{jCCSA}
James~I. Lathrop, Jack~H. Lutz, Matthew~J. Patitz, and Scott~M. Summers.
\newblock Computability and complexity in self-assembly.
\newblock {\em Theory Comput. Syst.}, 48(3):617--647, 2011.

\bibitem{jSSADST}
James~I. Lathrop, Jack~H. Lutz, and Scott~M. Summers.
\newblock Strict self-assembly of discrete {S}ierpinski triangles.
\newblock {\em Theoretical Computer Science}, 410:384--405, 2009.

\bibitem{jNegativeGluesShapes}
Austin Luchsinger, Robert Schweller, and Tim Wylie.
\newblock Self-assembly of shapes at constant scale using repulsive forces.
\newblock {\em Natural Computing}, Aug 2018.
\newblock \href {https://doi.org/10.1007/s11047-018-9707-9}
  {\path{doi:10.1007/s11047-018-9707-9}}.

\bibitem{NegativeGluesShapes}
Austin Luchsinger, Robert~T. Schweller, and Tim Wylie.
\newblock Self-assembly of shapes at constant scale using repulsive forces.
\newblock In {\em {UCNC}}, volume 10240 of {\em Lecture Notes in Computer
  Science}, pages 82--97. Springer, 2017.

\bibitem{Temp1PathsSTOC}
Pierre{-}{\'{E}}tienne Meunier, Damien Regnault, and Damien Woods.
\newblock The program-size complexity of self-assembled paths.
\newblock In Konstantin Makarychev, Yury Makarychev, Madhur Tulsiani, Gautam
  Kamath, and Julia Chuzhoy, editors, {\em Proccedings of the 52nd Annual {ACM}
  {SIGACT} Symposium on Theory of Computing, {STOC} 2020, Chicago, IL, USA,
  June 22-26, 2020}, pages 727--737. {ACM}, 2020.
\newblock \href {https://doi.org/10.1145/3357713.3384263}
  {\path{doi:10.1145/3357713.3384263}}.

\bibitem{jSignals}
Jennifer~E. Padilla, Matthew~J. Patitz, Robert~T. Schweller, Nadrian~C. Seeman,
  Scott~M. Summers, and Xingsi Zhong.
\newblock Asynchronous signal passing for tile self-assembly: Fuel efficient
  computation and efficient assembly of shapes.
\newblock {\em International Journal of Foundations of Computer Science},
  25(4):459--488, 2014.

\bibitem{SingleNegative}
Matthew~J. Patitz, Robert~T. Schweller, and Scott~M. Summers.
\newblock Exact shapes and {T}uring universality at temperature 1 with a single
  negative glue.
\newblock In Luca Cardelli and William~M. Shih, editors, {\em {DNA} Computing
  and Molecular Programming - 17th International Conference, {DNA} 17,
  Pasadena, CA, USA, September 19-23, 2011. Proceedings}, volume 6937 of {\em
  Lecture Notes in Computer Science}, pages 175--189. Springer, 2011.

\bibitem{jSADS}
Matthew~J. Patitz and Scott~M. Summers.
\newblock Self-assembly of decidable sets.
\newblock {\em Natural Computing}, 10(2):853--877, 2011.

\bibitem{ShapeIdentAlgo}
Matthew~J. Patitz and Scott~M. Summers.
\newblock Identifying shapes using self-assembly.
\newblock {\em Algorithmica}, 64(3):481--510, 2012.

\bibitem{RotWin00}
Paul W.~K. Rothemund and Erik Winfree.
\newblock The program-size complexity of self-assembled squares (extended
  abstract).
\newblock In {\em STOC '00: Proceedings of the thirty-second annual ACM
  Symposium on Theory of Computing}, pages 459--468, Portland, Oregon, United
  States, 2000. ACM.

\bibitem{SchulYurWinfEvolution}
Rebecca Schulman, Bernard Yurke, and Erik Winfree.
\newblock Robust self-replication of combinatorial information via crystal
  growth and scission.
\newblock {\em Proceedings of the National Academy of Sciences},
  109(17):6405--10, 2012.
\newblock URL:
  \url{http://www.biomedsearch.com/nih/Robust-self-replication-combinatorial-information/22493232.html}.

\bibitem{SolWin07}
David Soloveichik and Erik Winfree.
\newblock Complexity of self-assembled shapes.
\newblock {\em SIAM Journal on Computing}, 36(6):1544--1569, 2007.

\bibitem{SummersTemp}
Scott~M. Summers.
\newblock Reducing tile complexity for the self-assembly of scaled shapes
  through temperature programming.
\newblock {\em Algorithmica}, 63(1-2):117--136, June 2012.
\newblock URL: \url{http://dx.doi.org/10.1007/s00453-011-9522-5}, \href
  {https://doi.org/10.1007/s00453-011-9522-5}
  {\path{doi:10.1007/s00453-011-9522-5}}.

\bibitem{Winf98}
Erik Winfree.
\newblock {\em Algorithmic Self-Assembly of {D}{N}{A}}.
\newblock PhD thesis, California Institute of Technology, June 1998.

\bibitem{drmaurdsa}
Damien Woods, David Doty, Cameron Myhrvold, Joy Hui, Felix Zhou, Peng Yin, and
  Erik Winfree.
\newblock Diverse and robust molecular algorithms using reprogrammable
  {D}{N}{A} self-assembly.
\newblock {\em Nature}, 567:366--372, 2019.

\end{thebibliography}

\end{document}